\documentclass[12pt,oneside]{amsart}
\usepackage[latin9]{inputenc}
\usepackage[a4paper]{geometry}
\usepackage{fancyhdr}
\usepackage{color}
\usepackage{enumitem}
\usepackage{amstext}
\usepackage{amsthm}
\usepackage{amssymb}
\usepackage{stackrel}
\usepackage{graphicx}

\makeatletter
\numberwithin{equation}{section}
\theoremstyle{plain}
\newtheorem{thm}{\protect\theoremname}[section]
  \theoremstyle{plain}
  \newtheorem{lem}[thm]{\protect\lemmaname}
  \theoremstyle{plain}
  \newtheorem{cor}[thm]{\protect\corollaryname}
  \theoremstyle{remark}
  \newtheorem{rem}[thm]{\protect\remarkname}

\@ifundefined{date}{}{\date{}}
\usepackage{amsthm}
\usepackage{longtable}
\usepackage{wasysym}

\DeclareSymbolFont{extraup}{U}{zavm}{m}{n}
\DeclareMathSymbol{\varheart}{\mathalpha}{extraup}{86}
\DeclareMathSymbol{\vardiamond}{\mathalpha}{extraup}{87} 
\renewcommand{\textendash}{--}

\makeatother

  \providecommand{\corollaryname}{Corollary}
  \providecommand{\lemmaname}{Lemma}
  \providecommand{\remarkname}{Remark}
\providecommand{\theoremname}{Theorem}

\begin{document}
\global\long\def\C{\mathbb{C}}
\global\long\def\Cd{\C^{\delta}}
\global\long\def\Cprim{\C^{\delta,\circ}}
\global\long\def\Cdual{\C^{\delta,\bullet}}
\global\long\def\Od{\Omega^{\delta}}
\global\long\def\Oprim{\Omega^{\delta,\circ}}
\global\long\def\Odual{\Omega^{\delta,\bullet}}

\global\long\def\P{\mathsf{\mathbb{P}}}
 \global\long\def\E{\mathsf{\mathbb{E}}}
 \global\long\def\sF{\mathcal{F}}
 \global\long\def\ind{\mathbb{I}}

\global\long\def\R{\mathbb{R}}
 \global\long\def\Z{\mathbb{Z}}
 \global\long\def\N{\mathbb{N}}
 \global\long\def\Q{\mathbb{Q}}

\global\long\def\C{\mathbb{C}}
 \global\long\def\Rsphere{\overline{\C}}
 \global\long\def\re{\Re\mathfrak{e}\,}
 \global\long\def\im{\Im\mathfrak{m}\,}
 \global\long\def\arg{\mathrm{arg}}
 \global\long\def\i{\mathfrak{i}}
\global\long\def\eps{\varepsilon}
\global\long\def\lamb{\lambda}
\global\long\def\lambb{\bar{\lambda}}

\global\long\def\D{\mathbb{D}}
 \global\long\def\H{\mathbb{H}}
\global\long\def\F{\mathcal{F}}
\global\long\def\outside{\mathcal{\text{out}}}
\global\long\def\winding{\mathrm{wind}}
\global\long\def\sing{\mathcal{S}}

\global\long\def\dist{\mathrm{dist}}
 \global\long\def\reg{\mathrm{reg}}
 \global\long\def\pa{\partial}
 \global\long\def\Leb{\mathrm{Leb}}

\global\long\def\half{\frac{1}{2}}
 \global\long\def\sgn{\mathrm{sgn}}
\global\long\def\Conf{\mathrm{Conf}}
\global\long\def\ZLT{\mathrm{Z_{LT}}}
\global\long\def\Wind{\mathrm{winding}}
\global\long\def\Edges{\mathrm{\mathcal{E}}}
\global\long\def\Dual#1{#1^{\bullet}}

\global\long\def\bdry{\partial}
 \global\long\def\cl#1{\overline{#1}}
\global\long\def\ZFK{Z_{\text{FK}}}
\global\long\def\Clusters{\mathcal{C}}
\global\long\def\Odp{\hat{\Omega}^{\delta}}
\global\long\def\Edges{\mathcal{E}}
\global\long\def\Vertices{\mathcal{V}}
\global\long\def\bcond{\mathcal{\beta}}
\global\long\def\free{\mathcal{\mathrm{free}}}
\global\long\def\wired{\mathrm{wired}}
\global\long\def\Medial{\mathcal{M}}
\global\long\def\diam{\mathrm{diam}\,}
\global\long\def\Chol{C_{\mathrm{H\ddot{o}l}}}

\global\long\def\Edgesprim{\hat{\Edges}}
\global\long\def\Conf{\mathrm{Conf}}
\global\long\def\Pp{P}
\global\long\def\Tr{\mathrm{Tr\,}}
\global\long\def\pp{p}
\global\long\def\dimphi{k}
\global\long\def\SingularSet{\Omega^{\dagger}}
\global\long\def\HomType{\mathrm{\mathbf{a}}}
\global\long\def\ddt{\frac{dt}{t}}
\global\long\def\supp{\mathrm{\mathrm{supp}\,}}
\global\long\def\osc#1{\mathrm{\mathrm{osc}}_{#1}}
\global\long\def\Det{\mathrm{det}^{\star}}

\global\long\def\HK#1#2{P^{#1}\left(#2\right)}
\global\long\def\ddt{\frac{dt}{t}}
\global\long\def\Od{\Omega^{\delta}}
 \global\long\def\loc{\mathrm{loc}}
\global\long\def\Cone#1{\mathcal{C}^{#1}}
\global\long\def\H{\mathbb{H}}
\global\long\def\C{\mathbb{C}}
\global\long\def\OO#1{O\left(#1\right)}
\global\long\def\ConstC{C}
\global\long\def\ConstD{D}
\global\long\def\HalfCorn{Y}
\global\long\def\Left{\text{left}}
\global\long\def\Right{\text{right}}
\global\long\def\BulkCorn{B}
\global\long\def\AnyOne{\star}
\global\long\def\AnyTwo{\Diamond}
\global\long\def\Crnrs{\text{Corners}}
\global\long\def\HKLog#1#2#3{\tilde{P}(#1,#2,#3)}
\global\long\def\Dir{\mathcal{D}}
\global\long\def\Neu{\mathcal{N}}
\global\long\def\CornerD{\Upsilon_{\Dir}}
\global\long\def\CornerN{\Upsilon_{\Neu}}
\global\long\def\CornerDN{\Upsilon_{\Dir\Neu}}
\global\long\def\Corner{\Upsilon}
\global\long\def\Id{\mathrm{Id}}
\global\long\def\Puncture{\mathcal{P}}
\global\long\def\EulerGamma{\gamma_{\mathrm{Euler}}}
\global\long\def\En{\mathcal{E}}
\global\long\def\T{\mathbb{T}}
\global\long\def\rwZ{\gamma^{\mathbb{Z}}}

\global\long\def\I{I_{0}}
 \global\long\def\P{\mathbb{P}}

\global\long\def\LatticeA{+}
 \global\long\def\LatticeB{\square}
 \global\long\def\LatticeC{\Diamond}
 \global\long\def\LatticeE{\triangle}
\global\long\def\LatticeE{\triangle}
\global\long\def\LatticeD{\triangleright}

\fancyhead[R]{}

\fancyhead[L]{}

\thispagestyle{fancy}

\fancyhead[C]{Determinants of discrete Laplacians on triangulations and quadrangulations}

\email{konstantin.izyurov@helsinki.fi, mikhail.khristoforov@helsinki.fi}

\title{Asymptotics of the determinant of discrete Laplacians on triangulated
and quadrangulated surfaces}

\author{K. Izyurov and M. Khristoforov}

\date{May 27th, 2022}
\begin{abstract}
Consider a surface $\Omega$ with a boundary obtained by gluing together
a finite number of equilateral triangles, or squares, along their
boundaries, equipped with a vector bundle with a flat unitary connection.
Let $\Omega^{\delta}$ be a discretization of this surface, in which
each triangle or square is discretized by a bi-periodic lattice of
mesh size $\delta$, possessing enough symmetries so that these discretizations
can be glued together seamlessly. We show that the logarithm of the
product of non-zero eigenvalues of the discrete Laplacian acting on
the sections of the bundle is asymptotic to 
\[
A|\Omega^{\delta}|+B|\partial\Omega^{\delta}|+C\log\delta+D+o(1).
\]
Here $A$ and $B$ are constants that depend only on the lattice,
$C$ is an explicit constant depending on the bundle, the angles at
conical singularities and at corners of the boundary, and $D$ is
a sum of lattice-dependent contributions from singularities and a
universal term that can be interpreted as a zeta-regularization of
the determinant of the continuum Laplacian acting on the sections
of the bundle. We allow for Dirichlet or Neumann boundary conditions,
or mixtures thereof. Our proof is based on an integral formula for
the determinant in terms of theta function, and the functional Central
limit theorem.
\end{abstract}

\maketitle

\section{Introduction}

Let $\Omega$ be a connected surface, possibly with boundary, obtained
by gluing finitely many equal equilateral triangles, or squares, along
their boundaries. Thus, $\Omega$ may have conical singularities and
piece-wise straight boundary with corners; the cone and wedge angles
either all belong to $\frac{\pi}{3}k,$ or to $\frac{\pi}{2}k,$ $k\in\N$;
we refer to the former situation as a triangulation and the latter
one as a quadrangulation. The boundary of $\Omega$ will be decomposed
into two parts, $\partial_{\Dir}\Omega$ and $\partial_{\Neu}\Omega,$
that will carry Dirichlet and Neumann boundary conditions respectively;
we assume that there are finitely many points separating the two.
We assume that $\Omega$ is equipped with a finite rank vector bundle
with a unitary flat connection $\varphi$. We furthermore allow $\Omega$
to have finitely many punctures, distinct from the set of conical
singularities, and allow $\varphi$ to have monodromy around those
punctures. Our main results and techniques are new already in the
case of a trivial line bundle, and the reader interested in the simplest
situation may think of this case.

We summarize our setup briefly here, referring to Section \ref{sec: notation}
for the detailed definitions. By a \emph{lattice}, we mean an (infinite)
undirected planar graph with non-negative weights on edges embedded
bi-periodically in the plane $\C$. Given a lattice which is symmetric
under reflections $z\mapsto\overline{z}$ and under rotation by $\frac{\pi}{2}$
or $\frac{\pi}{3}$ (such as e.g., the square lattice in the case
of quadrangulations and the triangular lattice in the case of triangulations),
we can discretize $\Omega$ by this lattice scaled to have small mesh
$\delta$. We denote the discretized surface by $\Od$, see Figure
\ref{fig: surface}. In what follows, by a \emph{lattice-dependent}
constant, we mean a quantity that does not depend on $\delta$, but
may depend on the underlying lattice and weights (and, in the case
of $D_{p},$ on local geometric data); a \emph{lattice-independent},
or \emph{universal} quantity, is allowed to depend only on $\Omega,$
the boundary conditions, and the connection $\varphi.$ The parallel
transport of $\varphi$ along the edges gives a discrete connection,
and we can consider the corresponding discrete Laplacian $\Delta^{\Od,\varphi}$,
with Dirichlet/Neumann boundary conditions approximating those in
$\Omega$. We will denote the rank of $\varphi$ by $d$. The subject
of the present paper is the asymptotics of $\Det\Delta^{\Od,\varphi}$
as $\delta\to0$, where $\Det$ stands for the product of all non-zero
eigenvalues. 

We make the following conventions regarding the parameters, see Section
\ref{sec: notation} for details. The mesh size $\delta$ is defined
so that the lattice has $\delta^{-2}$ vertices per unit area. Thus,
for a given lattice, $\delta=\frac{\delta_{0}}{N}$ for a constant
$\delta_{0}$ and an integer $N.$ We denote by $|\Od|$ the size
of $\Od,$ measured in\emph{ }the number of fundamental ``plaquettes''
(triangles or squares) in their discretization; thus as $\delta\to0,$
$|\Od|$ is of order $\delta^{-2}.$ Similarly, $|\partial_{\Dir}\Od|,|\partial_{\Neu}\Od|$
denote the size of Dirichlet and Neumann parts of the boundary, measured
in the number of plaquette sides, thus $|\partial_{\Dir}\Od|,|\partial_{\Neu}\Od|$
have order $\delta^{-1}$. We will also assume that the weights on
the graph are normalized so that the random walk on $\Od$ converges
to the standard Brownian motion. 

Our main result is as follows:
\begin{thm}
\label{thm: main} As $\delta\to0$, one has the following asymptotics:

\begin{equation}
\log\Det\Delta^{\Od,\varphi}=A\cdot|\Od|+B_{\Dir}\cdot|\partial_{\Dir}\Od|+B_{\Neu}\cdot|\partial_{\Neu}\Od|+C\cdot\log\delta+D+o(1),\label{eq: main}
\end{equation}
where:
\begin{itemize}
\item $A$ and \textup{$B_{\Dir}=-B_{\Neu}$} are lattice-dependent constants
that are expressed in terms of continuous time lattice heat kernels,
see (\ref{eq: volume}), (\ref{eq: B}).
\item $C$ is a lattice-independent constant given by
\[
C=-2\dim\ker\Delta^{\Omega,\varphi}-d\cdot\sum_{p\in\Cone{}\cup\Corner\cup\Puncture}C_{p},
\]
where $\Delta^{\Omega,\varphi}$ stands for the Friedrichs extension
the Laplacian on $\Omega$ acting on sections of $\varphi$, see Section
\ref{sec: heat_kernel_continuum}; the sum is over the set of conical
singularities, corners, and punctures of $\Omega$, and the values
$C_{p}$ for a cone $\Cone{\alpha}$ of angle $\alpha$, corners $\Corner_{\Dir}^{\alpha},\Corner_{\Neu}^{\alpha},\Corner_{\Neu\Dir}$
of angle $\alpha$ with Dirichlet, Neumann, changing Neumann-to-Dirichlet
boundary conditions, and a puncture $\mathcal{P}^{M}$ with monodromy
operator $M$ are given by 
\begin{align*}
C_{\Cone{^{\alpha}}}\quad & =\frac{\alpha}{12\pi}-\frac{\pi}{3\alpha}; & C_{\Corner_{\Dir}^{\alpha}} & =C_{\Corner_{\Neu}^{\alpha}}=\frac{\alpha}{12\pi}-\frac{\pi}{12\alpha};\\
C_{\Corner_{\Neu\Dir}^{\alpha}}\quad & =\frac{\alpha}{12\pi}+\frac{\pi}{24\alpha}; & C_{\Puncture^{M}} & =\pi^{-2}\sum_{k=1}^{\infty}(1-d^{-1}\re\Tr M^{k})k^{-2}
\end{align*}
\item the constant $D$ has the form 
\[
D=d\cdot\sum_{p\in\Cone{}\cup\Corner\cup\Puncture}D_{p}+\log\Det_{\zeta}\Delta^{\Omega,\varphi},
\]
where $D_{p}$ are lattice-dependent constants (entirely determined
by the lattice and angle, boundary conditions, or the monodromy at
a puncture $p$); and $\Det_{\zeta}\Delta^{\Omega,\varphi}$ is the
zeta-regularized determinant of $\Delta^{\Omega,\varphi}$, see Section
\ref{sec: heat_kernel_continuum}.
\end{itemize}
\end{thm}

\begin{figure}[b]
\includegraphics[width=0.6\textwidth]{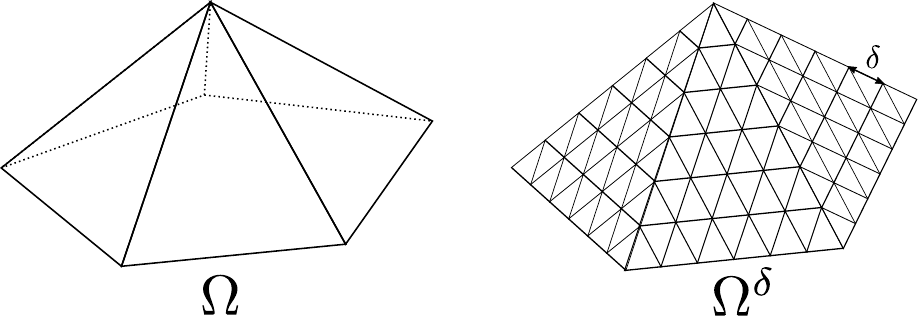}\caption{\label{fig: surface}An example of a surface $\Omega,$ glued of five
equilateral triangles, and its discretization $\protect\Od$ by triangular
lattice. In this case, $\Omega$ has one conical singularity of angle
$\frac{5\pi}{3}$ and five boundary corners of angle $\frac{2\pi}{3}.$
Note that the discrete triangles are glued so that the local graph
structure at the edges is no different from that in the bulk of the
triangles.}

\end{figure}

The constants $A,B_{\Dir},B_{\Neu}$ depend linearly on the rank $d$
of the bundle, i. e., when divided by $d,$ they only depend on the
underlying lattice and the weights.

There is a lot of previous work on the subject. Discretizations of
Laplacian were studied by Dodziuk \cite{Dodziuk} on a class of discretizations
Riemannian manifolds in arbitrary dimension. He established \cite[Section 5]{Dodziuk}
the convergence of the eigenvalues and the spectral zeta function
of a combinatorial Laplacian with weights inherited from the Riemannian
metric on the underlying manifold. Cardy and Peschel \cite{Cardy_Peschel}
conjectured that the asymptotic of partition function of any critical
2D model of a Riemann surface should take a form similar to (\ref{eq: main}).
Since the determinants of the discrete Laplacian and its vector bundle
versions are partition functions of a number of lattice models, such
as dimers and double dimers, discrete GFF, spanning trees and cycle-rooted
spanning forests \cite{Kenyon,kenyon2011spanning,kenyon2014_double_dimers,dubedat2018double,basok2018tau,kassel2017random,kassel_Levy,lupu2019topological,jan2008markov},
our results can be viewed as a rigorous proof of a particular case
of the Cardy\textendash Peschel conjecture. Duplantier and David \cite{DuplantierDavid}
computed the asymptotics of the determinant of the discrete square
lattice Laplacian on a torus and a rectangle; their results were extended
to cylinder, Möbius strip and Klein bottle by Brankov\textendash Priezzhev
and Izmailyan\textendash Oganesyan\textendash Hu \cite{brankov1993critical,izmailian2003exact}.
The approach in these papers is based on the fact that these geometries
have large groups of symmetries acting on them, and hence the discrete
Laplacian can be diagonalized explicitly. The determinant is then
an explicit product, whose asymptotics is still non-trivial, but doable
e.g. by Euler\textendash Maclaurin formula.

Kenyon \cite{Kenyon} proved the asymptotic expansion of the type
(\ref{eq: main}) in the case of the square lattice and planar simply
connected rectilinear domains, with a slightly weaker control of the
corner contributions. In this general setting, the explicit diagonalization
of the discrete Laplacian is not available; instead, Kenyon's method
is based on tracking the variation of the determinant of the Laplacian
when cutting the domain along vertical or horizontal line, using a
relation to the dimer model. Ananth Sridhar \cite{sridhar2015asymptotic}
extended this result to the case of the Laplacian with smoothly changing
inhomogeneous weights, by variation of these weights starting from
Kenyon's result. 

Recently, Finski \cite{finski2020ComptesRendus,finski2020finite,finski2020spanning}
obtained a version of Theorem~\ref{thm: main} in the case of the
square lattice quadrangulations of Riemann surfaces with Neumann boundary
conditions and cone angles restricted to integer multiples of $\pi$,
and with a slightly weaker control of the corner and cone contributions.
In particular, that work extended the aforementioned result of Kenyon
to multiply connected domains, and directly connected the constant
term $D$ to the $\zeta$-regularized determinant of the Friedrichs
extension of the Laplacian, settling two of the open problems stated
in \cite[Section 8]{Kenyon}. Like ours, Finski's method uses the
discrete spectral $\zeta$-function, but otherwise it is rather different.
He starts by proving, by Rayleigh method, the convergence of individual
eigenvalues and eigenfunctions of the discrete Laplacian \cite{finski2020finite}.
From this, together with some estimates such as uniform Weyl's law,
it follows that the discrete spectral $\zeta$-function converges
to its continuous counterpart in the region where its defining series
converges absolutely. In order to deduce the asymptotics in a neighborhood
of the origin, Finski follows Müller's proof of Ray\textendash Singer
conjecture, and regularizes the $\zeta$-function by comparing it
to the trace of a sum of ``localizations'' of the powers of the
Laplacian subordinate to a partition of unity. These localizations
live on squares and on a number of model surfaces constructed using
the infinite cones and infinite angles, and it turns out to be possible,
with some work, to compute their asymptotics. 

The zeta-regularized determinant of the Laplacian $\Det_{\zeta}\Delta^{\Omega,\varphi}$
that appears in the constant term of the expansion (\ref{eq: main})
goes back to Kronecker \cite{kronecker19294} who computed it for
the torus. It has subsequently received a lot of attention with the
introduction of analytic torsion by Ray and Singer \cite{raysinger,ray1973analytic}
and the celebrated proofs by Cheeger and Müller of its equivalence
to the R-torsion. On the physics side, its importance stems from its
role as a partition function of conformal field theories, and in particular,
from its conformal transformation properties given by the Polyakov\textendash Alvarez
formula \cite{alvarez1983theory,polyakov1981quantum,osgood1988extremals},
see also \cite{aldana2020polyakov} and the references therein for
the most recent developments. For other related recent work, see \cite{hou2020asymptotic,reshetikhinVertman,sridhar2015asymptotic,vertman2018regularized,Dubedat_Gheissari}.

The result of Theorem \ref{thm: main} improves on the state of the
art in the following ways. First, our class of surfaces is more general,
in that it allows for conical (and corner) angles any multiples of
$\frac{\pi}{3}$ or $\frac{\pi}{2},$ and also for punctures. Second,
we work simultaneously with \emph{general geometries} that do not
admit an explicit diagonalization of the Laplacian and in a \emph{universal}
setting, allowing for discretizations by arbitrary doubly periodic
lattices, possibly with weights, with enough symmetries. Third, we
allow for mixed Dirichlet and Neumann boundary conditions. Finally,
we improve the estimates of the corner and cone contribution from
$C_{p}\log\delta+o(\log\delta)$ to $C_{p}\log\delta+D_{p}+o(1).$ 

We also propose a new proof. The method is similar to that used by
Chinta\textendash Jorgenson\textendash Karlsson \cite{CJK1,CJK2}
and Friedli \cite{Friedli} who studied the square lattice Laplacians
on a torus: we use a representation for $\log\Det\Delta^{\Od,\varphi}$
as an integral transform of the theta function, i.e., the trace of
the discrete heat kernel. We then break the integral into parts that
we analyze separately, see the key formula (\ref{eq: key_formula}).
The main idea is to regularize the discrete heat kernel by subtracting
the heat kernel on one of a discretized model surface \textendash{}
the full plane, the half-plane, a punctured plane, a cone, or a corner
\textendash{} that matches the geometry of $\Od$ locally. This immediately
isolates the volume term $A\cdot|\Od|$ in (\ref{eq: main}), and
also the boundary, cone and puncture contributions, whose asymptotics
can be analyzed separately by studying heat kernels on model surfaces.
On the other hand, the regularized trace of the heat kernel receives
only Brownian scale contributions, and after rescaling converges to
its continuous counterpart, more or less, by the local Central limit
theorem. We derive this convergence from the functional CLT in the
plane and the parabolic Harnack inequality of Delmotte. To complete
the asymptotic analysis, we need to pass to the limit under the integral;
to this end, we employ a large deviation estimate for small $t$ and
a uniform spectral gap bound for large $t$. 

Thus, essentially, we only use three ingredients: the functional Central
limit theorem, parabolic regularity, and the fact that microscopically,
our lattice approximation ``looks the same at all places''. The
only reason we do not consider more general surfaces (e.g., allowing
for conical singularities with arbitrary angles, or for genuinely
curved surfaces) is that those do not admit nice discretizations with
this last property; see Remark \ref{rem: more_general}. To see the
difficulty it entails, note that the constant $A$, being essentially
the free energy per lattice site in the infinite volume limit, is
lattice-dependent; if the lattice is different in different places
of $\Od$, then $A$ will also fluctuate, and the volume term in the
asymptotics may be hard to control with meaningful precision. This
is why we believe that approximation schemes such as Dodziuk's are
too general to admit asymptotic formulae like (\ref{eq: main}). On
the positive side, apart from the volume term (\ref{eq: volume}),
the rest of our analysis, at least in the absence of the boundary,
does not use regularity of the lattice in any essential way. This
leaves hope that the method can be applied to more general surfaces
and approximations in sufficiently integrable case, e.g. on isoradial
graphs.

Note that if we only want to study the the difference $\log\Det\Delta^{\Od,\varphi_{1}}-\log\Det\Delta^{\Od,\varphi_{2}}$,
the aforementioned volume term cancels out, as do most of other terms
in the asymptotics, and and our method yields a simple proof the asymptotics
$2(\dim\ker\varphi_{1}-\dim\ker\varphi_{2})\log\delta+\log\Det_{\zeta}\Delta^{\Omega,\varphi_{1}}-\log\Det_{\zeta}\Delta^{\Omega,\varphi_{2}}$,
using only the functional CLT and the parabolic regularity. Such differences
are of interest since they often compute interesting topological observables
in the models, see e. g. \cite{dubedat2018double,Dubedat_Gheissari,kassel2017random,kassel2016covariant,jan2008markov},
and there are a number of results in this direction. Dubédat and Gheissari
\cite{Dubedat_Gheissari} proved convergence for tori in a different
way under even weaker assumptions and Kassel\textendash Kenyon \cite{kassel2017random}
proved a similar result on general Riemann surfaces, without identifying
the limit. A remark on cancellation of the singular terms in this
setting is also made by Kenyon \cite{Kenyon}, for trivial line bundles
with different simply-connected domains, and Finski \cite{finski2020spanning},
for general bundles.

Some of the ingredients of our proof could be alternatively established
by the methods \cite{finski2020finite,finski2020spanning}; e. g.,
the convergence of theta function in the Brownian scale follows easily
from the convergence of rescaled eigenvalues and the uniform Weyl
law, as does the spectral gap estimate of Lemma~\ref{prop: (Spectral-gap)}.
Arguably, our proofs via the functional CLT are more streamlined;
for example, cones and conical singularities, that seem to be a source
of technical difficulties in \cite{finski2020finite}, do not enter
this part of our proof at all, since the Brownian motion avoids them
almost surely. On the other hand, some intermediate results and techniques
in \cite{finski2020finite,finski2020spanning} are of independent
interest. 

Our method gives the explicit values of $A,B_{\Dir}=-B_{\Neu}$ and
$D_{p}$ as sums of integrals of discrete (continuous time) heat kernels
on the discretizations of the plane, half-plane, corners, or cones.
Note, however, that to determine those constant, it suffices to compute
the asymptotics (\ref{eq: main}) in some simple geometry; for example,
a torus for $A$, or a cylinder or a square (with corresponding boundary
conditions) for $B_{\Dir},B_{\Neu}.$ Thus, on the square lattice,
it follows from \cite[Eq. (4.24) and (5.14)]{DuplantierDavid} that
\[
A^{\LatticeB}=\frac{4G}{\pi}-\log2,\quad B_{\Dir}^{\LatticeB}=\frac{1}{2}\log(\sqrt{2}-1),\quad B_{\Neu}^{\LatticeA}=\frac{1}{2}\log(\sqrt{2}-1)-\frac{1}{2}\log2.
\]
where $G$ is the Catalan's constant $1-\frac{1}{3^{2}}+\frac{1}{5^{2}}-\dots$
Note that \cite{DuplantierDavid} computes $B_{\Dir}^{\LatticeB}$
and $B_{\Neu}^{\LatticeA}$ on two \emph{different} discretizations
of a square (see Figure \ref{fig:NiceLattices}), shifted with respect
to each other by $\frac{1}{2}+i\frac{\delta}{2}$, and $B_{\Dir}^{\LatticeB}$
is related to $B_{\Neu}^{\LatticeA}$ by the Uniform Spanning tree
duality. By contrast, our identity $B_{\Dir}=-B_{\Neu}$ holds for
a \emph{fixed} discretization. An expression for $A$ in terms of
polylogarithms is known, by a different method, for isoradial graphs
with critical weights, including the hexagonal and the triangular
lattices \cite[Theorem 1.1]{kenyon2002laplacian}. In Section \ref{sec: constantsB},
we compute closed-form expressions for $B_{\Dir}=-B_{\Neu}$ for triangular
lattice (rotated in two different ways), for the square lattice rotated
by $45^{\circ}$, and give an alternative computation in the case
of a non-rotated square lattice, recovering the result of \cite{DuplantierDavid}.
Note that the specific values of the constants are sensitive to conventions
such as the definition of $|\Od|$ and the normalization of weights,
hence e.g. the extra $-\log2$ in $A^{\LatticeB}$ above compared
to \cite{DuplantierDavid,Kenyon,finski2020spanning}.

In our proof, we employ the functional Central limit theorem with
an error bound that we derive from Einmahl's multidimensional KMT
coupling \cite{Einmahl}. The only place the error bound is used in
earnest is when we sharpen the asymptotics of the contribution of
cones and corners to (\ref{eq: main}), from $C_{p}\log\delta+o(\log\delta)$
in \cite{Kenyon,finski2020spanning} to $C_{p}\log\delta+D_{p}+o(1).$
(We also use the error bound it in the computation of the boundary
contribution; however, there we do not need the \emph{functional}
CLT, so e. g. Berry-Esseen bounds would suffice.) Since any power
law error bound suffices for that, one could use e. g. the Skorokhod
embedding instead, cf. \cite[Theorem 3.4.2]{LawlerLimic}; see also
\cite[Section 7.2]{LawlerLimic} for simple proofs of the KMT coupling
on the square and the triangular lattices. The same error bound allows
one to track the rate of convergence in Lemma \ref{lem: CLT_convergence},
Corollary \ref{cor: third_term}, and (\ref{eq: cone_contr}), (\ref{eq: bdry_contr}),
improving the $o(1)$ in (\ref{eq: main}) to $O(\delta^{\rho})$
with $\rho>0$. The recent independent work \cite{greenblatt2021discrete}
gives explicitly the error bound on the square lattice; it also shows
that similar methods allow one to treat polygonal boundaries with
any rational slopes. 

The assumption that the weights are symmetric is only used in the
proofs of technical Lemmas \ref{eq: HK_uniform_bound}\textendash \ref{lem: CLT_convergence},
where we found it convenient to use the parabolic Harnack inequality
of Delmotte~\cite{Delmotte}; we note that we use the continuous
time version that is significantly simpler than the discrete time
one. With some work, these lemmas, modified accordingly, can be given
alternative proofs, allowing one to lift the symmetry assumption.
The same applies to the unitarity of $\varphi$, which can probably
be relaxed under the assumption that the real parts of the eigenvalues
of the Laplacian remain positive.

We are grateful to Eveliina Peltola and Nikolai Reshetikhin for interesting
discussions, and to the referees for many useful suggestions. The
work is supported by Academy of Finland in the framework of Center
of Excellence in Analysis and Dynamics research and Academy project
``Critical phenomena in dimension two'', and the ERC advanced grant
741487. The work of M.~K. was supported by Russian Science Foundation
grant 19-71-30002.

\section{Laplacians, heat kernels and zeta functions}

\label{sec: notation}Although we are mainly interested in the discretizations
of the Riemann surfaces, we start from a more general setup. Let $G=(\Vertices,\Edges)$
be a connected finite undirected graph with weights $w_{e}>0$ assigned
to its edges. A \emph{vector bundle} of rank $d$ over $G$ is a collection
of $d$-dimensional complex vector spaces $V_{x}$ attached to its
vertices $x\in\Vertices$. A \emph{connection} on a vector bundle
is a collection of linear isomorphisms $\varphi_{xy}=\varphi_{yx}^{-1}:V_{x}\to V_{y}$
for each pair $x\sim y$ of adjacent vertices.  From now on, we will
only consider \emph{unitary} connections, that is, each $V_{x}$ is
equipped with an inner product and the maps $\varphi_{xy}$ are unitary.
A \emph{section} $f$ of a vector bundle is a choice of an element
$f(x)\in V_{x}$ for each $x\in\Vertices$; we will take the liberty
to refer to it as a section of $\varphi$ in order to lighten the
notation. When the vector bundle is unitary, the linear space of it
sections comes with the inner product $\langle f;g\rangle=\sum_{x\in\Vertices}\langle f(x);g(x)\rangle$. 

The Laplace operator $\Delta^{G,\varphi}$ acts on sections of the
vector bundles by the formula 
\[
\left(\Delta^{G,\varphi}f\right)(x):=\sum_{(yx)\in\Edges}w_{xy}\left(f(x)-\varphi_{yx}f(y)\right).
\]
Note that with respect to the inner product as above, we have 
\[
\langle\Delta^{G,\varphi}f;g\rangle=\frac{1}{2}\sum_{(xy)\in\Edges}w_{xy}\langle\varphi_{yx}f(y)-f(x);\varphi_{yx}g(y)-g(x)\rangle
\]
which shows that $\Delta^{G,\varphi}$ is non-negative and self-adjoint
and therefore diagonalizable with real non-negative eigenvalues.

We define the \emph{heat operator} associated to $\Delta^{G,\varphi}$
by 
\[
P_{t}^{G,\varphi}:=\exp(-t\Delta^{G,\varphi}).
\]
This is again a linear operator acting on the linear space of sections
of the vector bundle. By linearity, we can write $(P_{t}^{G,\varphi}f)(y)=\sum_{x\in\Vertices}P^{G,\varphi}(x,y,t)f(x)$,
where $P^{G,\varphi}(x,y,t):V_{x}\to V_{y}$ is a linear operator,
called the \emph{heat kernel}.

The trace of $P_{t}^{G,\varphi}$ is called the \textit{theta function}.
By computing the trace first as the sum of eigenvalues, and then as
the sum of the diagonal elements, we get the \emph{theta inversion
identity} 
\begin{equation}
\Theta^{G,\varphi}(t):=\Tr P_{t}^{G,\varphi}=\sum_{\lambda\in\sigma(\Delta^{G,\varphi})}e^{-\lambda t}=\sum_{x\in\Vertices}\Tr P^{G,\varphi}(x,x,t).\label{eq: theta_inversion}
\end{equation}
The \emph{zeta function} associated to $\Delta^{G,\varphi}$ is the
Mellin transofrm of $\Theta^{G,\varphi}(t)$, defined for $s\in\mathbb{C}$
as
\[
\zeta^{G,\varphi}(s):=\sum_{0\neq\lambda\in\sigma(\Delta^{G,\varphi})}\lambda^{-s}.
\]
Let $k=\dim\ker\Delta^{G,\varphi}$. By subtracting $k$ from both
sides of the theta inversion formula, multiplying by $t^{s-1}$ and
integrating, one gets 
\begin{equation}
\zeta^{G,\varphi}(s)=\frac{1}{\Gamma(s)}\int_{0}^{\infty}(\Theta^{G,\varphi}(t)-k)t^{s-1}\,dt,\quad\text{if }\re(s)>0;\label{eq:zeta=00003DintTheta}
\end{equation}
note that the integral converges at infinity because of the positivity
of the eigenvalues. Our main motivation for studying the zeta function
is the identity 
\begin{equation}
-\left(\zeta^{G,\varphi}\right)'(0)=\sum_{0\neq\lambda\in\sigma(\Delta^{G,\varphi})}\log\lambda=\log\Det\left(\Delta^{G,\varphi}\right).\label{eq:-zeta'=00003Dlogdet}
\end{equation}

We will need a probabilistic interpretation of the heat kernel. Let
$\gamma_{t}$ denote the continuous time random walk on the weighted
graph $(G;w)$, that, being at $x\in\Vertices$, moves following exponential
clock to an adjacent vertex $y$ with intensity $w_{xy}$. We denote
$\varphi_{\gamma_{[0;t]}}:=\varphi_{x_{n-1}x_{n}}\circ\dots\circ\varphi_{x_{0}x_{1}}$,
where $\gamma_{0}=x_{0}\sim x_{1}\sim\dots\sim x_{n}=\gamma_{t}$
are the vertices visited consecutively by $\gamma$ up to time $t$.
We denote by $\P^{x}$ and $\E^{x}$ the probability and the expectation
with respect to this random walk started at $x$. We have the following
lemma: 

\begin{lem}
We have 
\begin{equation}
P^{G,\varphi}(x,y,t)=\E^{x}(\varphi_{\gamma_{[0,t]}}\ind_{\gamma_{t}=y}).\label{eq: P_probab}
\end{equation}
\end{lem}

\begin{proof}
Denote the right-hand side by $\hat{P}^{G,\varphi}(x,y,t)$. We observe
that 
\[
\partial_{t}\hat{P}^{G,\varphi}(x,y,t)+\Delta_{x}^{G,\varphi}\hat{P}^{G,\varphi}(x,y,t)=0,
\]
and the same system of ODEs with the same initial conditions holds
for $P^{G,\varphi}(x,y,t)$, and the solution is unique. 
\end{proof}
We use (\ref{eq: P_probab}) to extend the definition of the heat
kernel to infinite graphs.

In what follows, the graph $G$ will be a discretization of a triangulated
or a quadrangulated surface $\Omega$, as in the introduction. To
discretize $\Omega$, choose an infinite locally finite weighted connected
graph $\C^{\delta_{0}}$ with vertices embedded in the plane; we assume
that the embedded weighted graph is bi-periodic with periods either
$1$ and $\frac{1}{2}+\frac{\sqrt{3}}{2}i$ (for triangulation case),
or $1$ and $i$ (for quadrangulation case) and has $\delta_{0}^{-2}$
vertices per unit area. We moreover assume that $\C^{\delta_{0}}$
is preserved under rotations by $\pi/3$, respectively, $\pi/2$,
around the origin, and under reflections with respect to the real
line. We denote $\C^{\delta}=\frac{1}{N}\C^{\delta_{0}},$ where $\delta=\delta_{0}/N$,
$N\in\N$, so that $\C^{\delta}$ has $\delta^{-2}$ vertices per
unit area. Let $T$ denote the unit triangle $\{0;1;\frac{1}{2}+\frac{\sqrt{3}}{2}i\}$
(respectively, the unit square $\{0,1,1+i,i\}$). We denote by $\Od$
the discrete surface obtained by discretizing each triangle/square
in $\Omega$ with $T^{\delta}=\C^{\delta}\cap T$; since $\C^{\delta}$
has all the symmetries of $\C^{\delta_{0}}$, these discretizations
can be naturally glued together. We will interchangeably use $\delta$
and $N=\delta_{0}\delta^{-1}$ as mesh parameters of the discretization.
We do not require $\C^{\delta_{0}}$ to be properly embedded, i.e.,
edges are allowed to intersect; however, we do assume that the graph
obtained by removing edges connecting a vertex strictly inside $T$
with one strictly outside $T$ is still connected, so that $\Od$
is connected. 

As an example, $\Z+i\text{\ensuremath{\Z}}$ and $\frac{1}{2}+\frac{i}{2}+\Z+i\Z$
can serve as $\C^{\delta_{0}}$ (with $\delta_{0}=1$ in this case)
in the quadrangulated case, leading to two different families of discretizations
(one of them will have vertices at cone tips). The square lattice
rotated by $45^{\circ}$ yields another discretization. The triangular
and the hexagonal lattices can serve as $\C^{\delta_{0}}$ in the
triangulation case, with $\delta_{0}=\frac{\sqrt[4]{3}}{\sqrt{2}}$
and $\delta_{0}=\frac{\sqrt[4]{3}}{2}$ respectively, and also one
may choose to rotate them, see Figure \ref{fig:NiceLattices}. 

It is easy to show (see Lemma \ref{prop:KMT} below) that the continuous
time random walk on $\C^{\delta}$ satisfies the Central limit theorem,
i.e., converges, as $\delta\to0,$ to a Brownian motion, whose covariance
matrix must be scalar, because of the symmetries. We will assume that
the weights $w_{xy}$ are chosen so that this matrix is the identity,
i.e., $\sum_{y\in B}P^{\C^{\delta}}(x,y,\delta^{-2}t)\stackrel{\delta\to0}{\longrightarrow}\int_{B}\frac{1}{2\pi t}\exp(-\frac{|x-y|^{2}}{2t})dy$
for any disc $B$. This can always be achieved by simultaneously multiplying
all the weights by a common factor, which only affects the values
of the lattice-dependent constants in (\ref{eq: main}). 

Let us comment on the boundary conditions. We will assume that if
a point $p\in\partial\Omega$ of Dirichlet-to-Neumann change is at
a corner, then, at $\Od,$ the boundary conditions also change at
the corresponding corner, and if $p$ is an inner point of a side
of one of the triangles/squares, then it is approximated by sequence
of points $p^{\delta}\to p$ at distance $m^{\delta}/N$, $m^{\delta}\in\N$,
from the corner of that triangle (or square). To define an action
of $\Delta^{\Od,\varphi}$ on a section $f$ with Dirichlet (respectively,
Neumann) boundary condition at a boundary segment $l$, we extend
$f$ across $l$ by $f(z^{\star})=\mp f(z),$ where $z\mapsto z^{\star}$
is the reflection with respect to $l$; if there are vertices on the
Dirichlet boundary, we only consider sections $f$ that are zero at
those vertices. This procedure may lead to non-symmetric weights that
are gauge equivalent to symmetric ones; the above formulae are not
affected. We adopt the convention that when writing the sum over $x\in\Od$
(as e.g. in (\ref{eq: theta_inversion})) we do not include vertices
lying on the Dirichlet part of the boundary, but do include those
on the Neumann part.

An approximation $p^{\delta}$ to a puncture $p$ will be realized
as a point of $\Omega$ disjoint from any edge of $\Od$; we moreover
insist that it is an image under scaling of a \emph{fixed} point in
a fundamental domain of $\C^{\delta_{0}}.$ 

\section{Heat kernels and theta functions in the continuum}

\label{sec: heat_kernel_continuum}We recall a continuous version
of the above theory, in which the random walk $\gamma^{\delta}$ on
$\Od$ is replaced by the Brownian motion $\gamma_{t}$ on $\Omega,$
reflected at $\partial_{\Neu}\Omega$ and absorbed at $\partial_{\Dir}\Omega$.
This is only needed for the interpretation of the constant term in
the asymptotics (\ref{eq: main}); a reader willing to accept (\ref{eq: Theta_cont_integral}\textendash \ref{eq: def_zeta_reg})
as the definition of the zeta-reguralized determinant of the Laplacian
may skip most of the rest. Recall \cite{mooers1999heat} that the
Laplacian on surfaces with conical singularities is not essentially
self-adjoint; hence some care is needed when specifying its self-adjoint
extension, see also \cite[Section 2.3]{finski2020finite}.

We start by constructing the Brownian motion on $\Omega,$ which can
be done by elementary means. Let $\hat{\Omega}$ be two copies of
$\Omega$ glued along the boundary with conical singularities removed;
we can define the Brownian motion $\hat{\gamma}$ on $\hat{\Omega}$
by coupling it to the Brownian motion $\check{\gamma}$ in the plane,
lifting $\check{\gamma}$ to $\hat{\Omega}$ by local isometries.
The lifting is well defined at least up to the first time $\hat{\gamma}$
hits a conical singularity of $\hat{\Omega}$, but this can only happen
when $\check{\gamma}$ hits a point of a triangular or a square lattice,
i.e., with probability $0.$ Hence, almost surely for $\Leb(\hat{\Omega})$
a. e. starting point, $\hat{\gamma}_{t}$ is defined for all $t$.
We then define $\gamma_{t}$ to be $\hat{\gamma}_{t}$ reflected to
$\Omega$ and stopped upon hitting $\pa_{\Dir}\Omega.$ 
\begin{lem}
\label{lem: symmetry}The Markov process $\gamma_{t}$ is symmetric
with respect to the Lebesgue measure on $\Omega.$ Moreover, for each
fixed $t>0,$ its transition kernel $P^{\Omega}(\cdot,\cdot,t)$ is
bounded, and it is smooth away from $\pa\Omega$ and the cone tips.
\end{lem}

\begin{proof}
See Section \ref{sec: proof_of_Lemmas}.
\end{proof}
The following construction is general for symmetric Markov processes,
see \cite[Section 1.4]{sznitman1998brownian} or \cite[Section 1.4]{fukushima2010dirichlet}.
Let $L^{2,\varphi}(\Omega)$ be the set of $L^{2}$ sections of $\varphi.$
Define the \emph{Markov semi-group} $P_{t}^{\Omega,\varphi}:L^{2,\varphi}(\Omega)\to L^{2,\varphi}(\Omega)$
by 

\begin{multline}
(P_{t}^{\Omega,\varphi}f)(x):=\E^{x}\left[\varphi_{\gamma_{[0,t]}}^{-1}(f(\gamma_{t}))\right]\\
=\int_{y\in\Omega}\E^{x}\left[\varphi_{\gamma_{[0,t]}}^{-1}(f(y))|\gamma_{t}=y\right]P^{\Omega}(x,y,t)\,dy=\int_{y\in\Omega}P^{\Omega,\varphi}(x,y,t)f(y)dy.\label{eq: semi-group}
\end{multline}
From the symmetry of $\gamma_{t}$ and the unitarity of $\varphi,$
we see that $P_{t}^{\Omega,\varphi}$ are self-adjoint; also, $\left\Vert P^{\Omega,\varphi}\right\Vert \leq P^{\Omega}$
pointwise. Hence, by Cauchy\textendash Schwarz, 
\[
\left\Vert (P_{t}^{\Omega,\varphi}f)(x)\right\Vert ^{2}\leq\int\left\Vert P^{\Omega,\varphi}(x,y,t)\right\Vert \,dx\int_{y\in\Omega}\left\Vert P^{\Omega,\varphi}(x,y,t)\right\Vert \left\Vert f(y)\right\Vert ^{2}\,dy\leq\left\Vert f\right\Vert ^{2},
\]
i.e., for each $t>0,$ $P_{t}^{\Omega,\varphi}$ is a self-adjoint
contraction on $L^{2,\varphi}(\Omega)$. The semi-group $P_{t}^{\Omega,\varphi}$
is strongly continuous: it is enough to check that $P_{t}^{\Omega,\varphi}f\stackrel{t\to0}{\longrightarrow}f$
in $L^{2}$ for a continuous $f$, in which case the convergence clearly
holds uniformly. Put $D_{t}^{\Omega,\varphi}=t^{-1}(I-P_{t}^{\Omega,\varphi});$
by looking at the spectral decomposition of the semigroup $\{P_{t}^{\Omega,\varphi}\}_{t>0}$,
we see that $\langle D_{t}^{\Omega,\varphi}f,f\rangle$ is decreasing
in $t$. We define the \emph{Dirichlet form }$\En^{\Omega,\varphi}(f,f):=\lim_{t\searrow0}\langle D_{t}^{\Omega,\varphi}f,f\rangle$
on the set $\mathcal{D}(\En^{\Omega,\varphi})$ of all $f\in L^{2,\varphi}(\Omega)$
for which the limit is finite. The form $\En^{\Omega,\varphi}$ is
non-negative and closed \cite[Lemma 4.2]{sznitman1998brownian}, and
we denote by $\Delta^{\Omega,\varphi}$ the unique non-negative self-adjoint
operator associated to $\En^{\Omega,\varphi},$ called the \emph{generator}
of $\{P_{t}^{\Omega,\varphi}\}_{t>0}.$ In terms of the common spectral
projections $E_{\lambda}$ for $P_{t}^{\Omega,\varphi},$ we have
$P_{t}^{\Omega,\varphi}f=\int_{0}^{\infty}e^{-\lambda t}\,dE_{\lambda}f,$
$\Delta^{\Omega,\varphi}f=\int_{0}^{\infty}\lambda\,dE_{\lambda}f,$
and $\En^{\Omega,\varphi}(f,g)=\int_{0}^{\infty}\lambda\langle dE_{\lambda}f,g\rangle,$
with the domains $L^{2,\varphi}(\Omega),$ $\mathcal{D}(\Delta^{\Omega,\varphi})=\{f\in L^{2,\varphi}(\Omega):\int_{0}^{\infty}\lambda^{2}\,\langle dE_{\lambda}f,f\rangle<\infty\}$
and $\mathcal{D}(\En^{\Omega,\varphi})=\{f\in L^{2,\varphi}(\Omega):\int_{0}^{\infty}\lambda\,\langle dE_{\lambda}f,f\rangle<\infty\}$
respectively. 

Since $P^{\Omega,\varphi}(\cdot,\cdot,t)$ is bounded, $P_{t}^{\Omega,\varphi}$
is Hilbert-Schmidt, and hence $P_{2t}^{\Omega,\varphi}$ is trace
class for all $t>0.$ Therefore, we can define the\emph{ spectral
theta function} by 
\begin{equation}
\Theta^{\Omega,\varphi}(t)=\Tr P_{t}^{\Omega,\varphi}=\sum_{i=1}^{\infty}e^{-\lambda_{i}t},\quad t>0.\label{eq: Theta_def_cont}
\end{equation}
where $0\leq\lambda_{1}\leq\lambda_{2}\leq\dots$ are the eigenvalues
of $\Delta^{\Omega,\varphi}.$ We have, for any $f\in L^{2}(\Omega),$
\[
\int_{\Omega}P^{\Omega,\varphi}(x,y,t)f(y)dy=(P_{t}^{\Omega,\varphi}f)(x)=\sum_{i=1}^{\infty}(P_{t}\psi_{i})(x)\langle\psi_{i};f\rangle=\sum_{i}e^{-\lambda_{i}t}\psi_{i}(x)\langle\psi_{i};f\rangle,
\]
thus, for any vector $v\in V_{y}$, $P^{\Omega,\varphi}(x,y,t)v=\sum_{i}e^{-\lambda_{i}t}\psi_{i}(x)\langle\psi_{i}(y),v\rangle,$
so that $\Tr P^{\Omega,\varphi}(x,x,t)=\sum_{i}e^{-\lambda_{i}t}\langle\psi_{i}(x);\psi_{i}(x)\rangle$
and, integrating, we arrive at the continuous theta inversion identity,
\begin{equation}
\int_{x\in\Omega}\Tr P^{\Omega,\varphi}(x,x,t)\,dx=\Theta^{\Omega,\varphi}(t).\label{eq: Theta_cont_integral}
\end{equation}

We define the \emph{spectral zeta function} by 
\begin{equation}
\zeta^{\Omega,\varphi}(s)=\frac{1}{\Gamma(s)}\cdot\int_{0}^{\infty}(\Theta^{\Omega,\varphi}(t)-k)t^{s-1}\,dt,\quad\re s>1,\label{eq: zet_int_theta_cont}
\end{equation}
postponing the proof of convergence to Section \ref{sec: proof_of_the_main}.
Here $k=\dim\ker\Delta^{\Omega,\varphi};$ we can also describe $k$
explicitly (for a connected $\Omega$) as follows: if $\pa_{\Dir}\Omega\neq\emptyset,$
then $k=0,$ otherwise, $k$ is the dimension of the \emph{maximal
trivial sub-bundle}. This identity follows from Lemma \ref{prop: (Spectral-gap)},
which uses the latter definition $k$, and whose conclusion, together
with (\ref{eq: CLT_1}), implies that both $\Theta^{\delta}(t)-k$
and $\Theta(t)-k$ tend to zero as $t\to\infty$. In particular, $k=\dim\ker\Delta^{\Omega,\varphi}=\dim\ker\Delta^{\Od,\varphi}=k^{\delta}$.
Alternatively, one could note from Lemma \ref{lem: domain} that,
if $\psi\in\ker\Delta^{\Omega,\varphi},$ then $\langle\nabla^{\varphi}\psi,\nabla^{\varphi}\psi\rangle=0,$
that is, $\psi$ is a \emph{covariant constant}, also known as \emph{a
flat section}, cf. \cite[Corollary 2.6]{finski2020finite}. 

As in the discrete case, we have $\zeta^{\Omega,\varphi}(s)=\sum_{i}\lambda_{i}^{-s}$
whenever either the series or the integral (\ref{eq: zet_int_theta_cont})
converges absolutely. Moreover, in in Section \ref{sec: proof_of_the_main},
we analytically continue $\zeta^{\Omega,\varphi}$ into a neighborhood
of the origin, cf. \cite{cheeger1983spectral}. In analogy with (\ref{eq:-zeta'=00003Dlogdet}),
this allows one to define the zeta-regularized determinant of $\Delta^{\Omega,\varphi}$
(cf. \cite{raysinger,ray1973analytic}) by 
\begin{equation}
\log\Det_{\zeta}\left(\Delta^{\Omega,\varphi}\right)=-\left(\zeta^{\Omega,\varphi}\right)'(0).\label{eq: def_zeta_reg}
\end{equation}
It is possible to describe the form $\mathcal{E}^{\Omega,\varphi}$
and its domain $\mathcal{D}(\mathcal{E}^{\Omega,\varphi})$ (or, equivalently,
the generator $\Delta^{\Omega,\varphi}$) more explicitly, which we
postpone to Section \ref{sec: Appendix}. 
\begin{figure}
\includegraphics[width=0.8\textwidth]{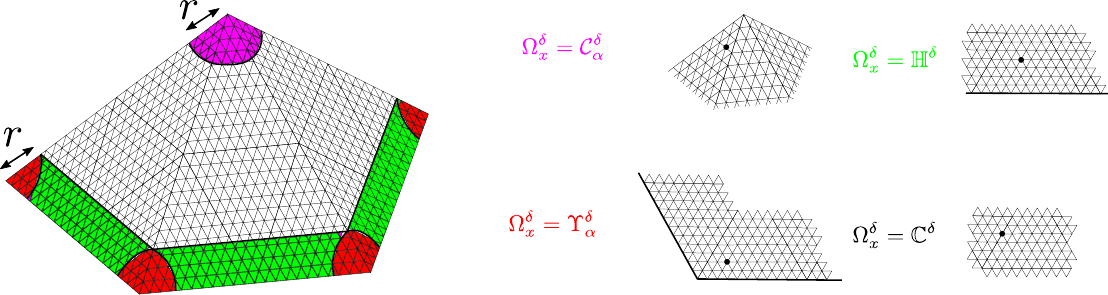}\caption{The discretized surface $\protect\Od_{x}$ is a plane, half-plane,
a infinite cone or an infinite wedge, depending on the local geometry
of $\protect\Od$ near $x.$ }
\end{figure}

\section{The key formula for the determinant of the discrete Laplacian}

\label{sec: the Key formula}For notational simplicity, we first assume
there are no punctures, and also that the lattice is such that there
are no vertices at conical singularities. We will then discuss the
necessary modifications in the general case.

Given $x\in\Omega$, we define $\Omega_{x}$ to be one of the model
surfaces, namely, a plane, a half-plane, an infinite cone, or an infinite
wedge, that agrees with $\Omega$ locally near $x$. More precisely,
fix a small $r>0$ in such a way that the $2r$\textendash neighborhoods
of the tips of the conical singularities and the boundary components
do not overlap or self-overlap. We define $r_{\alpha}:=r/\sin(\alpha/2)$
if $\alpha<\pi$ and $r_{\alpha}=r$ otherwise. We then define $\Omega_{x}$
to be the wedge of angle $\alpha$ if $x$ is at distance at most
$r_{\alpha}$ from the tip of a corner (see Figure \ref{fig: boundary})
with angle $\alpha$, else, if $x$ at distance $\leq r$ from $\partial\Omega$
(or a conical singularity), we define $\Omega_{x}$ to be the half-plane
(respectively, the cone); else, $\Omega_{x}$ is a plane. 

In a similar way (with the same $r$ independent of $\delta$), we
define the infinite graphs $\Omega_{x}^{\delta}$ discretizing each
respective $\Omega_{x}$. These graphs, when they have a boundary,
come equipped with boundary conditions inherited from $\Omega^{\delta}$.
We define the heat kernel $P$ in each of the discrete model domains
$\Omega_{x}^{\delta}$ by (\ref{eq: P_probab}), with $\varphi$ the
trivial connection on the rank one bundle, and the random walk being
stopped at the Dirichlet boundary and reflected at the Neumann one.

We start from (\ref{eq:zeta=00003DintTheta}) and rewrite it as
\begin{align*}
\zeta^{\Omega^{\delta},\varphi}(s) & =\frac{1}{\Gamma(s)}\int_{\delta^{-2}}^{\infty}\left(\Theta^{\Omega^{\delta},\varphi}(t)-k\right)t^{s-1}\,dt\\
 & +\frac{1}{\Gamma(s)}\sum_{x\in\Od}\int_{0}^{\delta^{-2}}\left(\Tr P^{\Omega^{\delta},\varphi}(x,x,t)-d\cdot P^{\Omega_{x}^{\delta}}(x,x,t)\right)t^{s-1}\,dt\\
 & +\frac{d}{\Gamma(s)}\sum_{x\in\Od}\int_{0}^{\infty}\left(P^{\Omega_{x}^{\delta}}(x,x,t)-P^{\C^{\delta}}(x,x,t)\right)t^{s-1}\,dt-\frac{d}{\Gamma(s)}\sum_{x\in\Od}\int_{\delta^{-2}}^{\infty}P^{\Omega_{x}^{\delta}}(x,x,t)t^{s-1}dt\\
 & +\frac{d}{\Gamma(s)}\sum_{x\in\Od}\int_{0}^{\infty}\left(P^{\C^{\delta}}(x,x,t)-e^{-w_{x}t}\right)t^{s-1}\,dt\\
 & +\frac{d}{\Gamma(s)}\sum_{x\in\Od}\int_{0}^{\infty}e^{-w_{x}t}t^{s-1}\,dt-\frac{1}{\Gamma(s)}\int_{0}^{\delta^{-2}}kt^{s-1}\,dt,
\end{align*}
where $w_{x}:=\sum_{y\sim x}w_{xy}.$ This identity is valid for $0<\re s<1$,
since each integral defines an analytic function in that region (see
Lemma \ref{lem: HK_infinite_tail_bound} below for large $t$ bounds).
Moreover, all the integrals but the last two are in fact analytic
at least in $-1<\re s<1$; indeed it follows from (\ref{eq: P_probab})
that $P^{\C^{\delta}}(x,x,t)=e^{-w_{x}t}\cdot\Id+O(t^{2}),$ where
$e^{-w_{x}t}$ is the probability that the random walk does not move
at all by time $t,$ and $O(t^{2})$ is the probability that it makes
at least two steps, which is bounded from above by $\max_{y\in T}(1-e^{-w_{x}t})(1-e^{-w_{y}t})$.
The last two terms are equal to
\[
d\sum_{x\in\Omega^{d}}w_{x}^{-s}-\frac{k\delta^{-2s}}{s\Gamma(s)},
\]
and since $1/\Gamma(s)=s(1+o(1))$ as $s\to0$, the derivative of
$\zeta^{\Omega^{\delta},\varphi}$ at zero evaluates to 
\begin{multline}
-\log\Det\Delta^{\Omega^{\delta},\varphi}\\
=\int_{\delta^{-2}}^{\infty}\left(\Theta^{\Omega^{\delta},\varphi}(t)-k\right)\ddt+\int_{0}^{\delta^{-2}}\sum_{x}\left(\Tr\Pp^{\Omega^{\delta},\varphi}(x,x,t)-d\cdot P^{\Omega_{x}^{\delta}}(x,x,t)\right)\ddt\\
-d\cdot\int_{\delta^{-2}}^{\infty}\sum_{x}P^{\Omega_{x}^{\delta}}(x,x,t)\ddt+d\cdot\int_{0}^{\infty}\sum_{x}\left(P^{\Omega_{x}^{\delta}}(x,x,t)-P^{\C^{\delta}}(x,x,t)\right)\ddt\\
+d\sum_{x}\left(\int_{0}^{\infty}\left(P^{\C^{\delta}}(x,x,t)-e^{-w_{x}t}\right)\ddt-\log w_{x}\right)+2k\log\delta-k\EulerGamma.\label{eq: key_formula}
\end{multline}
This is our \textbf{key formula}; analyzing it term by term will lead
to (\ref{eq: main}). The last two terms are already explicit, and
we see the $-2k=-2\dim\ker\Delta^{\Omega,\varphi}$ contribution to
the logarithmic term in (\ref{eq: main}). For the first three integrals,
going back to the probabilistic interpretation of the heat kernel,
we observe that only the walks with $\gtrsim\delta^{-1}$ steps contribute,
hence, these terms converge to their continuous counterparts by Central
limit theorem, see Section \ref{sec: CLT} and Corollary \ref{cor: third_term}
for details. As for the fourth term, we note that the summands are
zero unless $x$ is $r$\textendash close to a conical singularity
or to the boundary. Thus, the whole sum only depends on the number
of conical singularities and their angles, and on the geometry of
the boundary. We treat it in Section \ref{sec: Local_contr}.

Turning to the fifth term in (\ref{eq: key_formula}), denote 
\begin{equation}
A_{x}:=\int_{0}^{\infty}\left(P^{\C^{\delta}}(x,x,t)-e^{-w_{x}t}\right)\ddt-\log w_{x}=\int_{0}^{\infty}\left(P^{\C^{\delta}}(x,x,t)-e^{-t}\right)\ddt.\label{eq: A_x}
\end{equation}
 and observe that this quantity only depends on the vertex of $T^{\delta_{0}}$
corresponding to $x$ in the discretization procedure, in particular
in the vertex-transitive case this is just a constant. Assume first
that there are no vertices of $\C^{\delta_{0}}$ on $\partial T^{\delta_{0}}$.
Then, subdividing the vertices of $\Od$ into scaled copies of $T^{\delta_{0}}$,
the fifth term above gives the leading term of the asymptotics (\ref{eq: main}):
\begin{equation}
d\sum_{x\in\Od}A_{x}=-A\cdot|\Od|\quad\text{where}\quad A:=-d\sum_{x\in T^{\delta_{0}}}A_{x}.\label{eq: volume}
\end{equation}
If there are vertices of $T^{\delta_{0}}$ on $\partial T$ but not
in its corners, then the contribution of those should be included
in the definition of $A$ with weights $\frac{1}{2}$. This leads
to a miscount for the contribution of the vertices on $\partial\Omega$,
which we absorb into the $|\pa\Omega^{\delta}|$ term by re-defining
the constants $B_{\Neu},B_{\Dir},$ cf. the second term in (\ref{eq: B})
below. Similarly, if there are vertices at the corners of $T^{\delta_{0}}$,
they should be counted with weight $\frac{1}{6}$ or $\frac{1}{4}$,
which leads to a miscount for boundary corners and cones which we
absorb into $D_{p}$.

\textcolor{black}{In the case there are punctures, in an $r$-neighborhood
of a puncture $p$, we define $\Omega_{x}$ to be the punctured plane
$\C^{\delta}\setminus\{p\}$, equipped with the connection $\varphi_{p}$
obtained by first restricting $\varphi$ to the neighborhood of $p$
and then extending it to a flat connection on the whole $\C^{\delta}\setminus\{p\}.$
We then simply use $P^{\C^{\delta}\setminus\{p\},\varphi_{p}}$ instead
of $d\cdot P^{\Od_{x}}$ in the above formulae. If there are vertices
at the corners of $T$ (and thus at the conical singularities), then
the asymptotics of the heat kernel at a conical singularity $p$ of
angle $\alpha$ reads $P^{\Od_{p}}(p,p,t)=\frac{\alpha}{2\pi}e^{-w_{\hat{p}}t}\cdot\Id+O(t^{2}),$
where $\hat{p}$ is a corner of $T$; hence we should replace $P^{\C^{\delta}}(p,p,t)$
in the above formulae by $\frac{\alpha}{2\pi}\cdot P^{\C^{\delta}}(\hat{p},\hat{p},t)$.
This results in additional constant contributions to the asymptotics
that can be absorbed into $D_{p}$.}

\section{Contributions from the CLT part}

\label{sec: CLT}The goal for this section is to prove convergence
of the first two terms in the key formula. We start with five fairly
standard Lemmas, whose proof is deferred to Section \ref{sec: proof_of_Lemmas}.
We denote $\SingularSet:=\partial\Omega\cup\Cone{}\cup\Puncture$
\begin{lem}
\label{prop: FCLT}(Functional CLT) For any $T>0$, any $x\in\Omega\setminus\SingularSet$,
any sequence $x^{\delta}\to x$, and any bounded, continuous function
$f$ on the space of paths $\gamma:[0,T]\to\Omega$ (equipped with
sup-norm convergence), one has 
\[
\E^{x^{\delta}}\left[f(\gamma_{[0,\delta^{-2}T]}^{\delta})\right]\stackrel{\delta\to0}{\longrightarrow}\E^{x}\left[f(\gamma_{[0,T]})\right].
\]
\end{lem}

\begin{lem}
(Short time large diameter bound) \label{lem: Short-time-large-diameter}For
every $\eps>0$, there are constants $C,c>0$ such that, for all $x\in\Od$,
all $t>0$ and all $\delta<c$, 
\[
\P^{x}(\diam(\gamma_{[0,\delta^{-2}t]}^{\delta})\geq\eps)\leq C\cdot t^{3}.
\]
\end{lem}

\begin{lem}
\label{prop: HK_bound}(Uniform bound of the heat kernel) For each
$\eps>0$, there are constants $C,c>0$ such that, for all $\delta<c$,
\begin{equation}
P^{\Od}(x,y,\delta^{-2}t)\leq C\delta^{2}\label{eq: HK_uniform_bound}
\end{equation}
whenever either $\dist(x,y)>\eps,$ or $t>\eps$. 
\end{lem}

\begin{lem}
\label{prop: Holder-regularity}(Hölder regularity of heat kernels)
There exists a number $\theta>0$ such that, for any $\eta>0$, there
exists a constant $C_{\eta}$ with the following property: if $\dist(x,y)<\frac{1}{2}\dist(x,\SingularSet)<\frac{1}{2}\eta$
and $t>\eta^{2}$, then 
\begin{equation}
\left|\Pp^{\Od,\varphi}(x,x,\delta^{-2}t)-\Pp^{\Od,\varphi}(x,y,\delta^{-2}t)\right|\leq C_{\eta}\cdot|x-y|^{\theta}\cdot\delta^{2},\label{eq: Holder-regularity}
\end{equation}
where, to make sense of the left-hand side, we identify the vector
spaces $V_{y},$ $y\in B(x,\frac{\eta}{2}),$ using a trivialization
of $\varphi$ over $B(x,\frac{\eta}{2})$.
\end{lem}

\begin{lem}
\label{prop: (Spectral-gap)}(Spectral gap) There are constants $C,c>0$
independent of $\delta$, such that for all $\delta$ small enough,
one has
\begin{equation}
\left|\Theta^{\Od}(\delta^{-2}t)-\dimphi\right|<Ce^{-ct},\quad t\geq1.\label{eq: spectral_gap}
\end{equation}
\end{lem}

We are in the position to prove convergence of the first two terms
in the key formula:
\begin{lem}
\label{lem: CLT_convergence} We have the following convergence results:
\begin{equation}
\int_{1}^{\infty}\left(\Theta^{\Od,\varphi}(\delta^{-2}t)-k\right)\ddt\stackrel{\text{as }\delta\to0}{\longrightarrow}\int_{1}^{\infty}\left(\Theta^{\Omega,\varphi}(t)-k\right)\ddt\label{eq: CLT_1}
\end{equation}
and 
\begin{multline}
\int_{0}^{1}\sum_{x\in\Od}\left(\Tr\Pp^{\Od,\varphi}(x,x,\delta^{-2}t)-d\cdot P^{\Od_{x}}(x,x,\delta^{-2}t)\right)\ddt\stackrel{\text{as }\delta\to0}{\longrightarrow}\\
\int_{0}^{1}\int_{\Omega}\left(\Tr P^{\Omega,\varphi}(x,x,t)-d\cdot P^{\Omega_{x}}(x,x,t)\right)\,dx\ddt.\label{eq: CLT_2}
\end{multline}
\end{lem}

\begin{proof}
In view of Lemma \ref{prop: (Spectral-gap)}, for (\ref{eq: CLT_1}),
it suffices to prove the convergence 
\[
\int_{1}^{T}\Theta^{\Od,\varphi}(\delta^{-2}t)\ddt\stackrel{\text{as }\delta\to0}{\longrightarrow}\int_{1}^{T}\Theta^{\Omega,\varphi}(t)\ddt\,dx.
\]
for any fixed $T>0$. Let $0<\eta<\eta_{0}$, and let $\{\psi_{j}\}$
be a partition of unity for $\Omega$ such that $\mathrm{diam}(\mathrm{supp}\;\psi_{j})<\eta$
for all $j$. We write 
\[
\int_{1}^{T}\Theta^{\Od}(\delta^{-2}t)\ddt=\sum_{j}\int_{1}^{T}\sum_{x\in\Od}\psi_{j}(x)\Tr\Pp^{\Od,\varphi}(x,x,\delta^{-2}t)\ddt
\]
and split the sum according to whether $j\in J:=\{j:\dist(\mathrm{supp}\;\psi_{j},\Omega^{\dagger})>\eta_{0}\}$.
By Lemma \ref{prop: HK_bound}, 
\[
\sum_{j\notin J}\int_{1}^{T}\sum_{x\in\Od}\psi_{j}(x)\Tr\Pp^{\Od,\varphi}(x,x,\delta^{-2}t)\ddt\leq\sum_{\dist(x,\SingularSet)\leq2\eta_{0}}\sum_{j}\psi_{j}(x)C\delta^{2}\leq C\mathcal{A}(\eta_{0}),
\]
where $\mathcal{A}(\eta_{0})\stackrel{\eta_{0}\to0}{\longrightarrow}0$
is the area of $\{x\in\Omega:\dist(x,\SingularSet)\leq2\eta_{0}\}$.
For $j\in J$, we may apply Lemma~\ref{prop: Holder-regularity}
to any $x,y\in\supp\psi_{j}$. Averaging (\ref{eq: Holder-regularity})
with weights $\frac{\psi_{j}(x)\psi_{j}(y)}{S_{j}^{2}}$, where $S_{j}=\sum_{x\in\Od}\psi_{j}(x),$
multiplying by $t^{-1},$ and integrating yields 
\[
\left|\int_{1}^{T}\sum_{x\in\Od}\psi_{j}(x)\Tr\Pp^{\Od,\varphi}(x,x,\delta^{-2}t)\ddt-\int_{1}^{T}\frac{1}{S_{j}}\sum_{x,y\in\Od}\psi_{j}(x)\psi_{j}(y)\Tr\Pp^{\Od,\varphi}(x,y,\delta^{-2}t)\ddt\right|\leq C_{\eta_{0}}S_{j}\eta^{\theta}\delta^{2}.
\]
 Summing these bounds over $j\in J$ yields the upper bound $C_{\eta_{0}}|\Od|\delta^{2}\eta^{\theta}$,
which goes to zero uniformly in $\delta$ as $\eta\to0$. Finally,
\begin{equation}
\int_{1}^{T}\frac{1}{S_{j}}\sum_{x,y\in\Od}\psi_{j}(x)\psi_{j}(y)\Tr\Pp^{\Od,\varphi}(x,y,\delta^{-2}t)\ddt=\E^{X}\left[\int_{1}^{T}\Tr\varphi(\gamma_{[0,\delta^{-2}t]}^{\delta})\psi_{j}(\gamma_{t}^{\delta})\ddt\right],\label{eq: expectation}
\end{equation}
with the initial point $X$ chosen at random with $\P(X=x)=\psi_{j}(x)/S_{j},$
and we pick a trivialization of $\varphi$ over $\supp\psi_{j}$.
The expression inside the expectation is continuous with respect to
the path $\gamma^{\delta},$ therefore, by Lemma \ref{prop: FCLT},
(\ref{eq: expectation}) converges to its continuous counterpart 
\[
\E^{X}\left[\int_{1}^{T}\Tr\varphi(\gamma_{[0,t]})\psi_{j}(\gamma_{t})dt\right]=\int_{1}^{T}\left(\int_{\Omega}\psi_{j}\right)^{-1}\int_{x,y\in\Omega}\psi_{j}(x)\psi_{j}(y)\Tr\Pp^{\Omega,\varphi}(x,y,t)dxdy.=:I_{j}
\]
In view of the bounds we have collected, we have that 
\[
\int_{1}^{T}\Theta^{\Od,\varphi}(\delta^{-2}t)\ddt\stackrel{\text{as }\delta\to0}{\longrightarrow}\lim_{\eta_{0}\to0}\lim_{\eta\to0}\sum_{j\in J}I_{j}.
\]
Since the continuous heat kernels satisfy the suitable counterparts
of Lemmas \ref{prop: HK_bound}\textendash \ref{prop: Holder-regularity}
(for instance, as a consequence of the discrete bounds and the convergence),
an argument as above gives that the latter quantity is equal to $\int_{1}^{T}\Theta^{\Omega,\varphi}(t)\ddt,$
as required.

For (\ref{eq: CLT_2}), the same argument as above, applied to $\Od$
and each of $\Od_{x},$ gives the convergence of the integral from
$t_{0}$ to $1$, for any fixed $t_{0}>0.$ Hence, it suffices to
show that the integral from $0$ to $t_{0}$ converges to $0$ as
$t_{0}\to0$ uniformly in $\delta$. We can write 
\[
\Tr\Pp^{\Od,\varphi}(x,x,\delta^{-2}t)-d\cdot P^{\Od_{x}}(x,x,\delta^{-2}t)=\E^{x}\left[\Tr\varphi(\gamma_{[0,\delta^{-2}t]})\ind_{\gamma_{\delta^{-2}t}=x}-d\cdot\ind_{\hat{\gamma}_{\delta^{-2}t}=x}\right],
\]
 where $\gamma$ and $\hat{\gamma}$ are random walks on $\Od$ and
$\Od_{x}$, respectively, coupled in such a way that they coincide
up until $\tau_{r}:=\min\{s:\dist(\gamma_{\delta^{-2}s},x)\geq r\}.$
On the event $\tau_{r}\geq t$, the expression in the expectation
is zero. Hence, we can write 
\begin{multline*}
\left|\Tr\Pp^{\Od,\varphi}(x,x,\delta^{-2}t)-d\cdot P^{\Od_{x}}(x,x,\delta^{-2}t)\right|\\
\leq d\cdot\P(\tau_{r}<t)\left(\sup_{s<t,|y-x|\geq r}P^{\Omega^{\delta}}(y,x,\delta^{-2}s)+\sup_{s<t,|y-x|\geq r}P^{\Omega_{x}^{\delta}}(y,x,\delta^{-2}s)\right)\\
\leq Cd\cdot\P(\tau_{r}<t)\delta^{2}\leq C't^{3}\delta^{2}.
\end{multline*}
where we have used Lemma \ref{prop: HK_bound} and then Lemma \ref{lem: Short-time-large-diameter}.
Summing over $x$ gives a bound on the integrand in the left-hand
side of (\ref{eq: CLT_2}) that is independent of $\delta$ and integrable
at $0$. This concludes the proof.
\end{proof}

\section{Local contributions}

\label{sec: Local_contr}

In this section, we compute the asymptotics of the local term $\sum_{x\in\Od}I_{\C^{\delta}}^{\Omega_{x}^{\delta}}(x)$
in (\ref{eq: key_formula}), where 
\begin{equation}
I_{\C^{\delta}}^{\Omega_{x}^{\delta}}(x)=\int_{0}^{\infty}\left(\HK{\Od_{x}}{x,x,t}-\HK{\C^{\delta}}{x,x,t}\right)\ddt\label{eq: local_contribution}
\end{equation}
Each of the model surfaces $\Lambda=\H,\Corner^{\alpha},\Cone{\alpha}$
has scaling acting on it, and $I_{\C^{\delta}}^{\Lambda^{\delta}}(x)=I_{\C^{\delta_{0}}}^{\Lambda^{\delta_{0}}}(N\cdot x).$
Thus, decreasing $\delta$ by going from $N$ to $N+1$ is tantamount
to adding new terms to the sum, corresponding to those $x$ whose
distance to a conical singularity, the boundary, or a puncture is
between $r\frac{N}{N+1}$ and $r$. The asymptotics of those new terms
is governed by Central limit theorem. We postpone the proof of the
following Lemmas to Section~\ref{sec: proof_of_Lemmas}:
\begin{lem}
(Local CLT with error bound) \label{lem: LCLT}If $\Lambda$ is one
of the model surfaces, $\Lambda^{\delta}$ its discretization, and
$\eps>0$, then there exist $q>0$ and $C>0$ such that 
\[
\left|\delta^{-2}\cdot P^{\Lambda^{\delta}}(x,y,\delta^{-2}t)-P^{\Lambda}(x,y,t)\right|\leq C\delta^{q}\cdot\max\{t^{-1},1\},
\]
for all $\delta,$ all $t\in(\delta^{q},\delta^{-q})$ and $x,y$
at distance at least $\eps$ from the tip (if $\Lambda$ is a wedge
or a cone).
\end{lem}

\begin{lem}
(Uniform tail bound for the heat kernel) \label{lem: HK_infinite_tail_bound}If
$\Lambda^{\delta}$ is one of the model surfaces, then there exists
$C>0$ such that, for any $\delta,t>0$ and $x,y\in\Lambda^{\delta}$,
\[
P^{\Lambda^{\delta}}(x,y,\delta^{-2}t)\leq C\delta^{2}t^{-1}.
\]
\end{lem}

Let $\Lambda_{1}\ni x_{0}$ be a continuous model surface equipped
with boundary conditions. Let $\Lambda_{2}$ be another model surface
that contains an isometric copy $B'(x_{0},\eta)$ of the ball $B(x_{0},\eta)\subset\Lambda_{1}$
not containing tips of a wedge or a cone, with corresponding parts
of the boundary having the same boundary conditions, and let $\Lambda_{1,2}^{\delta}$
be their discretizations that respect the isometry. We will denote,
for $x\in B(x_{0},\eta),$ 
\[
I_{\Lambda_{1}}^{\Lambda_{2}}(x):=\int_{0}^{\infty}\left(\HK{\Lambda_{2}}{x,x,t}-\HK{\Lambda_{1}}{x,x,t}\right)\ddt,
\]
where we identify the points in $B(x_{0},\eta)$ with their isomorphic
copies. We use a similar notation for discretizations $\Lambda_{1,2}^{\delta}$
of $\Lambda_{1,2}$ We have the following Lemma:
\begin{cor}
\label{cor: int_convergence}In the above setup, there exist $\rho>0$
and $C>0$ such that 
\[
\left|\delta^{-2}I_{\Lambda_{1}^{\delta}}^{\Lambda_{2}^{\delta}}(x)-I_{\Lambda_{1}}^{\Lambda_{2}}(x)\right|\leq C\delta^{\rho};
\]
for all $\delta$ and all $x\in B(x_{0},\frac{\eta}{2}).$
\end{cor}

\begin{proof}
Let $q$ be as in Lemma \ref{lem: LCLT}. At small times $t\leq\delta^{q/2}$,
we repeat the argument in the end of the proof of Lemma \ref{lem: CLT_convergence}
for $d=1$ and $\Lambda_{1,2}^{\delta}$ instead of $\Od,\Od_{x}$;
this gives 
\[
\left|\HK{\Lambda_{2}^{\delta}}{x,x,\delta^{-2}t}-\HK{\Lambda_{1}^{\delta}}{x,x,\delta^{-2}t}\right|\leq C\delta^{2}t^{3},\quad t\leq\delta^{q/2},
\]
and, integrating, 
\[
\left|\int_{0}^{\delta^{q/2}}\left(\HK{\Lambda_{2}^{\delta}}{x,x,\delta^{-2}t}-\HK{\Lambda_{1}^{\delta}}{x,x,\delta^{-2}t}\right)\ddt\right|\leq C\delta^{2}\cdot\delta^{3q/2}.
\]
At large times $t>\delta^{-q}$, we use Lemma \ref{lem: HK_infinite_tail_bound}
to get 
\[
\left|\int_{\delta^{-q}}^{\infty}\left(\HK{\Lambda_{2}^{\delta}}{x,x,\delta^{-2}t}-\HK{\Lambda_{1}^{\delta}}{x,x,\delta^{-2}t}\right)\ddt\right|\leq2C\delta^{2}\int_{\delta^{-q}}^{\infty}t^{-2}dt\leq2C\delta^{2+q}.
\]
Clearly, similar estimates, with right-hand side divided by $\delta^{2},$
hold for continuous heat kernels, e.g., as a consequence of convergence.
At intermediate times $\delta^{\frac{q}{2}}<t<\delta^{-q}$, we apply
Lemma \ref{lem: LCLT} to each of $\Lambda_{1,2}$ separately to get
\begin{multline*}
\int_{\delta^{q/2}}^{\delta^{-q}}\left|\delta^{-2}\HK{\Lambda_{1,2}^{\delta}}{x,x,\delta^{-2}t}-\HK{\Lambda_{1,2}}{x,x,\delta^{-2}t}\right|\ddt\\
\leq C\delta^{q}\int_{\delta^{q/2}}^{\delta^{-q}}\ddt\max\{t^{-1},1\}\leq2Cq\delta^{q}(\log\delta^{-1}+\delta^{-\frac{q}{2}})\leq\hat{C}\delta^{\frac{q}{2}}.
\end{multline*}
Combining all the estimates above yields the result.
\end{proof}
\begin{cor}
\label{cor: third_term}We have 
\[
\sum_{x\in\Od}\int_{\delta^{-2}}^{\infty}P^{\Omega_{x}^{\delta}}(x,x,t)\ddt\stackrel{\delta\to0}{\longrightarrow}\int_{\Omega}\int_{1}^{\infty}P^{\Omega_{x}}(x,x,t)\ddt.
\]
\end{cor}

\begin{proof}
As in the proof of Lemma \ref{lem: CLT_convergence}, we change the
variable $t\to\delta^{-2}t$. We then use Lemma \ref{prop: HK_bound}
ensure that the sum and the integral over the $\eta_{0}$-neighborhood
of the boundary and singularities tends to zero with $\eta_{0},$
uniformly in $\delta.$ This allows to consider the sum over the complement
of that neighborhood only. In that region, we argue as in the proof
of Corollary \ref{cor: int_convergence} that the summand converges
to the integrand uniformly: apply Lemma \ref{lem: LCLT} on the integral
from $1$ to $\delta^{-q}$ and Lemma \ref{lem: HK_infinite_tail_bound}
to the integral from $\delta^{-q}$ to infinity. The only difference
with the proof of Corollary \ref{cor: int_convergence} is that now
we do not need to deal with small $t.$
\end{proof}

\subsection{Conical singularities}

\label{subsec: conical_sing}Let us compute the contribution of an
$r$-neighborhood of the tip of a conical singularity with angle $\alpha$
to (\ref{eq: local_contribution}). Changing the scale to $\delta_{0}$,
we see that 
\[
\sum_{x\in\Cone{\alpha,\delta}:|x|\leq r}I_{\C^{\delta}}^{\Cone{\alpha,\delta}}(x)=\sum_{x\in\Cone{\alpha,\delta_{0}}:|x|\leq rN}I_{\C^{\delta_{0}}}^{\Cone{\alpha,\delta_{0}}}(x),
\]
where $|\cdot|$ denotes the distance to the tip; that is, decreasing
$\delta$ for a fixed $r$ simply results in adding new terms to the
sum. The asymptotics of those terms as $|x|\to\infty$ can be read
off Corollary \ref{cor: int_convergence}: if $rN\leq|x|<r(N+1),$
then 
\[
I_{\C^{\delta_{0}}}^{\Cone{\alpha,\delta_{0}}}(x)=I_{\C^{\delta}}^{\Cone{\alpha,\delta}}\left(\frac{x}{N}\right)=\delta^{2}\cdot I_{\C}^{\Cone{\alpha}}\left(\frac{x}{N}\right)+O(\delta^{2+\rho})=I_{\C}^{\Cone{\alpha}}\left(1\right)\cdot\delta_{0}^{2}|x|^{-2}+O(|x|^{-2-\rho}),
\]
where $1$ is any point at distance $1$ from the tip, and we used
rotational invariance of $I_{\C}^{\Cone{\alpha}}(ax),$ Brownian scaling
$I_{\C}^{\Cone{\alpha}}(ax)=a^{-2}I_{\C}^{\Cone{\alpha}}(x)$, and
the relation $\delta N=\delta_{0}.$ The error term $O(|x|^{-2-\rho})$
sums to a constant over $\Cone{\alpha,\delta_{0}},$ and, recalling
that $\C^{\delta_{0}}$ has $\delta_{0}^{-2}$ vertices per unit area,
we have 
\[
\sum_{x\in\Cone{\alpha,\delta_{0}}:|x|\leq rN}\delta_{0}^{2}|x|^{-2}=\int_{x\in\Cone{\alpha}:1\leq|x|\leq rN}|x|^{-2}+\check{D}_{\alpha}+o(1)=\alpha\cdot\log(rN)+\check{D}_{\alpha}+o(1).
\]
Taking into account that $\log(rN)=\log r+\log\delta_{0}-\log\delta$,
we conclude
\begin{equation}
\sum_{x\in\Cone{\alpha,\delta_{0}}:|x|\leq rN}I_{\C^{\delta_{0}}}^{\Cone{\alpha,\delta_{0}}}(x)=-\alpha\cdot I_{\C}^{\Cone{\alpha}}(1)\cdot\log\delta+\hat{\ConstD}_{\alpha}+o(1),\label{eq: cone_contr}
\end{equation}
 where $\hat{D}_{\alpha}$ is a (lattice-dependent) constant. 

\subsection{Boundary segments}

\label{subsec: bdry_segments}Let $l\subset\text{\ensuremath{\partial}\ensuremath{\ensuremath{\Omega}}}$
be a side of a triangle or a square comprising $\Omega;$ we introduce
a local coordinate in which $l$ is identified with $(0;1)\subset\partial\H$.
Let $l^{\delta}$ be the corresponding segment of $\partial\Od$.
Let $\alpha_{0,1}$ be the angles of the wedges at its endpoints $0$
and 1, and denote $\hat{\alpha}_{0,1}:=\min\{\alpha_{0,1}/2;\pi/2\}$.
We consider the contribution to (\ref{eq: local_contribution}) of
the points that are at distance at most $r$ from $l^{\delta}$, but
at the distance greater than $r_{\alpha_{0,1}}=r/\sin(\hat{\alpha}_{0,1})$
from its endpoints $0$ and $1$, respectively. This contribution
reads 
\begin{equation}
\sum_{x\in R_{r}^{\delta}}I_{\C^{\delta}}^{\H^{\delta}}(x)-\sum_{x\in\Gamma_{0}^{\delta}}I_{\C^{\delta}}^{\H^{\delta}}(x)-\sum_{x\in\Gamma_{1}^{\delta}}I_{\C^{\delta}}^{\H^{\delta}}(x)-\sum_{x\in\HalfCorn_{0}^{\delta}}I_{\C^{\delta}}^{\H^{\delta}}(x)-\sum_{x\in\HalfCorn_{1}^{\delta}}I_{\C^{\delta}}^{\H^{\delta}}(x),\label{eq: side_breakdown}
\end{equation}
 where 
\begin{align*}
R_{r}^{\delta} & =\{x\in\H^{\delta}:\im x\leq r;0\leq\re x<1\};\\
Y_{0}^{\delta} & =\{x\in\H^{\delta}:|x|<r_{\alpha_{0}};0<\arg x\leq\hat{\alpha}_{0}\}; & Y_{1}^{\delta} & =\{x\in\H^{\delta}:|x-1|<r_{\alpha_{1}};\pi-\hat{\alpha}_{1}\leq\arg(x-1)<\pi\};\\
\Gamma_{0}^{\delta} & =\{x\in\H^{\delta}:\im x<r;\hat{\alpha}_{0}<\arg x\leq\pi/2\}; & \Gamma_{1}^{\delta} & =\{x\in\H^{\delta}:\im x<r;\pi/2<\arg(x-1)\leq\pi-\hat{\alpha}_{1}\},
\end{align*}
 see Figure \ref{fig: boundary}; 
\begin{figure}

\centering{}\includegraphics[width=0.7\textwidth]{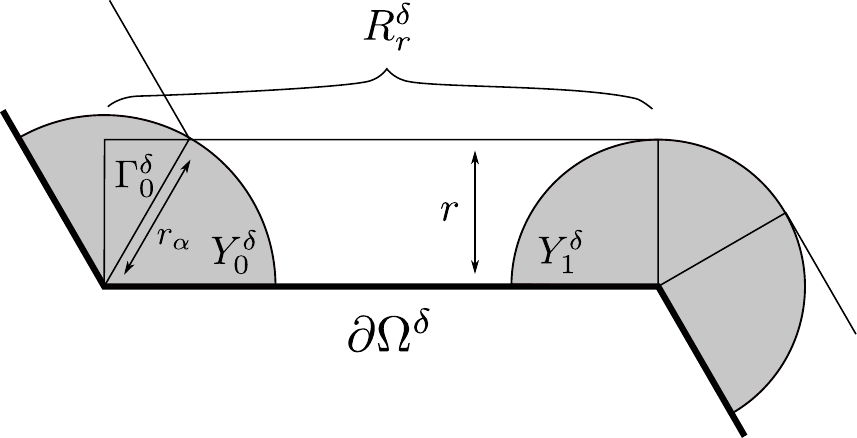}\caption{\label{fig: boundary}A decomposition of a neighborhood of a boundary
segment. The rectangle $R_{r}^{\delta}$ includes two sectors $Y_{0,1}^{\delta}$
and a triangle $\Gamma_{0}^{\delta};$ in this case, $\Gamma_{1}^{\delta}=\emptyset$
since the corresponding angle is greater than $\pi$. The shaded sectors
of radii $r_{\alpha}=r/\sin\frac{\alpha}{2}$ for $\alpha=2\pi/3$
and $r_{\alpha}=r$ for $\alpha=4\pi/3>\pi$ are the regions for which
$\protect\Od_{x}$ is a wedge; in the white part of $R_{r}^{\delta},$
$\Omega_{x}^{\delta}=\protect\H^{\delta}.$}
\end{figure}
the boundary conditions in $\H^{\delta}$ above are inherited from
$l^{\delta}$. We first treat the sum over $R_{r}^{\delta},$ which
we can split into $N=\delta_{0}\delta^{-1}$ strips $R_{r}^{\delta}(k):=\{\im x\leq r,\frac{k}{N}\leq\re x<\frac{k+1}{N}\}$
that all give equal contributions. As in the cone case, we see that
decreasing $\delta$ is tantamount to adding new terms to the sum
over $R_{r}^{\delta}(0)$, i.e.,
\begin{multline*}
\sum_{x\in R_{r}^{\delta}(0)}I_{\C^{\delta}}^{\H^{\delta}}(x)=\sum_{x\in R_{rN}^{\delta_{0}}(0)}I_{\C^{\delta_{0}}}^{\H^{\delta_{0}}}(x)=\hat{B}-I_{\C}^{\H}(i)\cdot\sum_{x\in R_{\infty}^{\delta_{0}}(0)\setminus R_{rN}^{\delta_{0}}(0)}\delta_{0}^{2}\left((\im x)^{-2}+O((\im x)^{-2-\rho})\right)\\
=\hat{B}-I_{\C}^{\H}(i)\cdot\left((rN)^{-1}+O(N^{-1-\rho})\right)
\end{multline*}
where $\hat{B}=\sum_{x\in R_{\infty}^{\delta_{0}}(0)}I_{\C^{\delta_{0}}}^{\H^{\delta_{0}}}(x)$
is a lattice-dependent constant, so that 
\[
\sum_{x\in R_{r}^{\delta}}I_{\C^{\delta}}^{\H^{\delta}}(x)=\hat{B}\cdot N-I_{\C}^{\H}(i)r^{-1}+O(\delta^{\rho}),
\]
where we can compute, with the sign $s=\pm1$ depending on the boundary
conditions as $s_{\Neu}=+1$ and $s_{\Dir}=-1$, 
\[
I_{\C}^{\H}(i):=\int_{0}^{\infty}\left(\HK{\H}{i,i,t}-\HK{\C}{i,i,t}\right)\ddt=s\cdot\int_{0}^{\infty}\HK{\C}{i,-i,t}\ddt=s\cdot\int_{0}^{\infty}\frac{1}{2\pi t}e^{-\frac{2}{t}}\ddt=s\cdot\frac{1}{4\pi},
\]
The contributions of $Y_{0,1}^{\delta}$ to (\ref{eq: side_breakdown})
will cancel the corresponding contributions from corners, thus we
will leave them as they are for a while. The contribution of $\Gamma_{0,1}^{\delta}$
is computed as in the cone case, applying Corollary \ref{cor: int_convergence}
and then using Brownian scaling and shift invariance of $I_{\C}^{\H}(i):$
\[
\sum_{y\in\Gamma_{0}^{\delta}}I_{\C^{\delta}}^{\H^{\delta}}(y)=s\cdot\sum_{\substack{0\leq\Im x\leq rN;\\
\hat{\alpha}_{0}<\arg x<\pi/2
}
}(\im x)^{-2}\delta_{0}^{2}\left(I_{\C}^{\H}(i)+O(\im x^{-\rho})\right)=s\cdot\frac{\cot(\hat{\alpha}_{0})}{4\pi}\log\delta+\hat{D}_{\alpha_{0}}+o(1),
\]
 and similarly 
\[
\sum_{y\in\Gamma_{1}^{\delta}}I_{\C^{\delta}}^{\H^{\delta}}(y)=s\cdot\frac{\cot(\hat{\alpha}_{1})}{4\pi}\log\delta+\hat{D}_{\alpha_{1}}+o(1).
\]

\subsection{Boundary corners}

\label{subsec: bdry_corners}We parameterize a boundary corner $\Upsilon^{\alpha}$
by a local coordinate $z$ so that $\Upsilon^{\alpha}=\{z\in\C:0<\arg z<\alpha\},$
and denote we denote by $\HalfCorn_{\Left}^{\delta}$ (respectively,
$Y_{\Right}^{\delta}$) the set $\HalfCorn_{1}^{\delta}$ (respectively,
$Y_{0}^{\delta}$) corresponding to the boundary segment adjacent
to $\Upsilon^{\alpha}$ on the left (respectively, on the right).
We also denote $Y_{\mathrm{middle}}^{\delta}:=\{x\in\Corner^{\alpha,\delta}:|x|\leq r_{\alpha}\}\setminus\left(Y_{\Left}^{\delta}\cup Y_{\Right}^{\delta}\right)$,
which is non-empty if and only if $\alpha>\pi$. The contribution
of $\Corner^{\alpha}$ to (\ref{eq: local_contribution}) can be written
as 
\begin{multline}
\sum_{x:|x|<r_{\alpha}}I_{\C^{\delta}}^{\Upsilon^{\alpha,\delta}}(x)=\sum_{x\in Y_{\mathrm{middle}}^{\delta}}I_{\C^{\delta}}^{\Upsilon^{\alpha,\delta}}(x)+\sum_{x\in\HalfCorn_{\Right}^{\delta}}I_{\H^{\delta}}^{\Upsilon^{\alpha,\delta}}(x)+\sum_{x\in\HalfCorn_{\Left}^{\delta}}I_{\hat{\H}^{\delta}}^{\Upsilon^{\alpha,\delta}}(x)\\
+\sum_{x\in\HalfCorn_{\Right}^{\delta}}I_{\C^{\delta}}^{\H^{\delta}}(x)+\sum_{x\in\HalfCorn_{\Left}^{\delta}}I_{\C^{\delta}}^{\hat{\H}^{\delta}}(x).\label{eq: corner_breakdown}
\end{multline}
 where $\hat{\H}^{\delta}$ stands for the upper-half plane $\H^{\delta}$
rotated counterclockwise by $\alpha-\pi$ around $\pi$ (so that its
boundary coincides with the \emph{left} boundary of the corner), and
the boundary conditions $\Upsilon^{\alpha,\delta},\H^{\delta},\hat{\H}^{\delta}$
are inherited from those in $\Omega^{\delta}$. The first three terms
yield, similarly to the computations above, 
\[
-\hat{C}_{\alpha}^{b\hat{b}}\cdot\log\delta+\hat{D}_{\alpha}^{b\hat{b}}+o(1),
\]
where 
\begin{equation}
\hat{C}_{\alpha}^{b\hat{b}}=\int_{0}^{\hat{\alpha}}I_{\H_{b}}^{\Upsilon^{\alpha}}(e^{i\theta})\,d\theta+\int_{\alpha-\hat{\alpha}}^{\alpha}I_{\hat{\H}_{\hat{b}}}^{\Upsilon^{\alpha}}(e^{i\theta})\,d\theta+\mathbf{1}_{\alpha>\pi}\int_{\pi/2}^{\alpha-\pi/2}I_{\C}^{\Upsilon^{\alpha}}(e^{i\theta})\,d\theta,\label{eq: C_corner_before}
\end{equation}
$\hat{D}_{\alpha}^{b\hat{b}}$ are constants and $b,\hat{b}\in\{\Dir,\Neu\}$
are boundary conditions on $\H,\hat{\H}.$

Observe that when collecting the contributions to (\ref{eq: local_contribution})
along $\partial\Od$, the last two terms in (\ref{eq: corner_breakdown})
cancel out the corresponding terms in (\ref{eq: side_breakdown}).
The total contribution of the $r$-neighborhood of $\partial\Od$
to (\ref{eq: local_contribution}) is therefore 
\begin{multline}
\hat{B}_{\Dir}\cdot N\cdot|\partial_{\Dir}\Omega|+\hat{B}_{\Neu}\cdot N\cdot|\partial_{\Neu}\Omega|-I_{\C}^{\H_{\Dir}}(i)r^{-1}|\partial_{\Dir}\Omega|-I_{\C}^{\H_{\Neu}}(i)r^{-1}|\partial_{\Neu}\Omega|\\
-\left(\sum_{p\in\Corner}C_{p}\right)\cdot\log\delta+\sum_{i\in\Crnrs}\hat{D}_{\alpha_{i}}^{b_{i}b_{i+1}}+o(1),\label{eq: bdry_contr}
\end{multline}
where $C_{p}$ depends only on the angle and the boundary conditions
$b,\hat{b}\in\{\Dir,\Neu\}$ on the segments adjacent to the $p\simeq\Corner_{b\hat{b}}^{\alpha}$
as 
\begin{equation}
C_{\Corner_{b\hat{b}}^{\alpha}}=\hat{C}_{\alpha}^{b\hat{b}}-(s_{b}+s_{\hat{b}})\cdot\frac{\cot(\min\{\frac{\alpha}{2};\frac{\pi}{2}\})}{4\pi}.\label{eq: C_corner_implicit}
\end{equation}
Taking into account the discussion at the end of Section \ref{sec: the Key formula},
we have
\begin{equation}
B_{\Dir,\Neu}=-d\cdot\left(\hat{B}_{\Dir,\Neu}\mp\frac{1}{2}\sum_{x\in T^{\delta_{0}}\cap[0,1)}A_{x}\right),\label{eq: B}
\end{equation}
where the sum is over the vertices in one boundary edge of the fundamental
domain $T^{\delta_{0}}.$

\subsection{Punctures}

\label{subsec: punctures}Using a suitable modification of Lemma \ref{lem: LCLT},
similarly to the conical singularity case, we have, in the local coordinate
where $p=0,$ 
\begin{multline}
\sum_{x\in\C^{\delta}:|x|\leq r}I_{\C^{\delta}}^{\C^{\delta}\setminus\{0\},\varphi_{p}}(x)=I_{\C}^{\C\setminus\{0\},\varphi_{p}}(1)\cdot\sum_{x\in\C^{\delta_{0}}:|x|\leq rN}\delta_{0}^{2}\left(|x|^{-2}+O(|x|^{-2-\rho})\right)\\
=-2\pi I_{\C}^{\C\setminus\{0\},\varphi_{p}}(1)\log\delta+\ConstD_{\varphi_{p}}+o(1).\label{eq: puncture}
\end{multline}

\section{Explicit computations for the logarithmic term}

\label{sec: logterm_explicit}In this section, we compute the integrals
involving heat kernels that contribute to the logarithmic term of
the asymptotics. The results are not new. Namely, as pointed out in
\cite{finski2020spanning,greenblatt2021discrete}, the constant $C$
in (\ref{eq: main}) is related to the spectral zeta-function by $C=-2\zeta_{\Omega}(0),$
see Remark~\ref{rem: zeta_0}. The value of $\zeta_{\Omega}(0)$
has been computed in a much greater generality by Cheeger, see \cite[Theorem 4.4]{cheeger1983spectral}
and the discussion thereafter. 

Here, we propose an alternative computation based on the following
identity for the heat kernel on the universal cover of a punctured
plane:
\begin{lem}
We have, for the heat kernel $\tilde{P}:=P^{\widetilde{\C\setminus\{0\}}}$,
\[
\int_{0}^{\infty}\tilde{P}(1,e^{i\alpha},t)\cdot\ddt=\frac{1}{\pi\alpha^{2}}.
\]
 
\end{lem}

\begin{proof}
We use the Brownian loop measure of Lawler and Werner, see \cite{lawler2004brownian}
or \cite[Section 5.6]{Lawler_book}, defined as a $\sigma$-finite
measure on the space of unrooted closed loops in a Riemann surface
$\Lambda$ by 
\begin{equation}
\mu_{\Lambda}=\int_{0}^{\infty}\frac{1}{t}P^{\Lambda}(z,z,t)\mu_{\Lambda,z,t}^{\sharp}\,|dz|^{2}dt,\label{eq: loop_measure}
\end{equation}
where $P^{\Lambda}$ is the heat kernel in $\Lambda$ with Dirichlet
boundary conditions, and $\mu_{\Lambda,z,t}^{\sharp}$ is the Brownian
probability measure on the paths from $z$ to $z$ of duration $t$.
We will only need the conformal invariance of this measure, see \cite[Proposition 6]{lawler2004brownian}
which is stated for planar domains, but the proof, being a local computation,
extends verbatim to Riemann surfaces. Consider the annular region
$\mathcal{A}_{r,\alpha}=\{r\leq|z|\leq1\}/\{z\sim e^{i\alpha}z\}$
in the cone of angle $\alpha$, $\tilde{\mathcal{A}}_{r,\alpha}$
its universal cover, and let $E$ denote the set of loops in $\mathcal{A}_{r,\alpha}$
that wind around the annulus once counterclockwise. The map $\phi:z\mapsto-i\log z$
maps $\text{\ensuremath{\mathcal{A}_{r,\alpha}}}$ onto the cylinder
$\mathcal{O}_{\alpha,r}=\mathcal{S}_{r}/\{z\sim z+\alpha\},$ where
$\mathcal{S}_{r}=\{0\leq\im z\leq-\log r\},$ so, by the conformal
invariance, we have 
\begin{multline}
\int_{\mathcal{A}_{r,\alpha}}\int_{0}^{\infty}P^{\mathcal{\tilde{A}}_{r,\alpha}}(z,ze^{i\alpha},t)|dz|^{2}\ddt=\mu_{\mathcal{A}_{r,\alpha}}(E)=\\
=\mu_{\mathcal{O}_{\alpha,r}}(\phi(E))=\int_{\substack{\{0\leq\re z<\alpha\}\cap S_{r}}
}\int_{0}^{\infty}P^{\mathcal{S}_{r}}(z,z+\alpha,t)|dz|^{2}\ddt.\label{eq: univ_cover_proof}
\end{multline}
Note that by scaling invariance, the total $\mu_{\Cone{\alpha}}$
measure of the loops that wind around $\Cone{\alpha}$ and intersect
a given circle $|z|=r$ does not depend on $r$. We claim that it
is also finite. Indeed, by conformal invariance, we can pass to the
cylinder $\mathcal{O}_{\alpha}=\C/\{z\sim z+\alpha\}$, when the circle
is mapped to $l_{h}:=\{w:\im w=h\},$ and then use that for some $c,C>0,$
\[
\mu_{\mathcal{O}_{\alpha},z,t}^{\sharp}(\{\gamma:\gamma\cap l_{h}\neq\emptyset)\leq Ce^{-c\frac{(\im z-h)^{2}}{t}}.
\]
Also, $P^{\mathcal{O}_{\alpha}}(z,z,t)\sim Ct^{-\frac{1}{2}}$ as
$t\to\infty$. These two bounds imply that the contribution to (\emph{\ref{eq: loop_measure}})
from the points $z$ with $|\im z-h|>1$ is finite. Since the probability
that a bridge with a small $t$ winds around $\mathcal{O}_{\alpha}$
is exponentially small, the region $|\im z-h|\leq1$ also gives a
finite contribution.

Hence, up to $O(1)$ as $r\to0$, the left-hand side of (\ref{eq: univ_cover_proof})
equals 
\[
\int_{\mathcal{A}_{r,\alpha}}\int_{0}^{\infty}\tilde{P}(z,ze^{i\alpha},t)\ddt=\int_{\mathcal{A}_{r,\alpha}}|z|^{-2}\int_{0}^{\infty}\tilde{P}(1,e^{i\alpha},t)\ddt=-\alpha\log r\int_{0}^{\infty}\tilde{P}(1,e^{i\alpha},t)\ddt.
\]
We conclude by comparing this to the right-hand side of (\ref{eq: univ_cover_proof}),
which is, up to $O(1),$
\[
\int_{\substack{\{0\leq\re z<\alpha\}\cap S_{r}}
}\int_{0}^{\infty}P^{\C}(z,z+\alpha,t)\ddt=-\alpha\log r\int_{0}^{\infty}\frac{1}{2\pi t}e^{-\frac{\alpha^{2}}{2t}}\ddt=-\frac{\log r}{\pi\alpha}.
\]
\end{proof}
Since $\widetilde{\C\setminus\{0\}}$ also covers each of the cones
$\Cone{\alpha}\backsimeq\C/\{z\sim e^{i\alpha}z\}$, we have 
\[
\HK{\Cone{\alpha}}{x,y,t}=\sum_{k\in\Z}\tilde{P}(x,ye^{ik\alpha},t).
\]
We now can compute, using that $\sum_{k=1}^{\infty}1/k^{2}=\frac{\pi^{2}}{6}$,
\begin{multline}
I_{\C}^{\Cone{\alpha}}(1)=\int_{0}^{\infty}\left(\sum_{k\in\Z}\tilde{P}(1,e^{\alpha ik},t)-\tilde{P}(1,e^{2\pi ik},t)\right)\ddt\\
=\sum_{k\in\Z\setminus\{0\}}\left(\frac{1}{\pi(\alpha k)^{2}}-\frac{1}{\pi(2\pi k)^{2}}\right)=\frac{\pi}{3\alpha^{2}}-\frac{1}{12\pi}
\end{multline}
By the reflection principle, we have $\HK{\Corner^{\alpha}}{x,y,t}=\HK{\Cone{2\alpha}}{x,y,t}\pm\HK{\Cone{2\alpha}}{x,\bar{y},t}$
for $\Corner^{\alpha}=\Corner_{\Neu}^{\alpha}$ and $\Corner^{\alpha}=\Corner_{\Dir}^{\alpha}$
respectively, and hence, using that $\sum_{k\in\Z}(x+k)^{-2}=\pi^{2}\sin^{-2}\pi x$,
we get 
\begin{multline}
\int_{0}^{\frac{\alpha}{2}}I_{\H}^{\Corner^{\alpha}}(e^{i\theta})d\theta=\int_{0}^{\frac{\alpha}{2}}\sum_{k\in\Z\setminus\{0\}}\left(\frac{1}{\pi(2\alpha k)^{2}}-\frac{1}{\pi(2\pi k)^{2}}\pm\frac{1}{\pi(2\theta+2\alpha k)^{2}}\mp\frac{1}{\pi(2\theta+2\pi k)^{2}}\right)d\theta\\
=\frac{\pi}{24\alpha}-\frac{\alpha}{24\pi}\pm\frac{\pi}{4\alpha^{2}}\int_{0}^{\alpha/2}\left(\sin^{-2}\left(\frac{\pi}{\alpha}\theta\right)-\frac{\alpha^{2}}{\pi^{2}\theta^{2}}\right)d\theta\mp\frac{1}{4\pi}\int_{0}^{\alpha/2}\left(\sin^{-2}\theta-\frac{1}{\theta^{2}}\right)d\theta\\
=\frac{\pi}{24\alpha}-\frac{\alpha}{24\pi}\pm\frac{1}{4\pi}\cot\frac{\alpha}{2}.\label{eq: C_corner_D_or_N}
\end{multline}
Using that, by reflection principle applied to the line $\arg z=\alpha$,
we have 
\begin{eqnarray}
\HK{\Corner_{\Neu\Dir}^{\alpha}}{x,y,t} & = & \HK{\CornerD^{2\alpha}}{x,y,t}+\HK{\CornerD^{2\alpha}}{x,e^{2i\alpha}\bar{y},t};\label{eq:refl3}\\
\HK{\Corner_{\Dir\Neu}^{\alpha}}{x,y,t} & = & \HK{\Corner_{\Neu}^{2\alpha}}{x,y,t}-\HK{\Corner_{\Neu}^{2\alpha}}{x,e^{2i\alpha}\bar{y},t},
\end{eqnarray}
a similar straightforward but tedious computation yields 
\begin{equation}
\int_{0}^{\frac{\alpha}{2}}I_{\H_{\Neu}}^{\Corner_{\Dir\Neu}^{\alpha}}(e^{i\theta})d\theta=-\frac{\pi}{48\alpha}-\frac{\alpha}{24\pi}-\frac{1}{4\alpha}+\frac{1}{4\pi}\cot\frac{\alpha}{2};\label{eq: DN1}
\end{equation}
\begin{equation}
\int_{0}^{\frac{\alpha}{2}}I_{\H_{\Dir}}^{\Corner_{\Neu\Dir}^{\alpha}}(e^{i\theta})d\theta=-\frac{\pi}{48\alpha}-\frac{\alpha}{24\pi}+\frac{1}{4\alpha}-\frac{1}{4\pi}\cot\frac{\alpha}{2}.\label{eq: DN2}
\end{equation}
Now we are ready to collect the values of $C_{\Corner^{\alpha}}$
for $\alpha\leq\pi$: $\hat{C}_{\alpha}^{\Dir\Dir}$ and $\hat{C}_{\alpha}^{\Neu\Neu}$
consist of two equal terms given by (\ref{eq: C_corner_D_or_N}),
with the cotangent terms canceling out the corresponding terms in
(\ref{eq: C_corner_implicit}), while $\hat{C}_{\alpha}^{\Dir\Neu}$
is given by the sum of (\ref{eq: DN1}) and (\ref{eq: DN2}) above.
If $\alpha>2\pi$, then we need, in addition, to compute the contribution
of the third term in (\ref{eq: C_corner_before}) which is done similarly,
and change the integration limits in (\ref{eq: C_corner_D_or_N}),
(\ref{eq: DN1}\textendash \ref{eq: DN2}) to $\frac{\pi}{2};$ we
leave it to the reader to check that the answer is (unsurprisingly)
given by the same analytic expression in $\alpha$.

In the case of a puncture, we can compute 
\begin{multline*}
I_{\C}^{\C\setminus\{0\},\varphi_{p}}(1)=\int_{0}^{\infty}\left(\sum_{k\in\Z}^{\infty}\Tr M^{k}\tilde{P}(1,e^{2\pi ik},t)-d\cdot\sum_{k\in\Z}\tilde{P}(1,e^{2\pi ik},t)\right)\ddt\\
=\sum_{k\in\Z\setminus\{0\}}^{\infty}\left(\Tr M^{k}-d\right)\frac{1}{\pi(2\pi k)^{2}}=\frac{1}{2\pi}\frac{1}{\pi^{2}}\sum_{k=1}^{\infty}\left(\re\Tr M^{k}-d\right)\frac{1}{k^{2}}.
\end{multline*}

\section{Explicit computation of the constant $B$ for some lattices}

\begin{figure}

\includegraphics[width=0.8\textwidth]{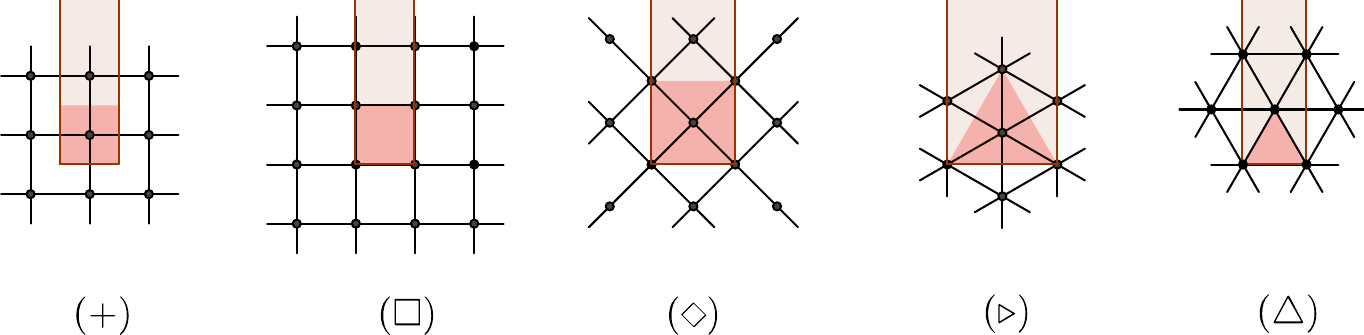}\caption{\label{fig:NiceLattices}The lattices considered in Section \ref{sec: constantsB},
their fundamental domains and half-strips $R_{\infty}^{\delta_{0}}$}
\end{figure}
\label{sec: constantsB}In this section, we show how the exact values
of $B_{\Neu}=-B_{\Dir}$ can be computed for several nice lattices,
see Figure \ref{fig:NiceLattices}. We assume $d=1,$ since, as remarked
before, the constants $A,B_{\Dir},B_{\Neu}$ depend linearly on $d.$
We note that the values $B_{\Neu}^{\LatticeA}$ and $B_{\Dir}^{\LatticeB}$
are related to each other by planar UST duality, see \cite[Section 5]{DuplantierDavid},
therefore, in fact, each of $B_{\Neu}^{\LatticeA},B_{\Neu}^{\LatticeB},B_{\Dir}^{\LatticeA},B_{\Dir}^{\LatticeB}$
can be deduced from the results in \cite{DuplantierDavid}. Similarly,
it can be deduced from duality that $B_{\Neu}^{\LatticeC}=B_{\Dir}^{\LatticeC},$
and hence $B_{\Neu}^{\LatticeC}=B_{\Dir}^{\LatticeC}=0.$ Here, we
give a direct self-contained computation.

For each of those lattices, the constant $\hat{B}_{\Neu}=-\hat{B}_{\Dir}$
can be expressed in the form
\[
\hat{B}_{\Neu}=\sum_{x\in R_{\infty}^{\delta_{0}}(0)}I_{\C^{\delta_{0}}}^{\H^{\delta_{0}}}(x)=\sum_{x\in R_{\infty}^{\delta_{0}}(0)}\int_{0}^{\infty}P^{\mathbb{C}^{\delta_{0}}}(x,\bar{x},t)\,\frac{dt}{t}=\int_{0}^{\infty}\sum_{y\in S}P^{\mathbb{C}^{\delta_{0}}}(0,y,t)\,\frac{dt}{t},
\]
where $S:=\{\bar{x}-x\mid x\in R_{\infty}^{\delta_{0}}(0)\}$, and
we shift the lattice ($+$) so that it has a vertex at the origin.
For the lattice $(\LatticeA)$ (resp. $(\LatticeB)$), we have $S=\{v\in\mathbb{C}^{\delta_{0}}:\re v=0,\im v<0\text{ odd}\}$
(resp., $S=\{v\in\mathbb{C}^{\delta_{0}}:\re v=0,\im v<0\text{ even}\}$).
For the lattices $(\LatticeC)$, $(\LatticeD)$, $(\LatticeE)$, we
have $S=\{v\in\mathbb{C}^{\delta_{0}}:\re v=0,\im v<0\}.$

Denote $(X_{t},Y_{t}):=\gamma_{t}^{\mathbb{C}^{\delta_{0}}}$, started
at the origin. Since $P^{\mathbb{C}^{\delta_{0}}}(0,y,t)=P^{\mathbb{C}^{\delta_{0}}}(0,\bar{y},t)$,
one has 
\begin{equation}
2\hat{B}_{\Neu}^{\ast}=\begin{cases}
\int_{0}^{\infty}\P(X_{t}=0\text{ and }Y_{t}\text{ is odd})\,\frac{dt}{t}, & \text{ \ensuremath{\ast=\LatticeA,}}\\
\int_{0}^{\infty}\left(\P(X_{t}=0\text{ and }Y_{t}\text{ is even})-\P(X_{t}=Y_{t}=0)\right)\,\frac{dt}{t}, & \text{ \ensuremath{\ast=\LatticeB},}\\
\int_{0}^{\infty}\left(\P(X_{t}=0)-\P(X_{t}=Y_{t}=0)\right)\,\frac{dt}{t}, & \text{ \ensuremath{\ast=\LatticeC,\LatticeD,\LatticeE}\ensuremath{.}}
\end{cases}\label{eq: 2BNhat}
\end{equation}
Let $\rwZ_{t}$ be the continuous time symmetric random walk on $\mathbb{Z}$
that moves with intensity $1$. Note that in the cases $*=\LatticeA,\LatticeB,\LatticeC,\LatticeD,$
we have $X_{t}\stackrel{\mathcal{D}}{=}\alpha^{\ast}\rwZ_{\beta^{\ast}t},$
with $\beta^{\LatticeA}=\beta^{\LatticeB}=1,$ $\beta^{\LatticeC}=2$,
$\beta^{\LatticeD}=\frac{2}{\sqrt{3}}.$ In the cases $*=\LatticeA$
and $*=\LatticeB$, $X_{t}$ and $Y_{t}$ are independent. We compute
$\P(\rwZ_{t}\text{ is even})=\frac{1}{2}(1+e^{-2t})=:K^{\Z_{2}}(t)$
and $\P(\rwZ_{t}=0)=e^{-t}\I(t)=:K^{\Z}(t)$ where $\I(t)=\frac{1}{\pi}\int_{0}^{\pi}e^{t\cos\phi}\,d\phi$
is the modified Bessel function of the first kind (see \cite{CJK1}
or the computation in the end of this section). Recall from Section
\ref{sec: the Key formula} that in all cases, $\int_{0}^{\infty}\left(\P(X_{t}=Y_{t}=0)-e^{-t}\right)\,\frac{dt}{t}=-A$~. 

Below we will use the formula 
\begin{equation}
\int_{0}^{\infty}\left(e^{-at}I_{0}(t)-e^{-bt}\right)\,\frac{dt}{t}=\log(a-\sqrt{a^{2}-1})+\log2b\text{ for }a\geq1,b>0,\label{eq: Bessel_Laplace}
\end{equation}
which follows from the formula for the Laplace transform of the Bessel
function, $\int_{0}^{\infty}e^{-at}\I(t)\,dt=1/\sqrt{a^{2}-1},$ $a>1$,
by taking the derivative of (\ref{eq: Bessel_Laplace}) with respect
to $a$ and matching the behavior as $a\to\infty.$

With these observations, we are ready to evaluate the right-hand side
of (\ref{eq: 2BNhat}) for $*=\LatticeA,\LatticeB,\LatticeC,\LatticeD$:

\begin{multline*}
\hat{B}_{\Neu}^{\LatticeA}=\frac{1}{2}\int_{0}^{\infty}K^{\Z}(t)(1-K^{\Z_{2}}(t))\,\frac{dt}{t}\\
=\frac{1}{4}\left[\int_{0}^{\infty}\left(e^{-t}I_{0}(t)-e^{-t}\right)\,\frac{dt}{t}-\int_{0}^{\infty}\left(e^{-3t}I_{0}(t)-e^{-t}\right)\,\frac{dt}{t}\right]=\frac{1}{4}\log(3-\sqrt{8})=\frac{1}{2}\log(\sqrt{2}-1),
\end{multline*}

\begin{multline*}
\hat{B}_{\Neu}^{\LatticeB}=\frac{1}{2}\int_{0}^{\infty}K^{\Z}(t)\left(K^{\Z_{2}}(t)-K^{\Z}(t)\right)\,\frac{dt}{t}\\
=\frac{1}{2}\left[\underbrace{\int_{0}^{\infty}\left(e^{-t}I_{0}(t)-e^{-t}\right)\,\frac{dt}{t}}_{=\log2}-\underbrace{\int_{0}^{\infty}[K^{\Z}(t)]^{2}-e^{-t}\,\frac{dt}{t}}_{=-A}-\underbrace{\int_{0}^{\infty}K^{\Z}(t)(1-K^{\Z_{2}}(t))\,\frac{dt}{t}}_{=\log(\sqrt{2}-1)}\right],
\end{multline*}
\begin{multline*}
\hat{B}_{\Neu}^{\LatticeC}=\frac{1}{2}\int_{0}^{\infty}\left(K^{\Z}(2t)-\P(X_{t}=Y_{t}=0)\right)\,\frac{dt}{t}\\
=\frac{1}{2}\left[\underbrace{\int_{0}^{\infty}\left(e^{-2t}I_{0}(2t)-e^{-t}\right)\,\frac{dt}{t}}_{=0}-\underbrace{\int_{0}^{\infty}\left(\P(X_{t}=Y_{t}=0)-e^{-t}\right)\,\frac{dt}{t}}_{=-A}\right]
\end{multline*}
\begin{multline*}
\hat{B}_{\Neu}^{\LatticeD}=\frac{1}{2}\int_{0}^{\infty}\left(K^{\Z}(2/\sqrt{3}\cdot t)-\P(X_{t}=Y_{t}=0)\right)\,\frac{dt}{t}\\
=\frac{1}{2}\left[\underbrace{\int_{0}^{\infty}\left(e^{-2/\sqrt{3}\cdot t}I_{0}(2/\sqrt{3}\cdot t)-e^{-t}\right)\,\frac{dt}{t}}_{=\log\sqrt{3}}-\underbrace{\int_{0}^{\infty}\left(\P(X_{t}=Y_{t}=0)-e^{-t}\right)\,\frac{dt}{t}}_{=-A}\right]
\end{multline*}
and plugging into (\ref{eq: B}) gives 
\begin{align*}
B_{\mathcal{N}}^{\LatticeA} & =-B_{\mathcal{D}}^{\LatticeA}=-\frac{1}{2}\log(\sqrt{2}-1);\\
B_{\mathcal{N}}^{\LatticeB} & =-B_{\mathcal{D}}^{\LatticeB}=\frac{1}{2}(\log(\sqrt{2}-1)-\log2);\\
B_{\mathcal{N}}^{\LatticeC} & =-B_{\mathcal{D}}^{\LatticeC}=0\\
B_{\mathcal{N}}^{\LatticeD} & =-B_{\mathcal{D}}^{\LatticeD}=-\frac{1}{4}\log3.
\end{align*}

In the case $\star=\LatticeE,$ $X_{t}$ is equal in law to a scaled
copy of a continuous time walk on $\mathbb{Z}$ jumping by $-2,-1,1,2$
with intensities $\frac{\sqrt{3}}{6},\frac{2\sqrt{3}}{6},\frac{2\sqrt{3}}{6},\frac{\sqrt{3}}{6}$
respectively. Its generator is $\tilde{\Delta}=f(\Delta_{\mathbb{Z}})$,
where $f(x):=\frac{2\sqrt{3}}{3}(x^{2}+3x)$ and $\Delta_{\mathbb{Z}}$
is the generator of $\rwZ_{t}$. When considered on $(\Z/M\Z),$ both
$\Delta$ and $\tilde{\Delta}$ have eigenfunctions $\psi_{m}(x)=e^{\frac{2\pi imx}{M}}$,
and the eigenvalues are $\lambda_{m}=1-\cos\frac{2\pi im}{M}$ and
$\tilde{\lambda}_{m}=f(\lambda_{m})$, respectively. Thus, for any
$M$ we explicitly find that $\P(X_{t}\in M\mathbb{Z})=\frac{1}{M}\sum_{m=1}^{M}\exp(-tf(1-\cos\frac{2\pi im}{M}))$
and, passing to the limit as $M\to\infty$, we have 
\[
\P(X_{t}=0)=\int_{0}^{1}\exp(-tf(1-\cos{2\pi is}))\,ds,
\]
so 
\begin{multline*}
B_{\mathcal{N}}^{\LatticeE}=-B_{\mathcal{D}}^{\LatticeE}=-\frac{1}{2}\int_{0}^{\infty}\left(\P(X_{t}=0)-e^{-t}\right)\,\frac{dt}{t}\\
=-\frac{1}{2}\int_{0}^{1}\int_{0}^{\infty}\left(\exp(-tf(1-\cos{2\pi is}))-\exp(-t)\right)\,\frac{dt}{t}\,ds=\frac{1}{2}\int_{0}^{1}\log f(1-\cos{2\pi is})\,ds
\end{multline*}

\section{Proof of Theorem \ref{thm: main}}

\label{sec: proof_of_the_main}To identify the universal constant
term in (\ref{eq: main}), we derive an analog of the key formula
in the continuum. This is slightly more delicate than in the discrete
because the continuous heat kernel is more singular at $t=0$, in
particular, there is no $s$ for which the integral $\int_{0}^{\infty}P^{\C}(x,x,t)t^{s-1}dt$
converges. Therefore, we perform the analytic continuation in two
steps. Starting with (\ref{eq: zet_int_theta_cont}), we write 

\begin{align}
\Gamma(s)\zeta^{\Omega,\varphi}(s) & =\int_{1}^{\infty}\left(\Theta^{\Omega,\varphi}(t)-k\right)t^{s-1}\,dt+\int_{\Omega}\int_{0}^{1}\left(\Tr P^{\Omega,\varphi}(x,x,t)-d\cdot P^{\Omega_{x}}(x,x,t)\right)t^{s-1}\,dtdx\label{eq: regul}\\
 & +d\int_{\Omega}\int_{0}^{1}\left(P^{\Omega_{x}}(x,x,t)-P^{\C}(x,x,t)\right)t^{s-1}\,dxdt+d|\Omega|\int_{0}^{1}\frac{1}{2\pi t}t^{s-1}dt-\int_{0}^{1}kt^{s-1}\,dt.\nonumber 
\end{align}
The first two terms converge for all $s$ and rest converge when $\re s>1$.
The last two terms evaluate to $\frac{d|\Omega|}{2\pi(s-1)}-\frac{k}{s}$.
Let's have a closer look at the third term. To construct its analytic
continuation to $s=0$, we first remark that it is, in fact, already
analytic for $\re s>\frac{1}{2}.$ Indeed, for $x\in\H$, let $\rho:=\im x$;
by Brownian scaling, $P^{\H}(x,x,t)=\rho^{-2}P^{\H}(i,i,t/\rho^{2})$
and thus 
\[
\int_{0}^{1}\left|P^{\H}(x,x,t)-P^{\C}(x,x,t)\right||t^{s-1}|\,dt=(2\pi)^{-1}\rho^{2\re s-2}\int_{0}^{\rho^{-2}}e^{-\frac{2}{t}}|t^{s-2}|\,dt\leq\begin{cases}
C\rho^{2\re s-2}, & \re s<1;\\
C|\log\rho|, & \re s=1,\\
C, & \re s>1.
\end{cases}
\]
Hence, the integral in the third term in (\ref{eq: regul}) over $x:\Omega_{x}\backsimeq\H$
converges absolutely for all $s$ with $\re s>\frac{1}{2}.$ For $\Omega_{x}$
a cone, the same scaling argument leads to the same bound with $\rho$
the distance to the tip, thus, the contribution of those $x$ converges
absolutely for $\re s>0$. Finally, near a corner, we break the integral
down as in (\ref{eq: corner_breakdown}) and treat the first three
terms as in the cone case and the last two as in the half-plane case.

Now, for $\frac{1}{2}<\re s<1$, we can write 
\[
\int_{0}^{1}\left(P^{\Omega_{x}}(x,x,t)-P^{\C}(x,x,t)\right)t^{s-1}\,dt=I_{\C}^{\Omega_{x}}(x,s)-\int_{1}^{\infty}P^{\Omega_{x}}(x,x,t)t^{s-1}\,dt-\frac{1}{2\pi(s-1)},
\]
where $I_{\Lambda_{1}}^{\Lambda_{2}}(x,s):=\int_{0}^{\infty}\left(P^{\Lambda_{2}}(x,x,t)-P^{\Lambda_{1}}(x,x,t)\right)t^{s-1}\,dt,$
cf. the notation in Section \ref{sec: Local_contr}. The last two
terms give a contribution that is analytic over $\re s<1,$ hence
our task is to analytically continue $\int_{\Omega}I_{\C}^{\Omega_{x}}(x,s).$
We split it into contributions of neighborhoods of cones, boundary
segments, corners and punctures, and leverage the fact that for the
scaling on each of $\Omega_{x},$ one has $I_{\C}^{\Omega_{x}}(x,s)=a^{2-2s}I_{\C}^{\Omega_{x}}(ax,s)$.
For $\Omega_{x}\backsimeq\Cone{\alpha}\backsimeq\C/\{z\sim e^{i\alpha}z\},$
we thus get
\begin{equation}
\int_{x\in\Cone{\alpha}:|x|<r}I_{\C}^{\Cone{\alpha}}(x,s)=I_{\C}^{\Cone{\alpha}}(1,s)\cdot\int_{x\in\Cone{\alpha}:|x|<r}|x|^{2s-2}=\alpha\cdot I_{\C}^{\Cone{\alpha}}(1,s)\cdot\frac{r^{2s}}{2s},\label{eq: cone_cont_continuation}
\end{equation}
This is analytic for $\re s\leq1$ when divided by $\Gamma(s)$, which
is the only thing we care about; when we eventually evaluate the derivative
at $0$, we get some value that can be absorbed into the constant
$D_{\Cone{\alpha}}.$ For the contributions of boundary segments and
boundary corners, we split their neighborhoods as in Sections \ref{subsec: bdry_segments}
and \ref{subsec: bdry_corners}. The integral over $R_{r}$ yields
\begin{equation}
\int_{\im x\leq r,0\leq\re x\leq1}I_{\C}^{\H}(x,s)\,dx=I_{\C}^{\H}(i,s)\cdot\int_{0}^{r}\rho^{2s-2}d\rho=I_{\C}^{\H}(i,s)\cdot\frac{r^{2s-1}}{2s-1}.\label{eq: side_contrib_zeta_zero}
\end{equation}
For the other contributions, note that those of $Y_{0,1}$ cancel
out as in Sections (\ref{subsec: bdry_segments}\textendash \ref{subsec: bdry_corners}),
and other contributions can be treated as in the cone case, eventually
contributing a constant that can be absorbed into $D_{\Corner^{\alpha}};$
same applies to punctures. Dividing (\ref{eq: regul}) by $\Gamma(s)$
and differentiating at $s=0$, we get 
\begin{multline}
-\log\Det_{\zeta}\Delta^{\Omega,\varphi}=\int_{1}^{\infty}\left(\Theta^{\Omega^{\delta},\varphi}(t)-k\right)\ddt+\int_{\Omega}\int_{0}^{1}\left(\Tr P^{\Omega,\varphi}(x,x,t)-d\cdot P^{\Omega_{x}}(x,x,t)\right)\ddt dx\\
-d\int_{\Omega}\int_{1}^{\infty}P^{\Omega_{x}}(x,x,t)\ddt dx-k\EulerGamma-I_{\C}^{\H_{\Dir}}(i)r^{-1}|\partial_{\Dir}\Omega|-I_{\C}^{\H_{\Neu}}(i)r^{-1}|\partial_{\Neu}\Omega|+\sum_{p\in\Cone{}\cup\Corner\cup\Puncture}\tilde{D}_{p}.\label{eq: logdet_cont}
\end{multline}

\begin{rem}
\label{rem: zeta_0}The equation (\ref{eq: regul}) also allows one
to see that in fact, $C=-2\zeta_{\Omega}(0).$ Since $\frac{1}{\Gamma(s)}\sim s$
near the origin, $\zeta_{\Omega}(0)$ only receives the contributions
from those terms in the right-hand side of (\ref{eq: regul}) that
have a pole at the origin, that is, the third term and $-\frac{k}{s}$.
The third term is split into the contributions from near boundary
segments, i.e., where $\Omega_{x}=\H,$ and the contributions from
neighborhoods of the punctures, the cone tips and the corners. The
former is evaluated at (\ref{eq: side_contrib_zeta_zero}) and has
no pole at the origin, and the contribution from a cone tip is computed
in (\ref{eq: cone_cont_continuation}), with the residue at $s=0$
matching the coefficient found in (\ref{eq: cone_contr}). A similar
result holds for boundary corners and punctures. 
\end{rem}

We are in the position to put everything together and prove Theorem
\ref{thm: main}:
\begin{proof}[Proof of Theorem \ref{thm: main}.]
 We look at the key formula (\ref{eq: key_formula}) term by term.
The first three terms converge to the first three terms of (\ref{eq: logdet_cont})
by Lemma \ref{lem: CLT_convergence} and Corollary \ref{cor: third_term}.
The fourth term can be broken into the contributions of neighborhoods
of conical singularities, punctures, and the boundary, whose asymptotics
is given by (\ref{eq: cone_contr}), (\ref{eq: puncture}), (\ref{eq: bdry_contr}),
and the constants $C_{p}$ are made explicit in Section \ref{sec: logterm_explicit}.
The fifth term gives the ``volume'' contribution that is discussed
in the end of Section \ref{sec: the Key formula}.
\end{proof}
\begin{rem}
\label{rem: more_general}The underlying triangulation or quadrangulation
structure of $\Omega$ was only used in the discretization procedure,
but otherwise it plays no role in the proof. While we found no elegant
way to state Theorem \ref{thm: main} in a more general form that
would account for that, we give an example: let $\T_{\omega_{1}^{\delta},\omega_{2}^{\delta}}=\C^{\delta}/(\omega_{1}^{\delta}\Z+\omega_{2}^{\delta}\Z)$
be a sequence of discretized tori, whose periods $\omega_{1,2}^{\delta}$
converge as $\delta\to0$ to $\omega_{1,2}\in\C\setminus\{0\}$ with
$\omega_{2}/\omega_{1}\notin\R$. Then, the asymptotics (\ref{eq: main})
holds, and since there are no boundary or corners/cones, it takes
the form 
\[
\log\Det\Delta^{\Od,\varphi}=A\cdot|\T_{\omega_{1}^{\delta},\omega_{2}^{\delta}}|-2\dim\ker\Delta^{\Od,\varphi}\cdot\log\delta+\log\Det_{\zeta}\Delta^{\Omega,\varphi}+o(1),\quad\delta\to0.
\]
This extends the results in \cite{DuplantierDavid,CJK1,CJK2,Friedli}.
The above proof applies verbatim; also note that the symmetries of
the lattice (other than the double periodicity) are not needed here,
as long as the embedding is such that the random walk converges to
the Brownian motion. 
\end{rem}

\section{Proof of the Lemmas}

\label{sec: proof_of_Lemmas}

In preparation for the proof of Lemmas \ref{lem: symmetry}, \ref{lem: domain},
we prove the following bounds on the transition kernel $P^{\hat{\Omega}}(x,y,t)$
of for the process $\hat{\gamma}$ defined in Section \ref{sec: heat_kernel_continuum},
in particular, establishing the existence of this transition kernel.
\begin{lem}
The transition kernel $P^{\hat{\Omega}}(x,y,t)$ exists, and there
is a constant $C>0$ such that if $B_{R}(x_{0})\subset\hat{\Omega}$
is isometric to a Euclidean disc, then 
\begin{align}
P^{\hat{\Omega}}(x,y,t) & \leq\frac{C}{R^{2}}e^{-\frac{R^{2}}{8t}}, & \forall x\notin B_{R}(x_{0}),y\in B_{\frac{R}{2}}(x_{0}),t>0;\label{eq: comparison1}\\
|P^{\hat{\Omega}}(x,y,t)-P^{\C}(x,y,t)| & \leq\frac{C}{R^{2}}e^{-\frac{R^{2}}{2t}}, & \forall x,y\in B_{\frac{R}{2}}(x_{0}),t>0.\label{eq: comparison2}
\end{align}
\end{lem}

\begin{proof}
If $\check{\gamma}$ is the Brownian motion in $\C$ started at the
origin and $\tau=\min\{t:|\check{\gamma}(t)|\geq r\},$ then, by symmetry,
$\P(|\check{\gamma}_{t}|\geq r|\tau<t)\geq\frac{1}{2}.$ Therefore,
\begin{equation}
\P(\diam(\check{\gamma}_{[0,t]})>2r)\leq\P(\tau_{r}<t)\leq2\P(|\check{\gamma}_{t}|>r)=2e^{-\frac{r^{2}}{2t}}.\label{eq: Brownian_short_time}
\end{equation}
Define the sequence of stopping times $\tau_{0}=0,$ $\tau_{2k+1}=\min\{t>\tau_{2k}:|\hat{\gamma}_{t}-x_{0}|=\frac{3}{4}R\},$
$\tau_{2k}=\min\{t>\tau_{2k-1}:|\hat{\gamma}_{t}-x_{0}|=R\}.$ If
$\phi$ is any non-negative continuous function supported inside $B_{R/2}(x_{0}),$
we have $\E^{x}\phi(\hat{\gamma}_{t})=\sum_{k=1}^{\infty}\E^{x}[\phi(\hat{\gamma}_{t})\ind_{t\in[\tau_{2k-1},\tau_{2k}]}].$
Conditionally on $\F(\hat{\gamma}_{[0,\tau_{2k-1}]})$ and on the
event $t\in[\tau_{2k-1},\tau_{2k}],$ the distribution of $\hat{\gamma}_{t}$
is that of the (time-shifted) Brownian motion started at a point on
$\pa B_{\frac{3R}{4}}(x_{0})$ and conditioned to stay in $B_{R}(x_{0}).$
If $\tilde{P}(x,y,t)$ denotes the heat kernel of this conditioned
Brownian motion, then we have, for some $C>0,$ 
\begin{equation}
\sup_{x\in\pa B_{\frac{3R}{4}}(x_{0}),y\in B_{\frac{R}{2}}(x_{0}),t>0}\tilde{P}(x,y,t)\leq\frac{C}{R^{2}}.\label{eq: conditioned_to_stay}
\end{equation}
Indeed, by Brownian scaling, we may assume $R=1$. Write $P^{\D}(x,y,t)=\sum_{i}e^{-\lambda_{i}t}\psi_{i}(x)\psi_{i}(y)$
and $Q(x,t)=\int_{\D}P^{\D}(x,y,t)dy$, where $\psi_{i}$ and $\lambda_{i}$
are normalized eigenfunctions and eigenvalues of the Laplacian in
the unit disc $\D$ with zero boundary conditions. We have $\tilde{P}(x,y,t)=P^{\D}(x,y,t)/Q(x,t),$
and $P^{\D}(x,y,t)\leq P^{\C}(x,y,t)=\frac{1}{2\pi t}e^{-\frac{|x-y|^{2}}{2t}}.$
From this, it follows that $\tilde{P}(x,y,t)$ is bounded for small
$t$ and $\tilde{P}(x,y,t)\stackrel{t\to\infty}{\longrightarrow}\psi_{1}(y)$
for large $t,$ and (\ref{eq: conditioned_to_stay}) follows. We arrive
at 
\begin{multline*}
\E^{x}\phi(\hat{\gamma}_{t})\leq\sum_{k=0}^{\infty}\frac{C}{R^{2}}\left(\int_{B_{\frac{R}{2}}(x_{0})}\phi\right)\P(t\in[\tau_{2k-1},\tau_{2k}])\\
\leq\frac{C}{R^{2}}\left(\int_{B_{\frac{R}{2}}(x_{0})}\phi\right)\P(\tau_{1}<t)\leq\frac{2C}{R^{2}}e^{-\frac{R^{2}}{8t}}\left(\int_{B_{\frac{R}{2}}(x_{0})}\phi\right)
\end{multline*}
which proves (\ref{eq: comparison1}). To prove (\ref{eq: comparison2}),
let $\tau_{R}=\min\{s:|\hat{\gamma}_{s}-x_{0}|=R\}$ and note that
the Brownian motions $\hat{\gamma},\check{\gamma}$ are coupled to
coincide up to $\tau,$ so that 
\[
|\E^{x}\phi(\hat{\gamma}_{t})-\E^{x}\phi(\check{\gamma}_{t})|\leq(\E(\E^{\hat{\gamma}_{\tau}}(\phi(\hat{\gamma}_{t})|\tau<t)+\E^{\hat{\gamma}_{\tau}}(\phi(\check{\gamma}_{t})|\tau<t))\P(\tau_{R}<t),
\]
and (\ref{eq: comparison2}) follows by (\ref{eq: comparison1}) and
(\ref{eq: Brownian_short_time}).
\end{proof}
\begin{proof}[Proof of Lemma \ref{lem: symmetry}.]
 We first check that $\hat{\gamma}_{t}$ is a $\Leb(\hat{\Omega})$-symmetric
Markov process. For $\hat{x}_{0},\hat{y}_{0}\in\hat{\Omega}$ and
small $\eps_{x,y}>0,$ fix a ball $B_{\eps}(\check{x}_{0})\subset\C$
and an isometry $\tilde{\sigma}$ of $B_{\eps_{x}}(\check{x}_{0})$
to $B_{\eps_{x}}(\hat{x}_{0}).$ Let $\check{\gamma}:[0,\tau]\to\C$
be a path with $\check{\gamma}(0)\in B_{\eps_{x}}(\check{x}_{0}).$
If its lift $\hat{\gamma}$ ends up in $B_{\eps_{y}}(\hat{y}_{0})$,
we denote by $\sigma_{y}(\check{\gamma},\tilde{\sigma})$ the isometry
of $B_{\eps_{y}}(\check{y}_{0})$ to $B_{\eps_{y}}(\hat{y}_{0})$
obtained by extending $\tilde{\sigma}$ along $\check{\gamma}$ (otherwise,
$\sigma_{y}(\check{\gamma},\tilde{\sigma})$ is undefined.) We denote
by $\mathcal{Y}$ the set of all $\check{y}_{0}$ obtained in this
way, and by $\Sigma_{y}$ the set of all isometries $\sigma_{y}(\check{\gamma},\tilde{\sigma})$
modulo shifts. If $\hat{\gamma}$ ends up in $B_{\eps_{x}}(\hat{x}_{0}),$
we define $\sigma_{x}(\check{\gamma},\tilde{\sigma}),$ $\mathcal{X}$
and $\Sigma_{x}$ similarly. Note that $\mathcal{X},\mathcal{Y}$
are discrete sets in the plane, and $\Sigma_{x},\Sigma_{y}$ are finite
with $|\Sigma_{x}|=|\Sigma_{y}|$.

We now write, for the coupled Brownian motions $\check{\gamma},\hat{\gamma},$
\begin{multline*}
\int_{\hat{x}\in B_{\eps_{x}}(\hat{x}_{0})}\P^{\hat{x}}(\hat{\gamma}_{t}\in B_{\eps_{y}}(\hat{y}_{0}))\,d\hat{x}\\
=\frac{1}{|\Sigma_{x}|}\sum_{\sigma\in\Sigma_{x}}\sum_{\sigma'\in\Sigma_{y}}\sum_{\check{y}_{0}\in\mathcal{Y}}\int_{\check{x}\in B_{\eps_{x}}(\check{x}_{0})}\int_{\check{y}\in B_{\eps_{y}}(\check{y}_{0})}\P^{\check{x}}\left(\sigma_{y}(\check{\gamma}_{[0,t]},\sigma)\sim\sigma'|\check{\gamma}_{t}=\check{y}\right)P^{\C}(\check{x},\check{y},t)\,d\check{x}d\check{y},
\end{multline*}
where $\sim$ stands for equality of isometries modulo shifts. By
the reversibility of the planar Brownian notion, denoting by $\check{\gamma}^{-1}$
the time-reversal of $\check{\gamma},$ 
\[
\P^{\check{x}}\left(\sigma_{y}(\check{\gamma}_{[0,t]},\sigma)\sim\sigma'|\check{\gamma}_{t}=\check{y}\right)=\P^{\check{x}}\left(\sigma_{x}(\check{\gamma}_{[0,t]}^{-1},\sigma')\sim\sigma|\check{\gamma}_{t}=\check{y}\right)=\P^{\check{y}}\left(\sigma_{x}(\check{\gamma}_{[0,t]},\sigma')\sim\sigma|\check{\gamma}_{t}=\check{x}\right).
\]
Also, $P^{\C}(\check{x},\check{y},t)=P^{\C}(\check{y},\check{x},t),$
and if we shift every point in $\mathcal{Y}$ to some fixed point
in this set, then the point $\check{x}_{0}$ gets shifted to every
point in $\mathcal{X}.$ Collecting all these observations together,
we end up with 
\[
\int_{\hat{x}\in B_{\eps_{x}}(\hat{x}_{0})}\P^{\hat{x}}(\hat{\gamma}_{t}\in B_{\eps_{y}}(\hat{y}_{0}))\,d\hat{x}=\int_{\hat{y}\in B_{\eps_{y}}(\hat{y}_{0})}\P^{\hat{y}}(\hat{\gamma}_{t}\in B_{\eps_{x}}(\hat{x}_{0}))\,d\hat{y},
\]
which implies that $\hat{\gamma}_{t}$ is $\Leb(\hat{\Omega})$-symmetric.
The transition kernel of the reflected process can be written as $\hat{P}(x,y,t)+\hat{P}(x,\overline{y},t),$
implying symmetry, and killing upon hitting a closed subset preserves
the class of symmetric processes.

For the boundedness, the bound (\ref{eq: comparison1}) and the symmetry
imply that $P^{\hat{\Omega}}(x,y,t)\leq CR^{-2}$ if either $x$ or
$y$ are at distance at least $R$ from any conical tip of $\hat{\Omega}.$
Now, fix a small $R>0$, and let $x\in B_{R}(x_{0}),$ where $x_{0}$
is a conical tip. If $y\notin B_{R}(x_{0}),$ then $P^{\hat{\Omega}}(x,y,t)\leq CR^{-2}$
by the strong Markov property with respect to exit time from $B_{R}(x_{0}).$
Else, coupling the Brownian motions in $\hat{\Omega}$ and in the
infinite cone $\mathcal{C^{\alpha}}_{x_{0}}$ up to exiting $B_{R}(x_{0}),$
we similarly obtain 
\[
P^{\hat{\Omega}}(x,y,t)\leq CR^{-2}+P^{\mathcal{C}_{x_{0}}}(x,y,t).
\]
Coupling to the Brownian motion in $\C$, we have $P^{\mathcal{C}_{x_{0}}}(x,y,t)\leq\sum_{\check{y}\in\mathcal{Y}}P^{\C}(\check{x},\check{y},t)\leq\text{\ensuremath{\mathrm{const}}}(t,\alpha)$
where $\mathcal{Y}$ is the finite set of endpoints of paths in $\C$
starting at $\check{x}$ that lift to a path from $x$ to $y.$

For the smoothness, if $B_{R}(x_{0})$ is Euclidean, $y\in B_{\frac{R}{2}}(x_{0})$
and $x\notin B_{R}(x_{0}),$ we can write as in the proof of (\ref{eq: comparison1}):

\[
P^{\hat{\Omega}}(x,y,t)=\sum_{k}\E\E\left[\left.\tilde{P}(\gamma_{\tau_{2k-1}},y,t-\tau_{2k-1})\ind_{t\in\tau_{2k-1},\tau_{2k}}\right|\F_{\tau_{2k-1}}\right],
\]
 and note that all the derivatives of $\tilde{P}(x,y,t)$ are uniformly
bounded over $x\in\pa B_{\frac{3R}{4}}(x_{0}),$ $y\in B_{\frac{R}{2}}(x_{0}),$
$t>0.$ Since $R$ is at our disposal, this proves smoothness for
$x\neq y;$ for $x=y$, one can use a similar decomposition, on the
event that $t<\tau_{2},$ we can use that the heat kerned in the disc
is smooth.
\end{proof}
In order to prove Lemma \ref{prop: FCLT}, we first invoke the functional
CLT in the plane:
\begin{lem}
(Coupling the random walk to the Brownian motion.) \label{prop:KMT}
It is possible to multiply the weights $w_{xy}$ by a common factor
so that for any $\eta>0$, there exist $C,\epsilon>0$ such that for
any $\delta,T>0$ the random walk $\gamma_{\delta^{-2}t}^{\delta}$
on $\C^{\delta}$ can be coupled to the Brownian motion $\gamma_{t}$
so that 
\begin{equation}
\mathbb{P}\left[\sup\limits _{t\in[0,T]}|\gamma_{\delta^{-2}t}^{\delta}-\gamma_{t}|>T^{\frac{1}{4}+\eta}\delta^{\frac{1}{2}-2\eta}\right]<C\exp(-(T\delta^{-2})^{\epsilon})\label{eq:PropKMT}
\end{equation}
\end{lem}

\begin{proof}
For the random walk $\gamma_{t}^{\delta_{0}}$on $\C^{\delta_{0}}$,
we define a sequence of times $t_{0}=0$, 
\[
t_{k+1}:=\min\{t\geq t_{k}+1:\gamma_{t}^{\delta_{0}}\cong\gamma_{0}^{\delta_{0}}\},
\]
where $\cong$ means equality modulo a shift of $\C^{\delta_{0}}$.
Then $t_{k}$ and $\gamma_{t_{k}}^{\delta_{0}}$ have i.i.d. increments
with exponentially small tails; because of the symmetries of the lattice,
the increments of $\gamma_{t_{k}}^{\delta_{0}}$ have zero mean and
a scalar covariance matrix $\Sigma$. Let $\tau=\E t_{1}.$ Put $n=\left\lfloor 2\tau^{-1}\delta^{-2}T\right\rfloor .$
Einmahl's version of KMT theorem (\cite[Theorem 4]{Einmahl}, plug
in $H(t):=\exp(\sqrt{t})$, $x:=\delta_{0}n^{\frac{1}{4}}/3$) provides
a coupling of $\gamma_{t_{k}}^{\delta_{0}}$ and a Brownian motion
$\tilde{\gamma}_{t}$ with covariance matrix $\Sigma$ such that
\[
\mathbb{P}\left[\sup\limits _{k\leq n}|\gamma_{t_{k}}^{\delta_{0}}-\tilde{\gamma}_{k}|>\frac{\delta_{0}}{3}n^{\frac{1}{4}}\right]<K_{1}\cdot n\cdot\exp(-K_{2}n^{\frac{1}{4}}).
\]
with $K_{1,2}$ depending only on $\C^{\delta_{0}}.$ We put $\gamma_{t}=\tilde{\gamma}_{t/\tau}.$
For $t>0$ we set $k(t):=\max\{k:t_{k}<t\}$, then $t_{k(t)}<t\leq t_{k(t)+1}$.
We estimate 
\begin{equation}
|\gamma_{t}^{\delta_{0}}-\gamma_{t}|\leq|\gamma_{t}^{\delta_{0}}-\gamma_{t_{k(t)}}^{\delta_{0}}|+|\gamma_{t_{k(t)}}^{\delta_{0}}-\gamma_{\tau k(t)}|+|\gamma_{\tau k(t)}-\gamma_{t}|.\label{eq: three_events}
\end{equation}
By Chernoff bound, given $\eta>0$, we have $\P(|t_{k}-\tau k|>n^{\frac{1}{2}+\eta})\le C\exp(-cn^{2\eta})$
for each $k\leq n$. Therefore, $\P(\exists k\leq n:|t_{k}-\tau k|\geq n^{\frac{1}{2}+\eta})\leq Cn\exp(-cn^{2\eta}).$
In particular, $\P(\exists t\leq T\delta^{-2}:k(t)>n)\leq Cn\exp(-cn^{2\eta}).$
Also, $\P(\exists k\leq n:t_{k+1}-t_{k}\geq n^{\frac{1}{2}+\eta})\leq C\exp(-cn^{\eta}).$
Together, this implies that $\P(\exists t\leq T\delta^{-2}:|t-\tau k(t)|\geq2n^{\frac{1}{2}+\eta})\leq Cn\exp(-cn^{\eta}).$
For the Brownian motion $\gamma_{t}$, for each fixed $k\leq n$,
we have $\P(\exists t:|t-\tau k|\leq2n^{\frac{1}{2}+\eta};|\gamma_{t}-\gamma_{\tau k}|\geq\frac{\delta_{0}}{3}n^{\frac{1}{4}+\eta})\leq C\cdot\exp(-cn^{\eta}).$
Summing over $k,$ we conclude that $\P(\exists t\leq\delta^{-2}T:|\gamma_{t}-\gamma_{\tau k(t)}|\geq\frac{\delta_{0}}{3}n^{\frac{1}{4}+\eta})\leq C\cdot n\cdot\exp(-cn^{\eta}).$
Also, because of exponential tails of $t_{k}-t_{k-1},$ we have $\P(\exists t\leq\delta^{-2}T:|\gamma_{t}^{\delta_{0}}-\gamma_{t_{k(t)}}^{\delta_{0}}|\geq\frac{\delta_{0}}{3}n^{\frac{1}{4}})\leq C\cdot n\cdot\exp(-c(n^{1/4}+n^{2\eta})).$
Combining the estimates of the three terms in (\ref{eq: three_events})
together, we see that 
\[
\P\left[\sup\limits _{t\in[0,T\delta^{-2}]}|\gamma_{t}^{\delta_{0}}-\gamma_{t}|>\delta_{0}n^{\frac{1}{4}+\eta}\right]\leq\hat{C}\exp(-n^{\epsilon}),
\]
 for any $\epsilon<\min(\eta;\frac{1}{4}).$ Scaling time by $\delta^{-2}$,
the lattice by $N$, and the weights $w_{xy}$ so that $N^{-1}\cdot\gamma_{\delta^{-2}t}\sim\delta_{0}^{-1}\cdot\gamma_{t}$
is a standard Brownian motion, yields the result.
\end{proof}
\begin{proof}[Proof of Lemma \ref{prop: FCLT}]
 Let $\hat{\Omega}$ (respectively, $\hat{\Omega}^{\delta}$) be
two copies of $\Omega$ glued along the boundary. The random walk
$\hat{\gamma}_{[\delta^{-2}t]}^{\delta}$ in $\hat{\Omega}^{\delta}$
(respectively, the Brownian motion $\hat{\gamma}_{t}$ in $\hat{\Omega}$)
can be coupled to a random walk $\check{\gamma}_{[\delta^{-2}t]}^{\delta}$
in $\C^{\delta}$ (respectively, to planar Brownian motion $\check{\gamma}_{t}$)
by moving in the same way locally; note that the BM in $\hat{\Omega}$
never visits conical singularities. By (\ref{eq:PropKMT}), $\check{\gamma}_{[\delta^{-2}t]}^{\delta}$
can be coupled to $\check{\gamma}_{t}$ in such a way that $\sup_{t\leq T}|\check{\gamma}_{\delta^{-2}t}^{\delta}-\check{\gamma}_{t}|\to0$
as $\delta\to0$ almost surely. On the event of probability $1$ that
$\hat{\gamma}_{[0,T]}$ does not visit conical singularities, this
implies $\dist(\hat{\gamma}_{\delta^{-2}t}^{\delta},\hat{\gamma}_{t})\leq|\check{\gamma}_{\delta^{-2}t}^{\delta}-\check{\gamma}_{t}|$
for all $t\leq T$ eventually. Reflecting the random walk and the
Brownian motion at the Neumann boundary does not increase distances.
If $\tau$ (resp. $\tau^{\delta}$) is the first time $\hat{\gamma}_{t}$
(resp. $\hat{\gamma}_{\delta^{-2}t}^{\delta}$) hits $\partial_{D}\Omega$
(resp. $\partial_{D}\Od$), then, almost surely, $\hat{\gamma}_{t}$
will have points on both sides of the boundary in each interval $(\tau,\tau+\eps)$.
On that event, almost surely, $\tau^{\delta}\to\tau$ and hence $\hat{\gamma}_{\delta^{-2}\tau^{\delta}}^{\delta}\to\hat{\gamma}_{\tau}$.
Therefore, stopping at Dirichlet boundary also does not affect the
convergence, and $\sup_{t\leq T}\dist(\gamma_{[\delta^{-2}t]}^{\delta};\gamma_{t})\to0$,
almost surely. This completes the proof.
\end{proof}
\begin{proof}[Proof of Lemma \ref{lem: Short-time-large-diameter}.]
 It suffices to prove the bound for the walk on $\C^{\delta}$ (by
passing first to $\hat{\Omega}^{\delta}$ as in the proof of Lemma
\ref{prop: FCLT} and then assuming by Markov property that $x$ is
at distance at least $\eps/10$ from conical singularities). If $t>\delta^{\frac{3}{2}},$
we use Lemma \ref{prop:KMT} and bound $\P(\diam\gamma_{[0,t\delta^{-2}]}^{\delta}>\eps)\leq\P(\diam\gamma_{[0,t]}>\frac{\eps}{3})+\P(\sup_{s\in[0,t]}|\gamma_{s\delta^{-2}}^{\delta}-\gamma_{s}|>\frac{\eps}{3}).$
The first term converges to $0$ super-polynomially in $t$ while
the second one bounded from above by $C\cdot\exp(-(t\delta^{-2})^{\epsilon})\leq C\cdot\exp(-t{}^{-\frac{\epsilon}{3}})$
provided that $t\leq1$ and $\delta^{\frac{1}{2}-2\eta}<\frac{\eps}{3}.$

If $t\leq\delta^{\frac{3}{2}},$ then $\delta^{-2}t\le t^{\frac{1}{3}}\delta^{-1}.$
Pick $\hat{c}>0$ in such a way that for all $\delta$ small enough,
the random walk on $\C^{\delta}$ needs at least $K:=\left\lfloor \hat{c}\delta^{-1}\right\rfloor $
steps to reach diameter $\eps$. The probability of this is bounded
by $\P(X\geq K),$ where $X$ is a Poisson random variable with mean
$MK$, and $M=t^{\frac{1}{3}}\delta^{-1}\max_{x\in\C^{\delta_{0}}}\{\sum_{y\sim x}w_{xy}\}/K\leq c't^{\frac{1}{3}},$
with $c'$ a constant depending on $\C^{\delta_{0}}$ and $\eps.$
If $M<\frac{1}{2},$ then, for any $\alpha>0$, we have by Stirling
bound 
\[
\P(X\geq K)\leq\sum_{n=K}^{\infty}\frac{(MK)^{n}}{n!}e^{-MK}\leq2\frac{(MK)^{K}}{K!}e^{-MK}\leq2M^{\alpha K}e^{K(1-M+(1-\alpha)\log M)}.
\]
Since $1-M+\log M<0$ for $M<1,$ we can pick $\alpha>0$ such that
the exponential is bounded by $1$ for all $M<\frac{1}{2},$ i.e.,
$\P(X\geq K)\leq2(c't^{\frac{1}{3}})^{\alpha\left\lfloor \hat{c}\delta^{-1}\right\rfloor }.$
For $\delta$ small enough, the exponent is at least $20$, and so
we have $\P(\diam\gamma_{[0,t\delta^{-2}]}^{\delta}>\eps)\leq\P(X\geq K)\leq2t^{\frac{20}{6}}$
provided that $t<(c')^{-6}.$
\end{proof}
For $x\in\Od$ and $r>10N^{-1}>0,$ denote $Q(x,r):=[0;r^{2}]\times B(x,r).$
The \emph{parabolic Harnack inequality} (PHI) asserts that there exists
a constant $C_{H}$ such that, for any $\delta,$ any $x,r$ such
that $B(x,r)\cap\partial_{D}\Od=\emptyset$, and any $u$ positive
and satisfying 
\begin{equation}
\partial_{t}u=\delta^{2}\cdot\Delta^{\Od}u\label{eq: Parabolic}
\end{equation}
 in $Q(x,r)$, one has 
\begin{equation}
\inf_{Q_{+}(x,r)}u\geq C_{H}\sup_{Q_{-}(x,r)}u,\label{eq: PHI}
\end{equation}
 where $Q_{-}(x,r)=[\frac{1}{4}r^{2};\frac{1}{2}r^{2}]\times B(x,\frac{r}{2})$
and $Q_{+}(x,r)=[\frac{3}{4}r^{2};r^{2}]\times B(x,\frac{r}{2}).$
In our setting, PHI follows from \cite[Theorem 1.7]{Delmotte}. Delmotte
uses normalized Laplcian in which the random walk jumps at rate one;
however since he allows for jumps from a vertex to itself, the two
setups are equivalent; note that we rescale the graph distance but
don't rescale time, hence the additional factor of $\delta^{2}$ in
(\ref{eq: Parabolic}). 

Of the three conditions of \cite[Theorem 1.7]{Delmotte}, the volume
doubling condition $DV(C_{1})$ and uniform ellipticity conditions
$\Delta(c)$ are obvious in our setting: they state that $|B(x,r)|\leq C_{1}|B(x,2r)|$
for any $x,r$, and $\min_{y\sim x}w_{yx}\geq cw_{x}$ for any $x$,
respectively. This third one, the Poincaré inequality $P(C_{2})$,
asserts that 
\[
\sum_{B(x_{0},r)}w_{x}(f(x)-\overline{f}_{B})^{2}\leq C_{2}r^{2}\sum_{x\sim y\in B(x_{0},2r)}w_{xy}(f(x)-f(y))^{2},\quad\overline{f}_{B}=\frac{1}{\sum_{x\in B(x_{0},r)}w_{x}}\sum_{B(x_{0},r)}w_{x}f(x)
\]
Recall (see e.g. \cite{EvansGariepi}) the classical proof of the
Poincaré inequality for a ball $B$: by Cauchy-Schwarz, $(f(x)-f(y))^{2}\leq\left(\int_{[xy]}|\nabla f(z)|\,dz\right)^{2}\leq|x-y|\int_{[xy]}|\nabla f(z)|^{2}\,dz,$
from which the Poincaré inequality follows by integrating over $x,y\in B$.
This proof extends to the discrete settings of balls in $\C^{\delta},$
simply by replacing integration with summation, in particular, integration
over a segment $[xy]$ with summation over $H\delta$-neighborhood
of $[xy]$ for a large enough fixed $H.$ In the final step, we sum
with weights $w_{x}w_{y};$ note that $w_{xy}$ in the right-hand
side can be ignored since they are uniformly bounded away from $0.$
For balls in an infinite cone or an infinite wedge, it suffices to
map the cone or the wedge by a bi-Lipschitz map to the plane or the
half-plane and apply the same proof, increasing $H$ if necessary.
Since we only apply PHI to balls in model surfaces, these cases are
all we need. A typical example of a function $u$ that PHI is applied
to is $P^{\Od}(x_{0},x,\delta^{-2}t).$ 

We recall the standard argument that PHI implies Hölder regularity
of solutions to (\ref{eq: Parabolic}), namely, there exist $\theta>0$
and $\Chol>0$ such that any $\delta>0$ and for $Q(x,r)$ as above,
one has 
\begin{equation}
|u(r^{2},x)-u(r^{2},y)|\leq\Chol\cdot\left(\frac{|x-y|}{r}\right)^{\theta}\cdot\osc{Q(x,r)}u,\quad y\in B(x,r).\label{eq: Holder}
\end{equation}
To prove (\ref{eq: Holder}), note that if $\hat{u}$ is $u$ normalized
so that $\inf_{Q(x,r)}\hat{u}=0$ and $\sup_{Q(x,r)}\hat{u}=1$, and
$\sup_{Q_{-}(x,r)}\hat{u}\geq\frac{1}{2},$ then 
\[
\osc{Q_{+}(x,r)}\hat{u}=\sup_{Q_{+}(x,r)}\hat{u}-\inf_{Q_{+}(x,r)}\hat{u}\leq1-\frac{C_{H}}{2}.
\]
 If $\sup_{Q_{-}(x,r)}\hat{u}\leq\frac{1}{2},$ then passing to $1-\hat{u}$
leads to the same conclusion. Hence, $\osc{Q_{+}(x,r)}u\leq c\cdot\osc{Q(x,r)}u,$
with $c=1-\frac{C_{H}}{2}<1.$ Applying the same reasoning to $Q_{+}(x,r)$
and iterating, we conclude that if $|y-x|<\frac{r}{2^{k}}$, then
$|u(r^{2},x)-u(r^{2},y)|<c^{k}\osc{Q(x,r)}u$, yielding (\ref{eq: Holder}).
\begin{proof}[Proof of Lemma \ref{prop: HK_bound}]
Since turning Dirichlet boundary into Neumann one only increases
$P^{\Od},$ by passing to $\hat{\Omega}^{\delta}$ as in the proof
of Lemma \ref{prop: FCLT}, we may assume that $\partial\Od=\emptyset$.
Let $M=M_{\eps,\eta}^{\delta}:=\sup_{t\geq\eps,\dist(x,y)\geq\eta}P^{\Od}(x,y,\delta^{-2}t),$
and fix $t_{0}\geq\eps$ and $x_{0},y_{0}$ with $\dist(x_{0},y_{0})\geq\eta$
such that $P^{\Od}(x_{0},y_{0},\delta^{-2}t_{0})>\frac{M}{2}.$ Applying
PHI to $P^{\Od}(x,\cdot,\cdot)$ and to $Q(y_{0},\sqrt{\eps})$ shifted
in time by $t_{0}-\frac{\eps}{4}$, we find that $P^{\Od}(x_{0},y,\delta^{-2}t)\geq C_{H}\frac{M}{2}$
if $\dist(y,y_{0})\leq\sqrt{\eps}/2$ and $t\in[t_{0}+\frac{\eps}{2},t_{0}+\frac{3\eps}{4}].$
Since $\sum_{y}P^{\Od}(x_{0},y,\delta^{-2}t)=1,$ this implies, for
$\eta=0,$ 
\[
C_{H}\frac{M_{\eps,0}^{\delta}}{2}\cdot\#\{y:|y-y_{0}|<\frac{\sqrt{\eps}}{2}\}\leq1,
\]
 that is, $M_{\eps,0}^{\delta}\leq C(\eps)\delta^{2}.$ If $\eta\neq0$
is fixed and $\eps<4\eta^{2}$, then $\dist(y,y_{0})<\frac{\sqrt{\eps}}{2}$
implies $\dist(y,x_{0})\geq\eta/2$. By Lemma \ref{lem: Short-time-large-diameter},
we have $\sum_{y:\dist(x_{0},y)>\eta/2}P^{\Od}(x_{0},y,\delta^{-2}t)\leq C(\eta)t^{3}$,
thus
\[
C_{H}\frac{M_{\eps,\eta}^{\delta}}{2}\cdot\#\{y:|y-y_{0}|<\frac{\sqrt{\eps}}{2}\}\leq C(\eta)t^{3},
\]
i.e., $M_{\eps,\eta}^{\delta}\leq C'(\eta)\frac{\delta^{2}}{\eps}t^{3}$.
Note that if $M_{\eps,\eta}^{\delta}>M_{2\eps,\eta}^{\delta},$ then
we could take $t_{0}\leq2\eps$ and thus $t\leq\frac{11}{4}\eps$,
in which case the last inequality becomes $M_{\eps,\eta}^{\delta}\leq C''(\eta)\delta^{2}\eps^{2}.$
Hence, $M_{0,\eta}^{\delta}=M_{\eps_{0}(\eta),\eta}^{\delta}\leq M_{\eps_{0}(\eta),0}^{\delta}$
for some $\eps_{0}(\eta)>0$. 
\end{proof}
\begin{proof}[Proof of Lemma \ref{prop: Holder-regularity}.]
 Since $||P^{\Od,\varphi}(x,y,\tau)||\leq P^{\Od}(x,y,\tau)$, we
have, by Lemma \ref{prop: HK_bound}, 
\[
\osc Q(P^{\Od,\varphi}(x,\cdot,\cdot))\leq C\delta^{2},
\]
 where $Q:=[t-\frac{\eta^{2}}{4},t]\times B(x,\frac{\eta}{2}).$ Since
the matrix components of $P^{\Omega,\varphi}(x,\cdot,\delta^{-2}t)$
satisfy (\ref{eq: Parabolic}), the result now follows directly from
(\ref{eq: Holder}).
\end{proof}
\begin{proof}[Proof of Lemma \ref{prop: (Spectral-gap)}.]
 The proof proceeds case by case.

\emph{Case 1.} Suppose that $\partial_{\Dir}\Omega\neq\emptyset$,
thus $k=0.$ Since $||P^{\Od,\varphi}(x,y,t)||\leq|P^{\Od}(x,y,t)|,$
it suffices to prove the result for the trivial line bundle. The probability
that by time $1$, the Brownian motion $\gamma_{t}$ started at $x$
has hit the Dirichlet boundary is a positive continuous function on
$\Omega,$ hence it is bounded from below, say by $2\eta.$ Hence,
for $\delta$ small enough, the probability that the random walk $\gamma_{\delta^{-2}t}^{\delta}$
hits the Dirichlet boundary before $t=1$ is bounded below by $\eta$,
independently of the starting point. By Markov property, this implies
that the probability that it does not hit $\partial_{\Dir}\Od$ by
time $t$ is bounded above by $(1-\eta)^{\left\lfloor t\right\rfloor }$,
i.e., 
\[
\sum_{y\in\Od}P^{\Od}(x,y,\delta^{-2}t)\leq(1-\eta)^{\left\lfloor t\right\rfloor }.
\]
Using PHI as in the proof of Lemma \ref{prop: HK_bound}, we see that
this implies $P^{\Od}(x,y,\delta^{-2}t)<C\delta^{2}(1-\eta)^{\left\lfloor t\right\rfloor }$
for any $x,y$. Summing this bound over $x=y\in\Od$ yields the desired
result.

\emph{Case 2.} Suppose that $\Omega$ has no Dirichlet boundary and
$\varphi$ is the trivial line bundle, thus $k=1$. We claim that
the heat kernel $P^{\Od}(x,y,\delta^{-2}t)$ at time $t=1$ is uniformly
contracting in the total variation distance, i.e., there exists $\eta>0$
such that for any $x_{1},x_{2}$, and any $\delta$ small enough,
one has
\begin{equation}
\frac{1}{2}\sum_{y}\left|P^{\Omega^{\delta}}(x_{1},y,\delta^{-2})-P^{\Omega^{\delta}}(x_{2},y,\delta^{-2})\right|\leq1-\eta.\label{eq:mixtv}
\end{equation}
Indeed, for any test function $\psi>0$, by compactness, we have for
the \emph{continuous} heat kernel, 
\[
\inf_{x}\int_{\Omega}P^{\Omega}\left(x,y,\frac{1}{2}\right)\psi(y)dy=:2c_{\psi}>0.
\]
since the expression under infimum is positive and continuous in $x$.
Hence, 
\[
\inf_{x}\sum_{y\in\Od}P^{\Od}\left(x,y,\delta^{-2}/2\right)\psi(y)>c_{\psi}.
\]
 for $\delta$ small enough. If $\supp\psi\subset B(y_{0},r)$, say
with with $r=1$, this implies that 
\[
\inf_{x}\sup_{y\in B(y_{0},r)}P^{\Od}(x,y,\delta^{-2}/2)\geq c'_{\psi}\delta^{2},
\]
and hence, by PHI, $\inf_{x}\inf_{y\in B(y_{0},r)}P^{\Od}(x,y,\delta^{-2})\geq C_{H}c'_{\psi}\delta^{2}.$
This gives the desired improvement on the trivial bound of $1$ on
the LHS of (\ref{eq:mixtv}). Now, iterating (\ref{eq:mixtv}), we
see that $\frac{1}{2}\sum_{y}\left|P^{\Omega^{\delta}}(x_{1},y,\delta^{-2}t)-P^{\Omega^{\delta}}(x_{2},y,\delta^{-2}t)\right|\leq(1-\eta)^{\left\lfloor t\right\rfloor }$,
hence 
\begin{multline*}
\sum_{x}P^{\Omega^{\delta}}(x,x,\delta^{-2}t)-k=\sum_{x}P^{\Omega^{\delta}}(x,x,\delta^{-2}t)-\frac{1}{|\Omega^{\delta}|}\left(\sum_{x,y}P^{\Omega^{\delta}}(y,x,\delta^{-2}t)\right)\\
=\frac{1}{|\Omega^{\delta}|}\sum_{x,y}\left(P^{\Omega^{\delta}}(x,x,\delta^{-2}t)-
P^{\Omega^{\delta}}(y,x,\delta^{-2}t)\right)\leq2\cdot(1-\eta)^{\left\lfloor t\right\rfloor }
\end{multline*}
for all $\delta$ small enough, independently of $t$, as required. 

\emph{Case 3.} Suppose that $\Omega$ has no Dirichlet boundary and
$k=0$. Pick any $x_{0}\in\Od$, and (non-contractible) loops $\beta_{(1)},\dots\beta_{(n)}$
rooted at $x_{0}$ such that $\{\varphi(\gamma_{j})\}_{j=0}^{n}$
do not have a common eigenvector of eigenvalue 1; if such loops did
not exist, then the translations of the common eigenvector would form
a covariant constant. Pick a a small $r$ such that $B(x_{0},r)$
is contractible, and, for $x,y\in B(x_{0},r),$ denote, for $i=0,1,\dots,n$,
$P_{(i)}^{\Od}(x,y,t):=\P^{x}([\gamma_{[0,t]}^{\delta}]=[\beta_{(i)}],\gamma_{t}^{\delta}=y)$,
where $\beta_{(0)}$ is a contractible loop, and we identify the points
of $B(x_{0},r)$ in order to compute the homotopy type of a non-closed
path. There exists a constant $c>0$ such that $c\delta^{2}\leq P_{(i)}^{\Od}(x,y,\delta^{-2})\leq c^{-1}\delta^{2}$
for each $i$ and each $x,y\in B(x_{0},r)$ and all $\delta$ small
enough; the upper bound follows from Lemma \ref{prop: HK_bound} and
the lower one is done exactly as in Case 2. We can write 
\begin{multline}
\|\Pp^{\Omega^{\delta},\varphi}(x,y,\delta^{-2}t)\|=\left\Vert \E^{x}\varphi(\gamma_{[0,\delta^{-2}t]}^{\delta})\ind_{\ind_{\gamma_{t}^{\delta}=y}}\right\Vert \\
\leq\left\Vert \sum_{i=0}^{n}\varphi(\beta_{(i)})\cdot P_{(i)}^{\Od}(x,y,\delta^{-2})\right\Vert +\P^{x}(\forall i,[\gamma_{[0,t]}^{\delta}]\neq[\beta_{(i)}],\gamma_{t}^{\delta}=y)\label{eq: norm_bound}
\end{multline}
where in the last term, we used that $||\varphi(\gamma_{[0,\delta^{-2}t]}^{\delta})||\leq1$
as $\varphi$ is unitary. We claim that 
\begin{equation}
\sup_{c\leq p_{i}\leq c^{-1},\left\Vert v\right\Vert =1}\frac{\left\Vert \sum_{i=0}^{n}p_{i}\varphi(\beta_{(i)})v\right\Vert }{\left(\sum_{i=0}^{n}p_{i}\right)}=:1-\eta<1.\label{eq: sup_bound}
\end{equation}
Indeed, the fraction is strictly smaller than $1$ unless all $\varphi(\beta_{(i)})v$
are non-negative multiples of each other, but since $\varphi(\beta_{0})=\mathrm{Id},$
this would mean that $v$ is a common eigenvector of eigenvalue $1$.
Applying (\ref{eq: sup_bound}) to (\ref{eq: norm_bound}), we get,
for any $x\in B(x_{0},r)$ and all $\delta$ small enough, 
\[
\sum_{y\in\Od}\left\Vert \Pp^{\Omega^{\delta},\varphi}(x,y,\delta^{-2})\right\Vert \leq(1-\eta)\P(E)+P(E^{c})\leq1-\eta'
\]
where $E=\{\gamma_{\delta^{-2}}^{\delta}\in B(x_{0},r)\text{ and }\exists i:[\gamma_{[0,\delta^{-2}]}^{\delta}]=[\beta_{(i)}]\},$
and we have used that $\P(E)\geq\sum_{y\in B(x_{0},r)}P_{(0)}^{\Od}(x,y,\delta^{-2})$
is uniformly bounded from below. Since the bounds obtained, in fact,
did not depend on $x_{0},$ the proof is now completed as in Case
1.

\emph{Case 4.} Suppose that $\Omega$ has no Dirichlet boundary. Let
$\varphi_{0}$ be the trivial sub-bundle of the maximal dimension
of $\varphi$, which is a direct sum of trivial line bundles. Since
$\varphi$ is unitary, we have $\varphi=\varphi_{0}\oplus\varphi_{0}^{\perp}$
for $\varphi_{0}^{\perp}$ the (point-wise) orthogonal complement
to $\varphi_{0}$; moreover, $\varphi_{0}^{\perp}$ has no trivial
line sub-bundles. Applying Case 2 to $\varphi_{0}$ and Case 3 to
$\varphi_{1}$ concludes the proof.
\end{proof}
\begin{proof}[Proof of Lemma \ref{lem: LCLT}]
 It suffices to consider the case when $\Lambda$ is a cone (or,
in particular, a plane). Indeed, $\P^{\H^{\delta}}(x,y,t)=\P^{\C^{\delta}}(x,y,t)\pm\P^{\C^{\delta}}(x,\bar{y},t),$
with $\bar{y}$ denoting the reflection with respect to the boundary
and the sign being $+$ for Neumann boundary condition and $-$ for
Dirichlet one. Similarly, for a corner $\Corner^{\alpha}$ with Dirichlet
or Neumann boundary conditions, one has $\P^{\Corner^{\alpha,\delta}}(x,y,t)=\P^{\Cone{2\alpha,\delta}}(x,y,t)\pm\P^{\Cone{2\alpha,\delta}}(x,\bar{y},t),$
where the cone $\Cone{2\alpha,\delta}$ is obtained by gluing two
copies of $\Corner^{\alpha,\delta}$ along the boundary. Finally,
$\P^{\Corner_{\Dir\Neu}^{\alpha,\delta}}(x,y,t)=\P^{\Corner_{\Dir}^{2\alpha,\delta}}(x,y,t)+\P^{\Corner_{\Dir}^{2\alpha,\delta}}(x,\bar{y},t),$
where $\Corner_{\Dir}^{2\alpha,\delta}$ is obtained by gluing two
copies of $\Corner_{\Dir\Neu}^{\alpha,\delta}$ along the Neumann
boundary. 

The proof below is for the case when there is no vertex at the tip
of the cone; we explain the necessary modifications for that case
in Remark \ref{rem: vertex_at_tip}. We begin by comparing the probabilities
$\P(\gamma_{\delta^{-2}t}^{\delta}\in B(y,\eps))$ and $\P(\gamma_{t}\in B(y,\eps))$
(hereinafter $\P=\P^{x}$) for a mesoscopic scale $\delta\ll\eps(\delta)\ll1$
to be specified later and $t\leq T:=\delta^{-\frac{1}{10}}$; the
goal is to show that they agree up to $O(\max\{t^{-1},1\}\eps^{2}\delta^{\rho}$)
for some $\rho>0$. We couple $\gamma_{\delta^{-2}t}^{\delta}$ and
$\gamma_{t}$ to the planar random walk and the Brownian motion $\hat{\gamma}_{[\delta^{-2}t]}^{\delta},\hat{\gamma}_{t}$
by the same local moves, and assume that $\hat{\gamma}_{\delta^{-2}t}^{\delta},\hat{\gamma}_{t}$
are coupled as in Lemma \ref{eq:PropKMT}, say, with $\eta=\frac{1}{10}.$
Pick a positive $\nu<\frac{1}{2}-2\eta-\frac{1}{10}(\frac{1}{4}+\eta)$,
so that $T^{\frac{1}{4}+\eta}\delta^{\frac{1}{2}-\eta}\ll\delta^{\nu}.$
We put $\mathcal{D}=\{\sup_{t\in[0,T]}|\hat{\gamma}_{\delta^{-2}t}^{\delta}-\hat{\gamma}_{t}|>\delta^{\nu}\}$,
$\mathcal{T}=\{\inf_{t\in(0,T)}(\dist(\gamma_{t},0))<2\delta^{\nu}\},$
where $0$ is the tip of the cone, and $\mathcal{B}:=\{\eps(\delta)-\delta^{\nu}\leq\dist(\gamma_{t};y)\leq\eps(\delta)+\delta^{\nu}\}.$
On the event $\mathcal{D}^{c}\cap\mathcal{T}^{c}$, we have $\sup_{t\in(0,T)}|\gamma_{\delta^{-2}t}^{\delta}-\gamma_{t}|\leq\delta^{\nu}$,
therefore, on $\mathcal{D}^{c}\cap\mathcal{T}^{c}\cap\mathcal{B}^{c}$,
either $\gamma_{\delta^{-2}t}^{\delta},\gamma_{t}\in B(y,\eps)$ or
$\gamma_{\delta^{-2}t}^{\delta},\gamma_{t}\notin B(y,\eps)$ simultaneously.
This implies 
\begin{multline}
\left|\P(\gamma_{\delta^{-2}t}^{\delta}\in B(y,\eps))-\P(\gamma_{t}\in B(y,\eps))\right|\leq\P(\mathcal{D})+\P(\mathcal{B})\\
+\left|\P(\gamma_{\delta^{-2}t}^{\delta}\in B(y,\eps),\mathcal{D}^{c},\mathcal{T})-\P(\gamma_{t}\in B(y,\eps),\mathcal{D}^{c},\mathcal{T})\right|.\label{eq: LCLT_meso_bound}
\end{multline}
We have $\P(\mathcal{D})\leq C\cdot\delta^{10}$ and $\P(\mathcal{B})\leq C\eps\delta^{\nu}\cdot t^{-1}$
provided that $\eps\gg\delta^{\nu},$ thus it remains to estimate
the last term. Let $\sigma$ denote the rotation of the cone around
its tip by $\pi/3$ (we assume here $\Omega$ is triangulated, the
other case is completely similar), and let $\hat{\sigma}$ be the
rotation, by the same angle, of the plane obtained as a quotient the
universal cover of the cone punctured at its tip; thus $\hat{\sigma}^{6}=\Id$
and $\sigma^{\frac{3\alpha}{\pi}}=\Id.$ We have $\gamma_{\delta^{-2}t}^{\delta}\in\cup_{k}\sigma^{k}\left(B(y,\eps)\right)$
if and only if $\hat{\gamma}_{\delta^{-2}t}^{\delta}\in\cup_{k}\hat{\sigma}^{k}\left(B(y,\eps)\right),$
and thus, as above, 
\begin{multline}
\left|\P(\gamma_{\delta^{-2}t}^{\delta}\in\cup_{k}\sigma^{k}(B(y,\eps),\mathcal{D}^{c},\mathcal{T})-\P(\gamma_{t}\in\cup_{k}\sigma^{k}(B(y,\eps)),\mathcal{D}^{c},\mathcal{T})\right|\\
\le\left|\P(\hat{\gamma}_{\delta^{-2}t}^{\delta}\in\cup_{k}\hat{\sigma}^{k}(B(y,\eps),\mathcal{D}^{c},\mathcal{T})-\P(\hat{\gamma}_{t}\in\cup_{k}\hat{\sigma}^{k}(B(y,\eps)),\mathcal{D}^{c},\mathcal{T})\right|\\
\leq\text{\ensuremath{\P}(\ensuremath{\mathcal{B}\cap\mathcal{T}})\ensuremath{\leq}}C\eps\delta^{\nu}t^{-1}\label{eq: LCLT_sym_nound}
\end{multline}
Now let $\tau:=\min\{t:\dist(\gamma_{t},0))<2\cdot\delta^{\nu}\}.$
On $\mathcal{T}$, to estimate the difference between probabilities
to arrive to $B(y,\eps)$ and to $\sigma(B(y,\eps))$, we use strong
Markov property with respect to $\tau.$ We have 
\begin{multline*}
\left|\P^{\gamma_{\delta^{-2}\tau}^{\delta}}\left[\gamma_{\delta^{-2}(t-\tau)}^{\delta}\in B(y,\eps)\right]-\P^{\gamma_{\delta^{-2}\tau}^{\delta}}\left[\gamma_{\delta^{-2}(t-\tau)}^{\delta}\in\sigma(B(y,\eps))\right]\right|\\
=\left|\P^{\gamma_{\delta^{-2}\tau}^{\delta}}\left[\gamma_{\delta^{-2}(t-\tau)}^{\delta}\in B(y,\eps)\right]-\P^{\sigma^{-1}(\gamma_{\delta^{-2}\tau}^{\delta})}\left[\gamma_{\delta^{-2}(t-\tau)}^{\delta}\in B(y,\eps)\right]\right|\\
\leq\sup_{x,z\in B(0,3\delta^{\nu}),\,t\leq T}\left|\P^{x}(\gamma_{\delta^{-2}t}^{\delta}\in B(y,\eps))-\P^{z}(\gamma_{\delta^{-2}t}^{\delta}\in B(y,\eps))\right|.
\end{multline*}
We estimate the expression in the supremum separately for $t>10\delta^{\nu}>(3\delta^{\nu})^{2}$
and for $t<10\delta^{\nu}.$ In the first case, we use Hölder continuity
(\ref{eq: Holder}) with $r=3\delta^{\frac{\nu}{2}}$, time shifted
by $t-r^{2}>0$, and bounding the oscillation by $1$; this gives
the bound of $\leq C\cdot\delta{}^{\theta\frac{\nu}{2}}.$ In the
second case, we use Lemma \ref{lem: Short-time-large-diameter} to
get the bound of $\leq C\delta^{3\nu}\ll\delta^{\theta\frac{\nu}{2}}.$
We infer that 
\[
\left|\frac{\pi}{3\alpha}\P(\gamma_{\delta^{-2}T}^{\delta}\in\cup_{k}\sigma^{k}(B(y,\eps)),\mathcal{D}^{c},\mathcal{T})-\P(\gamma_{\delta^{-2}T}^{\delta}\in B(y,\eps(\delta)),\mathcal{D}^{c},\mathcal{T})\right|\leq C\cdot\delta{}^{\frac{\theta\nu}{2}}.
\]
A similar estimate holds for the continuous heat kernel. Therefore,
we finally get
\begin{multline}
\left|\P(\gamma_{\delta^{-2}T}^{\delta}\in B(y,\eps))-\P(\gamma_{T}\in B(y,\eps))\right|\\
\leq C\delta{}^{10}+C\cdot\delta{}^{\frac{\theta\nu}{2}}+C\cdot\eps\delta^{\nu}\cdot t^{-1}\leq C\eps^{2}\delta^{\rho}\cdot\max\{t^{-1},1\}.\label{eq: mesosc_estimate}
\end{multline}
with $\rho>0$ if we choose $\eps(\delta):=\delta^{\mu}$ with $\mu$
small enough.

Now, for $\kappa>0,$ as in the proof of Lemma \ref{prop: HK_bound},
we can bound 
\[
\sup_{t\geq\frac{\delta^{2\kappa}}{2},z\in B(y,\eps)}\P(\gamma_{\delta^{-2}t}^{\delta}=z)\leq C\delta^{2-2\kappa}.
\]
 Using this to bound the oscillation in the Hölder bound (\ref{eq: Holder}),
we get, for $t\geq\delta^{2\kappa},$
\[
\left|\frac{1}{|B(y,\eps)|}\P(\gamma_{\delta^{-2}t}^{\delta}\in B(y,\eps))-\P(\gamma_{\delta^{-2}t}^{\delta}=y)\right|\leq\left(\frac{\eps}{\delta^{\kappa}}\right)^{\theta}\cdot\delta^{2-2\kappa}=\delta^{\mu\theta+2-(2+\theta)\kappa}=:\delta^{2+q},
\]
 provided that $\kappa$ is small enough (in which case also $\eps\ll\delta^{\nu}$
so that the Hölder bound applies). A similar estimate holds for the
continuous heat kernel. Combining this with (\ref{eq: mesosc_estimate})
gives the claim.
\end{proof}
\begin{rem}
\label{rem: vertex_at_tip} If there is a vertex at the tip of the
cone, then the coupling of $\gamma_{\delta^{-2}t}^{\delta}$ and $\hat{\gamma}_{\delta^{-2}t}^{\delta}$
fails after the moment the former hits the tip, in particular, the
distributions of the first time they leave the tip will be different.
However, forcing $\hat{\gamma}_{\delta^{-2}t}^{\delta}$ to leave
the tip simultaneously with $\gamma_{\delta^{-2}t}^{\delta}$ yields
a coupling of $\gamma_{\delta^{-2}t}^{\delta}$ to $\hat{\gamma}_{\delta^{-2}t+\tau}^{\delta}$
such that still $\gamma_{\delta^{-2}t}^{\delta}\in\cup_{k}\sigma^{k}\left(B(y,\eps)\right)$
if and only if $\hat{\gamma}_{\delta^{-2}t+\tau}^{\delta}\in\cup_{k}\hat{\sigma}^{k}\left(B(y,\eps)\right).$
Here $\tau$ a random variable given by a sum of $N$ i. i. d. contributions,
where $N$ is the number of visits to the tip. The expectation of
$N$, and hence that of $\tau$, is $O(\log(\delta^{-2}t)),$ and
therefore it will introduce a negligible error into the above computations.
\end{rem}

\begin{proof}[Proof of Lemma \ref{lem: HK_infinite_tail_bound}]
 As explained in in the proof of Lemma \ref{lem: LCLT}, it suffices
to consider the case of a cone. Moreover, since, in the notation of
that proof, $\gamma_{\delta^{-2}t}^{\delta}\in\cup_{k}\{\sigma^{k}(y)\}$
if and only if $\hat{\gamma}_{\delta^{-2}t}^{\delta}\in\cup_{k}\{\hat{\sigma}^{k}(y)\},$
it is in fact sufficient to consider the case of a plane, where it
is immediate from the local Central limit theorem. 
\end{proof}

\section{Appendix: the domain of the Dirichlet form $\protect\En^{\Omega,\varphi}$}

\label{sec: Appendix}For a unitary vector bundle $\varphi$, the
point-wise scalar product $\langle\cdot;\cdot\rangle$ induces the
scalar product on the set $L^{2,\varphi}(\Omega)$ be of $L^{2}$
sections of $\varphi$, given by $\int_{\Omega}\langle f(x);f(x)\rangle dx,$
and also on gradients of nice enough sections, using $\langle\nabla^{\varphi}f(x);\nabla^{\varphi}f(x)\rangle=\langle\pa_{\nu}\tilde{f}(x);\pa_{\nu}\tilde{f}(x)\rangle+\langle\pa_{\eta}\tilde{f}(x);\pa_{\eta}\tilde{f}(x)\rangle,$
where $(\nu,\eta)$ are isometric local coordinates, and we identify
$f$ with a function $\tilde{f}$ using a local trivialization of
$\varphi$. We thus define $H^{1,\varphi}(\Omega):=\{f\in L^{2,\varphi}(\Omega):\nabla^{\varphi}f\in L^{2}\},$
equipped with the scalar product $\langle f;f\rangle+\langle\nabla^{\varphi}f;\nabla^{\varphi}f\rangle.$
Put $H_{\pa}^{1,\varphi}(\Omega)=\{f\in H^{1,\varphi}(\Omega):\:f\equiv0\text{ on }\pa_{\Dir}\Omega\};$
the vanishing on the boundary can be understood in the sense of the
trace operator in the theory of Sobolev spaces, or, more elementarily,
by requiring that the extension of $f$ by $0$ across the boundary
has a finite $H^{1}$ norm. It is well known that for piece-wise smooth
boundaries, $H_{\pa}^{1,\varphi}(\Omega)$ is the closure in $H^{1,\varphi}(\Omega)$
of the set of $f\in H^{1,\varphi}(\Omega)$ supported away from $\pa_{\Dir}\Omega$
(cf. \cite[before Theorem 4.11]{sznitman1998brownian} or \cite{leoni2017first}.)
The next lemma describes explicitly the Dirichlet form $\mathcal{E}^{\Omega,\varphi}(f,f)=\lim_{t\searrow0}t^{-1}\langle(1-P_{t}^{\Omega,\varphi})f;f\rangle,$
defined in Section \ref{sec: heat_kernel_continuum}.
\begin{lem}
\label{lem: domain}We have $\mathcal{D}(\mathcal{E}^{\Omega,\varphi})=H_{\pa}^{1,\varphi}(\Omega)$
and $\En^{\Omega,\varphi}(f,f)=\langle\nabla^{\varphi}f,\nabla^{\varphi}f\rangle$
for any $f\in\mathcal{D}(\En^{\Omega,\varphi}).$
\end{lem}

In other words, the generator $\Delta^{\Omega,\varphi}$ is the \emph{Friedrichs}
\emph{extension} of the Laplacian acting on smooth sections of $\varphi$
compactly supported away from $\pa_{\Dir}\Omega$, corners, conical
singularities and punctures. Before giving the a proof, we remark
that this also follows from the general theory. By \cite[Theorem 7.2.2]{fukushima2010dirichlet},
we can \emph{define} $\gamma_{t}$ as the unique diffusion associated
to the Dirichlet form $\En_{BM}(f,f)=\langle\nabla f;\nabla f\rangle$
with the domain $\mathcal{D}(\En_{BM})=H_{\pa}^{1}(\Omega).$ Moreover,
if $\Lambda\subset\Omega$ is isometric to a domain in the plane,
then, up to hitting $\pa\Lambda\setminus\pa_{\Neu}\Omega$, the distribution
of $\gamma_{t}$ coincides with that of the Brownian motion in $\Lambda$
reflected at $\pa_{\Neu}\Omega$ \cite[Theorem 4.4.2 and Example 4.5.3]{fukushima2010dirichlet}.
Therefore, this construction of $\gamma_{t}$ coincides with the one
given above. This proves Lemma \ref{lem: domain} for trivial line
bundle; for the general case, observe that the multiplication by smooth
functions preserves both $\mathcal{D}(\mathcal{E}^{\Omega,\varphi})$
and $H_{\pa}^{1,\varphi}(\Omega),$ hence we can reduce the result
to the case of the trivial line bundle using a suitable partition
of unity.

\begin{proof}[Proof of Lemma \ref{lem: domain}.]
 Denote by $\sing$ the set of conical tips, corners, and punctures
of $\Omega.$ Let $\mathcal{C}_{\pa}^{\infty,\varphi}(\Omega)$ be
the set of all smooth sections of $\varphi$ that are supported away
from $\sing$ and $\pa_{\Dir}\Omega$, and which extend smoothly across
$\pa_{\Neu}\Omega$ by $\phi(\overline{x})=\phi(x),$ where $x\mapsto\overline{x}$
is the reflection in $\hat{\Omega}$. Using a trivialization of $\varphi,$
we can locally identify a section $\phi\in\mathcal{C}_{\pa}^{\infty,\varphi}(\Omega)$
with a function $\tilde{\phi}$; expanding the latter in a Taylor
series and using (\ref{eq: comparison1}\textendash \ref{eq: comparison2}),
it is easy to see that 
\[
D_{t}^{\Omega,\varphi}\phi\stackrel{t\to0}{\longrightarrow}\Delta^{\varphi}\phi=-\frac{1}{2}(\partial_{\eta}^{2}+\partial_{\nu}^{2})\tilde{\phi}(x_{\eta\nu})
\]
 uniformly and hence in $L^{2}(\Omega),$ where $(\eta,\nu)\mapsto x_{\eta\nu}$
is an isometric local chart. Therefore, we have $\mathcal{C}_{\pa}^{\infty,\varphi}(\Omega)\subset\mathcal{D}(\En^{\Omega,\varphi}),$
and $\En^{\Omega,\varphi}(\phi,\phi)=\langle\phi,\Delta^{\varphi}\phi\rangle=\langle\nabla^{\varphi}\phi,\nabla^{\varphi}\phi\rangle$
for all $\phi\in\mathcal{C}_{\pa}^{\infty,\varphi}(\Omega).$ Since
the form $\En^{\Omega,\varphi}$ is closed, this implies $H_{\pa}^{1,\varphi}(\Omega)\subset\mathcal{D}(\En^{\Omega,\varphi})$
and $\En^{\Omega,\varphi}(f,f)=\langle\nabla^{\varphi}f,\nabla^{\varphi}f\rangle,$
$f\in H_{\pa}^{1,\varphi}(\Omega),$ once we show that $\mathcal{C}_{\pa}^{\infty,\varphi}(\Omega)$
is dense in $H_{\pa}^{1,\varphi}(\Omega).$ 

The density is proven by a series of standard arguments. First, we
can approximate any $f\in H_{\pa}^{1,\varphi}(\Omega)$ supported
away from $\sing$ and $\pa_{\Dir}\Omega$ by $\phi\in\mathcal{C}_{\pa}^{\infty,\varphi}(\Omega),$
by extending $f$ by reflection across $\pa_{\Neu}\Omega$ and mollifying.
Second, any $f\in H_{\pa}^{1,\varphi}(\Omega)$ can be approximated
by bounded sections in $H_{\pa}^{1,\varphi}(\Omega),$ by truncation
at level lines (cf. \cite[Example 1.2.1]{fukushima2010dirichlet}):
let $g_{R}(v)=\frac{Rv}{\max\{R,|v|\}},$ $v\in\R^{d},$ then, approximating
$g_{R}$ with smooth functions and applying the chain rule, we see
that $\left\Vert f-g_{R}f\right\Vert _{H^{1}}^{2}\leq\int_{|f|>R}\langle\nabla^{\varphi}f,\nabla^{\varphi}f\rangle+\int_{|f|>R}\langle f,f\rangle\to0$
as $R\to\infty.$ Finally, to approximate a bounded $f\in H_{\pa}^{1,\varphi}(\Omega)$
by those compactly supported away from $\sing,$ let $z_{0}\in\sing$
and put $\phi_{\eps}(z)=\max\left\{ 0,\min\left\{ 1+\frac{\log|z-z_{0}|}{-\log\epsilon},1\right\} \right\} .$
Then, $\phi_{\eps}(z)\equiv0$ for $|z-z_{0}|<\eps,$ and
\[
\left\Vert f-\phi_{\eps}\cdot f\right\Vert _{H^{1}}^{2}=\left\Vert (1-\phi_{\eps})f\right\Vert _{L^{2}}^{2}+\left\Vert (1-\phi_{\eps})\nabla^{\varphi}f\right\Vert _{L^{2}}^{2}+\left\Vert f\nabla(1-\phi_{\eps})\right\Vert _{L^{2}}^{2}.
\]
Since we have $0\leq\phi_{\eps}(z)\nearrow1$ as $\eps\to0,$ the
first two terms above tend to zero as $\eps\to0$ by dominated convergence
theorem, and $\left\Vert f\nabla(1-\phi_{\eps})\right\Vert _{L^{2}}^{2}\leq\sup|f|^{2}\left\Vert \nabla(1-\phi_{\eps})\right\Vert _{L^{2}}^{2}\leq\frac{c}{-\log\eps}\to0.$

We turn to the proof of $\mathcal{D}(\En^{\Omega,\varphi})\subset H_{\pa}^{1,\varphi}(\Omega).$
For $f\in\mathcal{D}(\En^{\Omega,\varphi})$ and $\phi\in\mathcal{C}_{\pa}^{\infty,\varphi}(\Omega)$,
we have 
\[
\En^{\Omega,\varphi}(f,\phi)=\lim_{t\to0}\langle f,D_{t}^{\Omega,\varphi}\phi\rangle=\langle f,\Delta^{\varphi}\phi\rangle.
\]
On the other hand, by Cauchy-Schwarz, 
\[
\En^{\Omega,\varphi}(f,\phi)\leq\left(\En^{\Omega,\varphi}(f,f)\right)^{\frac{1}{2}}\left(\En^{\Omega,\varphi}(\phi,\phi)\right)^{\frac{1}{2}}=\left(\En^{\Omega,\varphi}(f,f)\right)^{\frac{1}{2}}\langle\nabla^{\varphi}\phi,\nabla^{\varphi}\phi\rangle^{\frac{1}{2}}.
\]
It follows that $\phi\mapsto\langle f,\Delta^{\varphi}\phi\rangle$
extends to a bounded linear functional on $H_{\pa}^{1,\varphi}(\Omega),$
and thus so does $\phi\mapsto\langle f,\Delta^{\varphi}\phi\rangle+\langle f,\phi\rangle.$
Therefore, by Riesz\textendash Markov, there exists $\tilde{f}\in H_{\pa}^{1,\varphi}(\Omega)$
such that $\langle f,\Delta^{\varphi}\phi\rangle+\langle f,\phi\rangle=\langle\nabla^{\varphi}\tilde{f,}\nabla^{\varphi}\phi\rangle+\langle\tilde{f,}\phi\rangle,$
$\phi\in\mathcal{C}_{\pa}^{\infty,\varphi}.$ If $\phi$ is in addition
compactly supported away from $\pa\Omega,$ the last expression is
equal to $\langle\tilde{f,}\Delta^{\varphi}\phi\rangle+\langle\tilde{f,}\phi\rangle.$
In other words, $h=f-\tilde{f}$ satisfies $\Delta^{\varphi}h+h\equiv0$
in the weak sense. But then, by elliptic regularity, $h$ is smooth
in the interior of $\Omega\setminus\sing$; also, since $H_{\pa}^{1,\varphi}(\Omega)\subset\mathcal{D}(\En^{\Omega,\varphi}),$
we have $h=f-\tilde{f}\in\mathcal{D}(\En^{\Omega,\varphi}).$ Thus,
it is enough to prove that if $h\in\mathcal{D}(\En^{\Omega,\varphi})$
and $h$ is smooth, then $h\in H_{\pa}^{1,\varphi}(\Omega).$

We follow \cite[Section 1.4]{sznitman1998brownian}. Write 
\[
(P_{t}^{\Omega,\varphi}f)(x)=\E^{x}\left[\varphi_{\gamma_{[0,t]}}^{-1}(f(\gamma_{t}))\right]=\E\E^{x}\left[\left.\varphi_{\gamma_{[0,t]}}^{-1}(f(\gamma_{t}))\right|\gamma_{t}\right]=\int_{\Omega}R_{x,y,t}^{\Omega,\varphi}f(y)P^{\Omega}(x,y,t)\,dy,
\]
where $R_{x,y,t}^{\Omega,\varphi}f(y)=\E^{x}\left[\left.\varphi_{\gamma_{[0,t]}}^{-1}(f(y))\right|\gamma_{t}=y\right].$
In the trivial bundle case, one simply has $R_{x,y,t}^{\Omega,\varphi}f(y)=f(y).$
Due to the reversibility of $\gamma$ and the unitarity of $\varphi$,
we have $\langle R_{x,y,t}^{\Omega,\varphi}f(y),g(x)\rangle=\langle f(y),R_{y,x,t}^{\Omega,\varphi}g(x)\rangle.$
Hence, we can write (cf. \cite[Lemma 4.8]{sznitman1998brownian})
\[
\langle D_{t}^{\Omega,\varphi}f;g\rangle=\En^{(1,t)}(f,g)+\En^{(2,t)}(f,g),
\]
where 
\begin{align*}
\En^{(1,t)}(f,g) & =\frac{1}{2t}\int_{\Omega\times\Omega}\langle f(x)-R_{x,y,t}^{\Omega,\varphi}f(y),g(x)-R_{x,y,t}^{\Omega,\varphi}g(y)\rangle P^{\Omega}(x,y,t)\,dxdy,\\
\En^{(2,t)}(f,g) & =\frac{1}{t}\int_{\Omega}(1-Q_{t}(x))f(x)g(x)\,dx.
\end{align*}
and $Q_{t}(x)=\int_{\Omega}P^{\Omega}(x,y,t)dy$ is the probability
that $\gamma_{t}$ started at $x$ did not stop by time $t$.

For a compact set $K\subset\Omega$ and $K_{\eps}$ its small neighborhood,
and $h\in\mathcal{D}(\En^{\Omega,\varphi})$ smooth, we have 
\begin{multline*}
\En^{\Omega,\varphi}(h,h)\geq\frac{1}{2t}\int_{K_{\eps}\times K}\langle h(y)-R_{x,y,t}^{\Omega,\varphi}h(x),h(y)-R_{x,y,t}^{\Omega,\varphi}h(x)\rangle P^{\Omega}(x,y,t)dxdy\\
\stackrel{t\to0}{\longrightarrow}\int_{K}\langle\nabla^{\varphi}h(x),\nabla^{\varphi}h(x)\rangle\,dx,
\end{multline*}
 as is readily seen by identifying $h$ with a function $\tilde{h}$
using a local trivialization near each $x$ and using Taylor approximation
of $\tilde{h}$ and (\ref{eq: comparison1}\textendash \ref{eq: comparison2}).
Taking a supremum over all such $K$, we conclude that $h\in H^{1,\varphi}(\Omega).$ 

Finally, we refine $\mathcal{D}(\En^{\Omega,\varphi})\subset H^{1,\varphi}(\Omega)$
to $\mathcal{D}(\En^{\Omega,\varphi})\subset H_{\pa}^{1,\varphi}(\Omega).$
We do the trivial line bundle case; the general one only differs by
heavier notation. Let $0\leq\phi\leq1$ be a smooth function compactly
supported in a Euclidean neighborhood of $x\in\pa_{\Dir}\Omega\setminus\sing$
and identically equal to $1$ in a smaller neighborhood. If $f\in\mathcal{D}(\En^{\Omega}),$
then we have $\En^{(2,t)}(\phi f,\phi f)\leq\En^{(2,t)}(f,f)$ and
\begin{multline*}
\En^{(1,t)}(\phi f,\phi f)=\frac{1}{2t}\int_{\Omega\times\Omega}(f(x)-f(y))^{2}\phi(x)^{2}P^{\Omega}(x,y,t)dxdy\\
+\int_{\Omega}\left(\int_{\Omega}\frac{1}{2t}(\phi(x)-\phi(y))^{2}P^{\Omega}(x,y,t)dy\right)f(x)^{2}dx,
\end{multline*}
and both terms remain bounded as $t\to0.$ Hence, also $\phi f\in\mathcal{D}(\En^{\Omega}).$
Identifying $\supp\phi f$ with a subdomain of the upper half-plane
$\H$ equipped with Dirichlet boundary conditions, due to (\ref{eq: comparison2}),
we have $\frac{1}{t}(P^{\Omega}(x,y,t)-P^{\H}(x,y,t))\to0$ and $\frac{1}{t}(Q_{t}^{\Omega}-Q_{t}^{\H})\to0$.
We conclude that $\phi f\in\mathcal{D}(\En^{\H}),$ where $\En^{\H}(f,f)=\lim_{t\to0}t^{-1}\langle D_{t}^{\H}f,f\rangle.$
If $g$ denotes $\phi f$ extended to the lower half-plane by $g(\bar{z})=-g(z),$
then we have $P_{t}^{\H}(\phi f)\equiv P_{t}^{\C}g$ in $\H,$ and
$\En^{\H}(\phi f,\phi f)=\frac{1}{2}\En^{\H}(g,g),$ so that $g\in\mathcal{D}(\En^{\C}).$
But it is easy to see, using Fourier transform, that $\mathcal{D}(\En^{\C})=H^{1}(\C)$
\cite[Example 4.1]{sznitman1998brownian}. If $g(\bar{z})=-g(z),$
then we can only have $g\in H^{1}(\C)$ if $g\equiv0$ a. e. on $\R.$
We conclude that $f\equiv0$ a. e. on $\pa_{\Dir}\Omega.$ 
\end{proof}
\bibliographystyle{plain}
\bibliography{Mixedcorr}

\end{document}